\documentclass[11pt]{amsart}
\usepackage{blindtext}
\usepackage{geometry}
\usepackage{amssymb}
\usepackage{pgfplots}
\pgfplotsset{compat=1.15}
\usepackage{mathrsfs}
\usetikzlibrary{arrows}
\usepackage{tikz}
\usepackage{tikz-cd}
\usepackage{graphicx} 
\usepackage{dcolumn}
\usepackage{bm}

\usepackage[utf8]{inputenc}
\usepackage[T1]{fontenc}
\usepackage{mathptmx}
\usepackage{etoolbox}
\usepackage{subcaption}
\usepackage{stmaryrd}
\makeatletter

\usepackage[english]{babel}  
\usepackage{mathtools}
\usepackage{amsthm}
\usepackage{amsmath}
\usepackage{xcolor}
\usepackage{comment}
\usepackage[colorlinks=true,linkcolor=black,anchorcolor=black,citecolor=black,filecolor=black,menucolor=black,runcolor=black,urlcolor=black]{hyperref} 

\newtheorem{definition}{Definition}[section]
\newtheorem{theorem}{Theorem}[section]
\newtheorem*{remark}{Remark}
\newtheorem{coro}[theorem]{Corollary}
\newtheorem{prop}{Proposition}[section]
\newtheorem{lemma}{Lemma}[section]

\newcommand{\R}{\mathbb{R}}
\newcommand{\Z}{\mathbb{Z}}
\newcommand{\N}{\mathbb{N}}
\newcommand{\C}{\mathbb{C}}
\newcommand{\Q}{\mathbb{Q}}

\newcommand{\cut}{\kappa}
\newcommand{\ch}{\text{ch}}
\newcommand{\chb}{\text{ch}^{(b)}}
\newcommand{\che}{\text{ch}^{(e)}}
\newcommand{\chab}{\text{ch}^{(ab)}}

\newcommand{\SquareMarker}[2]{%
  \pgfmathanglebetweenpoints{\pgfpointorigin}{\pgfpoint{#1}{#2}}%
  \edef\markerangle{\pgfmathresult}%
  \draw[shift={(#1,#2)}, rotate=\markerangle, fill=white, draw=none] 
       (-1pt,-1pt) rectangle (1pt,1pt);
}

\newcommand{\Ss}{\mathcal{S}}
\newcommand{\p}{p}
\newcommand{\q}{q}
\renewcommand{\b}{\psi}

\makeatother
\begin{document}

\title[Bulk Edge Correspondence for aperiodic tight binding operators]{Augmentation and Bulk Edge Correspondence for one dimensional aperiodic tight binding operators}
\author{Johannes Kellendonk, Lorenzo Scaglione}
\address{Institut Camille Jordan, Université Claude Bernard Lyon 1, 43 boulevard du 11 novembre 1918, F-69622 Villeurbanne Cedex}

\date{\today}

\begin{abstract}
	We consider a particular class of 1D aperiodic models with the aim to understand how their internal degrees of freedom contribute to their topological invariants and the possible relations (correspondences) among them. In order to handle models with finite local complexity we introduce the principle of augmentation. This allows us to relate the values of the Integrated Density of States at gap energies for the bulk system to spectral flows. We consider two different augmentations. The first is based on the mapping torus construction. It leads to an alternative proof of the result that the gap labelling group of Bellissard coincides with that of Johnson-Moser. It furthermore allows for an interpretation of the spectral flow via boundary forces. The second augmentation applies to models obtained by the cut and project method where we find for 2-cut models two different spectral flows, one attached to the edge modes and related to the phason motion whereas the other is an augmented bulk invariant.  
	Our approach is based on the well-established $C^*$-algebraic approach to solid state physics and the description of topological invariants by $K$-theory and cyclic cocycles. We also present numerical simulations to illustrate our theorems.
\end{abstract}
\maketitle

\section{Introduction}
Aperiodic topological insulators have a rich structure of topological invariants. For quasi-periodic insulators this can be traced back to their internal degrees of freedom\footnote{such degrees of freedom are sometimes also referred to as synthetic degrees of freedom, because they can be designed in metamaterials}  
which come about as these materials can be described as cuts in a higher dimensional periodic structure.  
We consider here one-dimensional aperiodic models with the aim to understand how their internal degrees of freedom contribute to their topological invariants and how these invariants are related. This question has been looked at, both theoretically and experimentally \cite{kraus2012topological,verbin2015topological,dana2014topologically,baboux2017measuring,jagannathan2025missing}, mostly by relating quasi-periodic chains to a two-dimensional system with a magnetic field, the tight binding approximation of the Integer Quantum Hall Effect. Such a relation is possible only for what we call here 1-cut models, a name we will explain below.  
Extending the work of \cite{kellendonk_prodan}, we develop here another approach to which we refer to as augmentation. This is an alternative way to overcome the difficulty that for chains based on the cut and project method, the edge states do not fill the gap. It has a wider range of applicability and brings in new features for 2-cut models, for instance.
 
Our approach is based on the by now well-established $C^*$-algebraic approach to solid state physics and the description of topological invariants by $K$-theory. A key feature of this approach is that materials are not described by a single Hamiltonian, but by a whole family of Hamiltonians, parametrized by the points of a topological space which can be derived from the geometric structure of the material, called its hull. The topological invariants of the material are related to the topology of this hull.   
 Bellissard realized that this approach \cite{Bellissard1992} can be conveniently described by a $C^*$-algebra together with a family of representations, for each point of the hull one, in such a way that the individual Hamiltonians of the family are the representatives of a single self-adjoint element of the $C^*$-algebra. Topological invariants can be derived from the $K$-groups of this algebra and correspondences between topological invariants are described by the boundary maps arising from exact sequences of algebras which relate different systems, as, for instance, the bulk and the edge of a material \cite{KelRichBal}. 

While the framework which we will set up below allows for far more general applications, let us describe our main model. 
Consider three real numbers $\theta$, $\phi$, $\cut$. Taken modulo $\Z$ we interprete $\theta$ and $\phi$ as angles on $\mathbb S^1=\R/\Z$ and call 
$\theta$ the rotation angle and $\phi$ the parameter angle. $\cut$ is supposed to lie in $(0,1)$ and called the cut value.

The family of Hamiltonians we study are defined on $\ell^2(\Z)$ and there given by 
\begin{equation}\label{eq-Koh} 
	H_\phi \psi(n) = \psi(n+1)+\psi(n-1) + V_\phi(n)\psi(n),\quad V_\phi(n)  = \left\{\begin{array}{ll}
	1 & \mbox{if } 0 < \{n\theta +\phi\} \leq \cut \\
	0 & \mbox{otherwise} \end{array}
\right.
\end{equation}
where 
$\{x\}$ denotes the fractional part of $x\in\R$. Note that translating the potential amounts to rotating $\phi$ by $\theta$, $V_\phi(n+1) = V_{\phi+\theta}(n)$ and so the rotation angle is related to translation in space. If we vary $\phi$ so that $n\theta +\phi$ crosses the boundary of the segment $[0,\cut]$ then the value $V_\phi(n)$ of the potential at site $n$ jumps from $0$ to $1$ or from $1$ to $0$ depending on whether we cross $0$ or $\cut$, respectively. Indeed, the orbits $\{n\theta|n\in \Z\}$ and $\{n\theta+\cut|n\in \Z\}$ 
of $0$ and $\cut$ under the rotation by $\theta$ are the discontinuity points of the function $\mathbb S^1\ni \phi \mapsto V_\phi\in \mathcal B(\ell^2(\Z))$ and they will be referred to as the cut points. In general the two orbits of cut points are different, and so we speak of a 2-cut model. But
if $\cut=\{m\theta\}$ for some $m\in\Z$--we simply say that $\kappa$ is a multiple of $\theta$--then the two orbits coincide and we speak of a 1-cut model. The discontinuous jumps in the potential values happen then in a single chain at two different sites; we call that a flip, in quasicrystal physics this is called a phason flip. 

If $\cut=\{\theta\}$ and this is an irrational number, then the model is the perhaps simplest model for a tight binding operator on a quasicrystalline structure. It has been studied extensively since it was proposed by Kohmoto, with recent work also focusing its topological invariants \cite{kraus2012topological,verbin2015topological,dana2014topologically,baboux2017measuring,jagannathan2025missing,kellendonk_prodan}. Bellissard and his co-authors have studied the model for irrational $\theta$ but with $\cut$ not a multiple of $\theta$. They were interested in the gap labelling. Gap labels are numerical topological invariants associated to the gaps in the spectrum of a Hamiltonian. Based on $K$-theory, Bellissard has defined the gap labelling group of a material. It is a subgroup of $\R$ which has the property that, whenever $E$ is an energy in a gap of the spectrum of the Hamiltonian
then the value of the integrated density of states (IDS) at that energy must belong to that group. Applied to the above operator $H_\phi$ Bellissard  \cite{bellissard1992gap} found that the IDS at an energy in a gap has the form 
\begin{equation}\label{eq-IDS-cut}
	\mathrm{IDS} = N + n_1\theta +n_2\cut
\end{equation}
where $N,n_1,n_2$ are integers, and this independently of the value of $\phi$. As for simple tight binding operators the IDS lies between $0$ and $1$, this means that the pair $(n_1,n_2)$ can be used to label the gap, provided $\theta$ and $\kappa$ are irrational and independent over $\Q$ (if this is not the case then there are ambiguities). An interesting question to ask is whether, given $(n_1,n_2)$, there is a gap in the spectrum which is labelled by these numbers. This need not be the case, but \cite{Bellissard_Iochum_Scoppola_Testard} could show that, if one multiplies the potential with a large enough constant then $H_\phi$ has at least some gaps with labels $(n_1,n_2)$ where $n_2\neq 0$. We will show below that, for $H_\phi$, $n_2$ can only be $0$ or $1$.
Recently \cite{band2024drymartiniproblemsturmian} has shown that if $\cut=\theta$ (so that the second label $n_2$ can be absorbed into the first) then for any choice of $n_1$ there is a gap with that label.

Our question here is: Do $n_1$ and $n_2$ have another physical interpretation? Can we obtain them as a spectral flow? Guided by what is already known, we expect this to be related to a bulk edge correspondence (BEC). Indeed, this approach was fruitful in the interpretation of $n_1$ given in \cite{kellendonk_prodan} for the case $\cut$ a multiple of $\theta$. To achieve this, the standard setup of the BEC had to be modified, notably by smoothening out the discontinuities of the function $\phi\mapsto V_\phi$. This can be done by interpolating the potentials in the following way.
For $\phi\in \mathbb S^1$ and $t\in [0,1]$ let
$$V_{\phi, t}(n) = \left\{\begin{array}{ll}
V_\phi(n) & \mbox{if } n\theta+\phi \neq 0,\cut \\
t & \mbox{if } n\theta+\phi = 0 \\
1-t & \mbox{if } n\theta+\phi =\cut
\end{array}
\right.$$
Figures~\ref{V_0} and \ref{V_0b} show examples for rational $\theta$, the first with $\cut$ not a multiple of $\theta$, the second with $\cut$ a multiple of $\theta$. It shows the values of this potential at some sites. At $t=0$ the potential  $V_{0,t}$ coincides with $V_0$ and 
as $t$ goes to $1$ the potential  $V_{0,t}$ goes to $V_\phi$ for some small $\phi>0$ making the above described jump (or phason flip) continuous. 
\begin{figure*}
	\centering
	\begin{subfigure}[t]{0.45\textwidth}
		\centering
		\begin{tikzpicture}[line cap=round,line join=round,x=1.0cm,y=1.0cm]
			\clip(-2.6,-0.6) rectangle (1.5,1.5);
			\draw [domain=-2.5:2.5] plot(\x,{(-0-0*\x)/10});
			\begin{scriptsize}
				\draw[loosely dotted] (-2.3,1) -- (2.5,1);
				\draw[color=black] (-2.5,1) node {1};
				\fill [color=black] (0,0) circle (1.5pt);
				\draw (0,0) -- (0,0);
				\draw[color=black] (0,0.2) node {$V_0(0)$};
				\draw[color=black] (0.09,-0.2) node {0};
				\fill [color=black] (1,0) circle (1.5pt);
				\draw (1,0) -- (1,1);
				\draw[color=black] (1,1.2) node {$V_0(1)$};
				\draw[color=black] (1.09,-0.2) node {1};
				\fill [color=black] (2,0) circle (1.5pt);
				\draw[color=black] (2.09,-0.2) node {2};
				\fill [color=black] (3,0) circle (1.5pt);
				\draw[color=black] (3.09,-0.2) node {3};
				\fill [color=black] (-1,0) circle (1.5pt);
				\draw[color=black] (-1,0.2) node {$V_0(-1)$};
				\draw[color=black] (-0.88,-0.2) node {-1};
				\fill [color=black] (-2,0) circle (1.5pt);
				\draw (-2,0) -- (-2,1);
				\draw[color=black] (-2,1.2) node {$V_0(-2)$};
				\draw[color=black] (-1.88,-0.2) node {-2};
			\end{scriptsize}
		\end{tikzpicture}
		\subcaption{Values of $V_0$}
	\end{subfigure}%
	~ 
	\begin{subfigure}[t]{0.45\textwidth}
		\centering
		\begin{tikzpicture}[line cap=round,line join=round,x=1.0cm,y=1.0cm]
			\clip(-2.6,-0.6) rectangle (1.5,1.5);
			\draw [domain=-2.5:2.5] plot(\x,{(-0-0*\x)/10});
			\begin{scriptsize}
				\draw[loosely dotted] (-2.3,1) -- (2.5,1);
				\draw[color=black] (-2.5,1) node {1};
				\fill [color=black] (0,0) circle (1.5pt);
				\draw (0,0) -- (0,0.5);
				\draw[color=black] (0,0.7) node {$V_{(0,t)}(0)=t$};
				\draw[color=black] (0.09,-0.2) node {0};
				\fill [color=black] (1,0) circle (1.5pt);
				\draw (1,0) -- (1,1);
				\draw[color=black] (1,1.2) node {$V_{(0,t)}(1)$};
				\draw[color=black] (1.09,-0.2) node {1};
				\fill [color=black] (2,0) circle (1.5pt);
				\draw[color=black] (2.09,-0.2) node {2};
				\fill [color=black] (3,0) circle (1.5pt);
				\draw[color=black] (3.09,-0.2) node {3};
				\fill [color=black] (-1,0) circle (1.5pt);
				\draw[color=black] (-1,0.2) node {$V_{(0,t)}(-1)$};
				\draw[color=black] (-0.88,-0.2) node {-1};
				\fill [color=black] (-2,0) circle (1.5pt);
				\draw (-2,0) -- (-2,1);
				\draw[color=black] (-2,1.2) node {$V_{(0,t)}(-2)$};
				\draw[color=black] (-1.88,-0.2) node {-2};
			\end{scriptsize}
		\end{tikzpicture}
		\subcaption{Values of $V_{(0,t)}$}
	\end{subfigure}

	\caption{Values of the potential for $\phi=0$, $\theta=0.2$ and $\cut=0.7$, $t=0.5$.}\label{V_0}
\end{figure*}
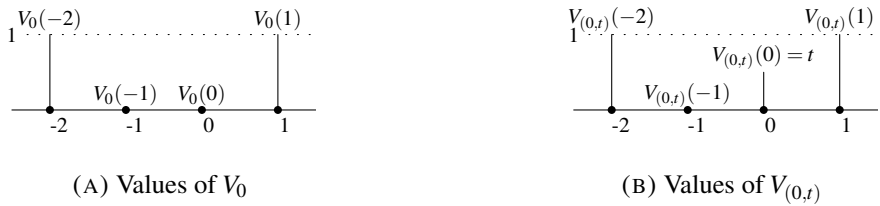
\begin{figure*}
	\centering
	\begin{subfigure}[t]{0.45\textwidth}
		\centering
		\begin{tikzpicture}[line cap=round,line join=round,x=1.0cm,y=1.0cm]
			\clip(-2.6,-0.6) rectangle (3.5,1.5);
			\draw [domain=-2.5:3.5] plot(\x,{(-0-0*\x)/10});
			\begin{scriptsize}
				\draw[loosely dotted] (-2.3,1) -- (3.5,1);
				\draw[color=black] (-2.5,1) node {1};
				\fill [color=black] (0,0) circle (1.5pt);
				\draw (1,0) -- (1,1);
				\draw[color=black] (0,0.2) node {$V_0(0)$};
				\draw[color=black] (0.09,-0.2) node {0};
				\fill [color=black] (1,0) circle (1.5pt);
				\draw[color=black] (1,1.2) node {$V_0(1)$};
				\draw[color=black] (1.09,-0.2) node {1};
				\fill [color=black] (2,0) circle (1.5pt);
				\draw[color=black] (2,0.2) node {$V_0(2)$};
				\draw[color=black] (2.09,-0.2) node {2};
				\fill [color=black] (3,0) circle (1.5pt);
				\draw (3,0) -- (3,1);
				\draw[color=black] (3,1.2) node {$V_0(3)$};
				\draw[color=black] (3.09,-0.2) node {3};
				\fill [color=black] (-1,0) circle (1.5pt);
				\draw[color=black] (-1,0.2) node {$V_0(-1)$};
				\draw[color=black] (-0.88,-0.2) node {-1};
				\fill [color=black] (-2,0) circle (1.5pt);
				\draw (-2,0) -- (-2,1);
				\draw[color=black] (-2,1.2) node {$V_0(-2)$};
				\draw[color=black] (-1.88,-0.2) node {-2};
			\end{scriptsize}
		\end{tikzpicture}
		\subcaption{Values of $V_0$}
	\end{subfigure}%
	~ 
	\begin{subfigure}[t]{0.45\textwidth}
		\centering
		\begin{tikzpicture}[line cap=round,line join=round,x=1.0cm,y=1.0cm]
			\clip(-2.6,-0.6) rectangle (3.5,1.5);
			\draw [domain=-2.5:3.5] plot(\x,{(-0-0*\x)/10});
			\begin{scriptsize}
				\draw[loosely dotted] (-2.3,1) -- (3.5,1);
				\draw[color=black] (-2.5,1) node {1};
				\fill [color=black] (0,0) circle (1.5pt);
				\draw (0,0) -- (0,0.3);
				\draw[color=black] (0,0.5) node {$V_{(0,t)}(0)=t$};
				\draw[color=black] (0.09,-0.2) node {0};
				\fill [color=black] (1,0) circle (1.5pt);
				\draw (1,0) -- (1,0.7);
				\draw[color=black] (1,0.8) node {$V_{(0,t)}(1)=1-t$};
				\draw[color=black] (1.09,-0.2) node {1};
				\fill [color=black] (2,0) circle (1.5pt);
				\draw[color=black] (2,0.2) node {$V_{(0,t)}(2)$};
				\draw[color=black] (2.09,-0.2) node {2};
				\fill [color=black] (3,0) circle (1.5pt);
				\draw (3,0) -- (3,1);
				\draw[color=black] (3,1.2) node {$V_{(0,t)}(3)$};
				\draw[color=black] (3.09,-0.2) node {3};
				\fill [color=black] (-1,0) circle (1.5pt);
				\draw[color=black] (-1,0.2) node {$V_{(0,t)}(-1)$};
				\draw[color=black] (-0.88,-0.2) node {-1};
				\fill [color=black] (-2,0) circle (1.5pt);
				\draw (-2,0) -- (-2,1);
				\draw[color=black] (-2,1.2) node {$V_{(0,t)}(-2)$};
				\draw[color=black] (-1.88,-0.2) node {-2};
			\end{scriptsize}
		\end{tikzpicture}
		\subcaption{Values of $V_{(0,t)}$}
	\end{subfigure}
	
	\caption{Values of the potential for  $\phi=0$, $\theta=\cut=0.4$, $t=0.3$.}\label{V_0b}
\end{figure*}
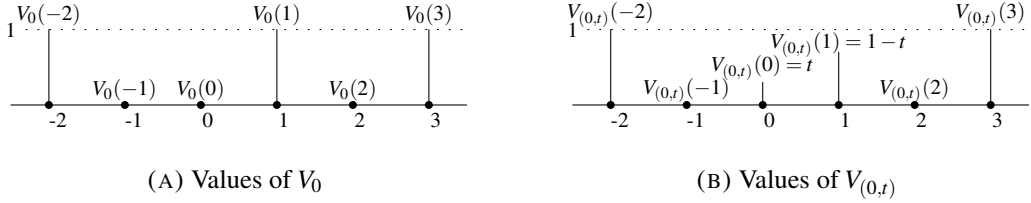
This procedure enlarges our family of Hamiltonians. The enlarged (or augmented) family is given by the Hamiltionians  
$H_{\phi,t}$ which are defined as  
$ H_{\phi}$ but with $V_\phi$ replaced by $V_{\phi,t}$.
We will show that we can do this interpolation on the level of hulls and algebras (we call that augmentation). 
Technically we work with two extensions of algebras, one defined by the augmentation of the bulk algebra and the other by the Toeplitz extension of the augmented bulk algebra. The resulting boundary maps in $K$-theory allow us identify $n_1$ and $n_2$ with topological invariants of different systems. The first, $n_1$ is essentially the same as in \cite{kellendonk_prodan}. It is associated to the edge system of the augmented system. If $n_2$ vanishes then $n_1$ is given by
\begin{equation}\label{eq-n1}
	n_1=\frac1{|\Delta|} \sum_{n\in\Z} \int_0^1 {\mathrm{Tr_{\ell^2(\Z^{\geq 0})}}}\big( 
	P_{\Delta}(\hat H_{\cut+n\theta,t})\partial_t\hat H_{\cut+n\theta,t}
	+P_{\Delta}(\hat H_{n\theta,t})\partial_t \hat H_{n\theta,t}\big)
	\text{d}t,
\end{equation}
where $\hat H_{\phi,t}$ is the restriction of $H_{\phi,t}$ to $\ell^2(\Z^{\geq 0})$. This restriction to half space creates edge states with eigenvalues  in the gap $\Delta$ of spectrum of the family of augmented bulk operators. 
$P_\Delta$ is the projector onto the corresponding eigenspaces.
The quantity $n_1|\Delta|$ has the interpretation of the work which the system exhibits on the edge states when we vary the potential as in Figure~\ref{V_0} to move through all jumps (or phason flips, Figure~\ref{V_0b}, if $\kappa$ is a multiple of $\theta$). The work is finite, as only the jumps near the edge have a significant effect on the edge states. Moreover, we can visualise $n_1$
as the spectral flow through the gap of the eigenvalues of edge states of $\hat H_{\phi,t}$ when $\phi$ varies through $\mathbb S^1$ and $t$ through the added intervals corresponding to the interpolation, see Figure~\ref{primary_gaps_WN}. 

If $\cut=\theta$ (we have a 1-cut model) then $n_2$ can be absorbed into $n_1$ and this situation has been discussed in \cite{kellendonk_prodan}. The Hamiltonian $H_\phi$ is then related to the Harper Hamiltonian $H^{Har}_\phi$ which one obtains if one replaces its potential by $V^{Har}_\phi(n) = 2\cos(n\theta+\phi)$. But this relation is subtle. While, 
$V_\phi$ can be approximated arbitrarily well by continuous deformations of the Harper potential \cite{kraus2012topological}, such an approximation cannot be uniform in $n$ and indeed, $H_\phi$ cannot be approximated by a continuously deformed Harper model in the norm topology. However, as was explained in \cite{kellendonk2024topological}, whenever  $H_\phi$ and $H^{Har}_\phi$ have the same IDS at the Fermi energy and this energy lies in that gap, they belong to the same topological phase (the same $K$-theory class). Finally, as the Harper Hamiltonian and the Hofstadter Hamiltonian arise through different representations of the same algebra element, their topological invariants are the same. This gives an explanation of most of the topological phenomena discussed in \cite{kraus2012topological,verbin2015topological,dana2014topologically,baboux2017measuring} from the point of view of Bellissard's $C^*$-algebraic approach to solid state systems, see also \cite{prodan2015virtual} for more details. 
But this is not what we will do in this work.

The second topological number $n_2$, which arises only if $\cut$ is not a multiple of $\theta$ (2-cut model), is not attached to edge states, but to the spectrum of the augmented bulk operator. It is thus a bulk invariant. 
In our numerical simulation, which requires rational values for $\theta$ and $\kappa$, we can see $n_2$ as a spectral flow density, that is, a spectral flow per unit length. We can also interpret it as a kind of work which the system performs on the bulk modes in the gap which are created through the interpolation. But contrarily to the above work exhibited on the edge states, we must count it with a relative sign, depending on whether $\phi$ belongs to the orbit of $\kappa$ or of $0$. The physical significance is therefore less clear. 

Remains the question what is $n_1$ if $n_2\neq 0$. This question is related to the group operation in $K$-theory which, in its physical interpretation, corresponds to the stacking of systems. We can make use of that here to make $n_1$ also visible as a spectral flow when $n_2\neq 0$. The idea is the following. A tight binding Hamiltonian without kinetic term represents an insulator which usually considered to be topologically trivial. But its IDS need not be $0$. So we take the direct sum of $H_\phi$ with a Hamiltonian which is a pure potential, namely $-V_\phi$, shifted in energy and rescaled in such a way that its central gap (the potential takes only two values) contains the gap $\Delta$ of $H_\phi$. It then turns out  that the direct sum Hamiltonian has an IDS at the gap of the form $N'+n_1\theta$, so the $n_1$-value stays the same whereas $n_2$ becomes zero. The trick is now to couple the two Hamiltonians with a small off-diagonal interaction term. Due to this interaction the interpolation procedure does not completely fill the gap so that we can compute $n_1$ by formula \eqref{eq-n1} where the Hamiltonian is now the augmentation of the coupled direct sum Hamiltonian, restricted to half space. This is an example of a physical realisation of Grothendieck's construction in $K$-theory.

We furthermore analyse in some detail a second augmentation of the bulk algebra which is based on the mapping torus construction. It allows to describe a system with a moving edge. Using this we find that there is an integer $N$ such that 
\begin{equation}\label{eq-IDS-shift}
	\mathrm{IDS} = N + \nu
\end{equation}
where $\nu$ is the spectral flow of the energies of the eigenvalues of the edge states in the gap when one moves the edge by one unit length. $\nu$ is the average work exhibited on the electrons by that move of the edge. Here the average is over all possible positions of the edge in the material. This average is thus topologically quantized.  
This result is comparable to \cite{dana1985quantised,kunz1986quantized}. 
It is the tight binding analog of the analogous result proven in \cite{ZoisKellendonk2005rotation} for Schr\"odinger operators in which the kinetic part is given by the Laplacian. In latter case we have the stronger result that $\mathrm{IDS} = \nu$ with the same interpretation of $\nu$. The need to allow for a possibly non-zero $N$ in the context of tight binding operators is related to an ambiguity in the extension problem. 

We can construct the augmented Hamiltonians via interpolation of the potential, somewhat similar to what we did above, and this leads to the physical interpretation of the r.h.s.\ of \eqref{eq-IDS-shift} through a boundary force. But from the point of view of $K$-theory other choices are possible and there is one, more mathematically then physically motivated, which leads to the equation $\mathrm{IDS} = \nu$. This choice is of theoretical interest, as this equation  says that the $K$-theoretical gap labelling group of Bellissard \cite{Bellissard} is equal to the gap labelling group defined by Johnson and Moser \cite{johnson1982rotation,johnson1986exponential}. An abstract argument showing the equivalence of these two versions of the gap-labelling was sketched \cite{bellissard1992gap}.

Combining (\ref{eq-IDS-cut}) and (\ref{eq-IDS-shift}) we obtain a relation between different spectral flows and these have different physical interpretation. For example, if the Fermi energy lies in a primary gap ($n_2=0$) then, modulo an integer, the average work exhibited on the edge states of the gap by moving the edge by one unit of length is $\theta$ times the work exhibited by the phason flips.
\bigskip

The question of how the above topological numbers $n_1$, $n_2$, $\nu$ can be visualized in a numerical simulation forces us to consider 
the case in which the rotation angle $\theta$ is rational. It is important to realize that, following the philosophy of the $C^*$-algebraic approach, rationally approximating the aperiodic system defined by irrational $\theta$
means that we approximate the hull with its $\Z$-action and not only one Hamiltonian. To explain the difference, if $\theta=\frac{p}{q}$ then the Hamiltonians $H_\phi$ become $q$-periodic. But in contrast to the aperiodic case, the spectrum of $H_\phi$ is no longer independent of $\phi$. 
Therefore, the joint spectrum 
$\bigcup_{\phi} \mathrm{spec}(H_\phi)$ is larger than the individual spectra $\mathrm{spec}(H_\phi)$. As Bellissard's gap labelling applies to the gaps of the joint spectrum the gap labelling group is richer than one would naively expect taking into account only the periodicity of the model. Indeed, the joint spectrum has generically $2q-1$ inner gaps. 
If $n_2=0$ so that there is no spectral flow of the augmented bulk spectrum then the spectral flow of the edge states is nicely visible. We call these gaps primary and the others secondary.
If $n_2\neq 0$ (secondary gap) the augmented bulk spectrum fills the gap and it is not possible to observe the spectral flow of the edge states. To see $n_1$ as the spectral flow of edge states properly, we need to couple the system to one which is an atomic limit as described above. 
In this way we can define and measure the $n_1$-value of (the family) $H_\phi$ also in secondary gaps.
Our simulations confirm (\ref{eq-IDS-cut}) and (\ref{eq-IDS-shift}). Here we have to take into account that the IDS has to be averaged over $\phi$. 

Note that in the rational case (\ref{eq-IDS-cut}) and (\ref{eq-IDS-shift}) are Diophantine equations which do not allow to solve for $n_1$ or $\nu$ in a unique way. Indeed, $n_1$ and $\nu$ are only determined modulo $q$ by the IDS. This so-called $q$-ambiguity is already visible on the level of $K$-theory.

\section{The $C^*$-algebraic approach to solids}
To study the topological invariants of aperiodic materials in the one-particle picture one should not look at a single Hamiltonian but consider a so-called covariant family of Hamiltonians which all have the same underlying spatial structure. This family can be described as a particular set of representatives of a single self-adjoint element $h$ (the abstract Hamiltonian) of a $C^*$-algebra. This algebra is the groupoid $C^*$-algebra of the topological groupoid defined by the spatial structure of the material (the tiling or pattern groupoid). When $h$ has a gap in its spectrum then its spectral projection onto the states of energy below the gap defines a $K_0$-class of the $C^*$-algebra. In the one-particle approximation to solid state systems, this class determines the topological phase of the material.
If the material is considered as infinitely extended and without boundary, then $h$ describes the physics in the interior of the material, its bulk, and the above algebra is referred to as the bulk algebra.
One way to describe the boundary of a material is to confine it to a half-space. Algebraically this can be described by means of an extension of the bulk algebra. The corresponding ideal will then contain the operators which are in some sense localized at the edge (boundary). 
The bulk-edge correspondence (BEC) in its standard form relates topological invariants of the bulk algebra with topological invariants of the ideal. A detailed account of this standard form of the BEC can be found in \cite{prodan2016bulk}, for its adaptation to aperiodic systems we also refer the reader to  \cite{kellendonk_prodan,kellendonk2024topological}.

\subsection{Algebra of short range pattern equivariant operators}
Let us describe the bulk algebra for one dimensional materials with finite local complexity. These have spatial structure which may be described by an infinite chain of symbols which encode the atomic types. We denote by $\mathfrak a$ the set of symbols 
and suppose that there are only finitely many. Thus the spatial structure of the material is given by a two-sided symbolic sequence 
$\omega_0\in \mathfrak a^{\Z}$. A tight binding operator describing the motion of the electrons in the material is an operator on  $\ell^2(\Z)$ of the form 
$$H\psi(n)=\sum_{m\in\Z} H_{nm}\psi(m)$$
where the kernel (or matrix) $(H_{nm})_{nm}$, $H_{nm}\in\C$, satisfies two properties: is it pattern equivariant and of short range. To explain these notions, let 
$$B_R\left[\omega_0\right]=\omega_0(-R) \cdots \omega_0(R) $$
be the word of radius $R\in\N$, that is, size $2R+1$ centered at $0$. We say that a function $f : \Z \rightarrow \C$ is pattern equivariant with range $R$ if for all $n,m \in \Z$
$$B_R [\omega_0 -n] = B_R [\omega_0 -m] \implies f (n) = f (m).$$
Here $\omega-n$ denotes the sequence $\omega$ shifted $n$ units to the left, that is $(\omega-n)(k) = \omega(k+n)$.
An operator $H$ with kernel $(H_{nm})_{nm}$ is pattern equivariant with range $R$ if for all $n_1, n_2, m_1, m_2 \in \Z$ with $n_2 - n_1 = m_2 - m_1$
$$\left(B_R [\omega_0 - n_1] = B_R [\omega_0 - m_1] \text{ and } B_R [\omega_0 - n_2] = B_R [\omega_0 - m_2] \right) \implies H_{n_1 n_2} = H_{m_1 m_2} .$$
An operator has finite interaction range if there is $M>0$ such that $H _{nm}=0$ if $|n-m|>M$.

\newcommand{\CPE}{\mathcal{C}_{\omega_0}}
\begin{definition} Let $\omega_0\subset\mathfrak a^\Z$ be two-sided sequence in the alphabet $\mathfrak a$. 
	The algebra of pattern equivariant functions 
	$\CPE$ is the smallest norm-closed sub-algebra of all bounded complex valued functions on $\Z$ containing
	all pattern equivariant functions of finite radius.  
	The algebra of pattern equivariant short range operators $\mathcal{A}_{\omega_0}$ is the smallest closed sub-algebra of $\mathcal{B}(\mathcal{H})$ containing operators of finite interaction range which are pattern equivariant with finite radius. 
\end{definition}
The simplest tight binding operators are of Kohmoto type, namely they are of the form $H = T+T^*+V$ with $V(n) = r(\omega_0(n))$ where $r:\mathfrak a \to \R$ is a function assigning to a symbol the potential value induced by the atom with that symbol, and $T\psi(n) = \psi(n-1)$. These operators 
are pattern equivariant with radius $0$ and have interaction range $1$. 


As $\CPE$ is a unital commutative $C^*$-algebra, it is isomorphic to the algebra of continuous complex valued functions on a compact Hausdorff space. This space is called the hull of $\omega_0$ and we denote it by
$\Xi_{\omega_0}$. If $\omega_0$ is the sequence underlying the spatial structure of a material we call $\Xi_{\omega_0}$ also the hull of the material. The hull can be described as follows: the set of all sequences $\mathfrak a^\Z$ is a compact totally disconnected space when equipped with the product topology where $\mathfrak a$ carries the discrete topology. The group $\Z$ acts on this space by shifting the sequence to the left. Then $\Xi_{\omega_0}$ is $\overline{Orb(\omega_0)}$, the closure of the orbit of $\omega_0$ under the shift action. An isomorphism 
$C(\Xi_{\omega_0})\ni \tilde f \mapsto f \in  \CPE$ is given by $$f(n) = \tilde f(\omega_0-n).$$

Each point $\omega\in\Xi_{\omega_0}$ is a sequence of symbols and so  describes a configuration of atoms, but it is such that all finite words of $\omega$ occur somewhere in $\omega_0$. 
When $\omega_0$ is repetitive, which means that any finite word of it reoccurs again and again to the left and to the right of the sequence and is equivalent to $\Xi_\omega = \Xi_{\omega_0}$ 
for all $\omega\in\Xi_{\omega_0}$, 
then we may say that all points of $\Xi_{\omega_0}$ are equally well suited to describe the physics of the material. This is the justification of why we look 
at families of Hamiltonians instead of only one Hamiltonian. Indeed, any sequence $\omega\in\Xi_{\omega_0}$ defines a potential 
through $V_\omega(n) = r(\omega(n))$ and therefore a Kohmoto model $H_\omega := T + T^* + V_\omega$.

For our main model we look at a very specific sequence $\omega_0$, namely we take $\mathfrak a = \{a,b\}$ and 
$$\omega_0(n) = \left\{\begin{array}{ll}
	a & \mbox{if } 0 < \{n\theta\} \leq \cut \\
	b & \mbox{otherwise} \end{array}\right. $$
and $r(a) = 1$, $r(b)=0$.
For $\theta=\cut=\frac{3-\sqrt{5}}{2}\approx 0.382$ $\omega_0$ is a symbolic Fibonacci sequence, a part of which looks like
\begin{figure}[!ht]
	\centering
	\begin{tikzpicture}[line cap=round,line join=round,x=1.0cm,y=1.0cm]
		\clip(-5.5,-0.6) rectangle (5.5,0.5);
		\draw [domain=-5.5:5.5] plot(\x,{(-0-0*\x)/10});
		\draw (-0.2,-0.14) node[anchor=north west] {$b$};
		\draw (-7.2,-0.14) node[anchor=north west] {$b$};
		\draw (-6.2,-0.14) node[anchor=north west] {$b$};
		\draw (-4.2,-0.14) node[anchor=north west] {$b$};
		\draw (-3.2,-0.14) node[anchor=north west] {$b$};
		\draw (-1.2,-0.14) node[anchor=north west] {$b$};
		\draw (0.8,-0.2) node[anchor=north west] {$a$};
		\draw (1.8,-0.14) node[anchor=north west] {$b$};
		\draw (3.8,-0.14) node[anchor=north west] {$b$};
		\draw (4.8,-0.16) node[anchor=north west] {$a$};
		\draw (5.8,-0.14) node[anchor=north west] {$b$};
		\draw (2.8,-0.2) node[anchor=north west] {$a$};
		\draw (-2.2,-0.2) node[anchor=north west] {$a$};
		\draw (-5.2,-0.2) node[anchor=north west] {$a$};
		\begin{scriptsize}
			\fill [color=gray] (0,0) circle (1.5pt);
			\draw[color=gray] (0.09,0.15) node {0};
			\fill [color=black] (1,0) circle (1.5pt);
			\draw[color=black] (1.09,0.15) node {1};
			\fill [color=black] (2,0) circle (1.5pt);
			\draw[color=black] (2.09,0.15) node {2};
			\fill [color=gray] (3,0) circle (1.5pt);
			\draw[color=gray] (3.09,0.15) node {3};
			\fill [color=black] (4,0) circle (1.5pt);
			\draw[color=black] (4.09,0.15) node {4};
			\fill [color=gray] (5,0) circle (1.5pt);
			\draw[color=gray] (5.09,0.15) node {5};
			\fill [color=black] (-1,0) circle (1.5pt);
			\draw[color=black] (-0.88,0.15) node {-1};
			\fill [color=gray] (-2,0) circle (1.5pt);
			\draw[color=gray] (-1.88,0.15) node {-2};
			\fill [color=black] (-3,0) circle (1.5pt);
			\draw[color=black] (-2.89,0.15) node {-3};
			\fill [color=black] (-4,0) circle (1.5pt);
			\draw[color=black] (-3.89,0.15) node {-4};
			\fill [color=gray] (-5,0) circle (1.5pt);
			\draw[color=gray] (-4.89,0.15) node {-5};
			\fill [color=black] (-6,0) circle (1.5pt);
			\draw[color=black] (-5.89,0.15) node {-6};
			\fill [color=black] (6,0) circle (1.5pt);
			\draw[color=black] (6.09,0.15) node {6};
			\fill [color=black] (7,0) circle (1.5pt);
			\draw[color=black] (7.09,0.15) node {7};
			\fill [color=black] (-7,0) circle (1.5pt);
			\draw[color=black] (-6.88,0.15) node {-7};
		\end{scriptsize}
	\end{tikzpicture}
\end{figure}

The particular form of this sequence allows us to view the hull in a far more concrete way. Let
\begin{equation*}\label{cut_up_circle}
	\R_{\theta, \cut}:=\{(\phi,\epsilon) | \phi\in \R,\epsilon \in \{\pm 1\}\mbox{ if }\phi\in
	\{0,\cut\}+\Z+\theta\Z, 
	\epsilon = o \mbox{ otherwise}\}. 
\end{equation*}
Thus we have doubled the points of the form $m+n\theta$ or $\cut+m+n\theta$ with $m,n\in\Z$.   
Consider the total order 
$$  (\phi,\epsilon) < (\phi',\epsilon') \Leftrightarrow \phi < \phi' \mbox{ or } 
\phi = \phi', \epsilon <\epsilon ' $$
(with the understanding that $o<o$ is false). This order induces a topology on $\R_{\theta, \cut}$, namely its open intervals are those of the form
$$((\phi,\epsilon) , (\phi',\epsilon'))   = \{(\phi'',\epsilon'' )| (\phi,\epsilon) < (\phi'',\epsilon'')  < (\phi',\epsilon')\}.$$
Note that the open interval 
$((n\theta+m,-),(n'\theta+m',+))$, $n\theta+m<n'\theta+m'$, is equal to the closed interval $[(n\theta+m,+),(n'\theta+m',-)]$. 
We have disconnected the real line along the cut points $$\mathcal C:=\{0,\cut\}+\Z+\theta\Z$$ by doubling them. Finally, we take the quotient
\begin{equation}\label{eq-hull}
\omega_{\theta, \cut} :=\R_{\theta, \cut}/\Z
\end{equation}
by the action of $\Z$ on $\R_{\theta, \cut}$ by translation of the first coordinate $(\phi,\epsilon)\mapsto (\phi+1,\epsilon)$. 
Define
$$v_-(x) = \left\{\begin{array}{ll}
	a & \mbox{if } 0 < \{x\} \leq \cut \\
	b & \mbox{otherwise} \end{array}
\right.\!\!\!\!\!,\:\:\:
v_+(x) = \left\{\begin{array}{ll}
	a & \mbox{if } 0 \leq \{x\} < \cut \\
	b & \mbox{else} \end{array}
\right.\!\!\!\!\!,\:\:\:
v_o(x) = \left\{\begin{array}{ll}
	a & \mbox{if } 0<\{x\} < \cut \\
	b & \mbox{else} \end{array}
\right.$$
Note that the above $\omega_0$ is given by $\omega_0(n)=v_-(n)$.
\begin{theorem}  
	The map $f:\Xi_{\theta, \cut}\to \overline{Orb(\omega_0)}$, 
	$$ f(\phi,\epsilon)=\omega_{\phi,\epsilon}(n) := v_\epsilon(n\theta + \phi)$$
	is a homeomorphism which intertwines the rotation action on $\Xi_{\theta, \cut}$ with the shift action on $\overline{Orb(\omega_0)}\subset \{a,b\}^\Z$.
\end{theorem}
\begin{proof}
	The intertwining property is easily verified.
Clearly $f$ maps $\Xi_{\theta, \cut}$ into $\{a,b\}^\Z$. A sub-base of the topology of $\{a,b\}^\Z$ is given by the sets  $Z_n(x) := \{\omega\in \{0,1\}^\Z:\omega(n) = x\}$ where $x\in \{a,b\}$. We have 
	$\omega_{\phi,+}(0)=a$ if and only if $0\leq \{\phi\} <\cut$ if and only if 
	$(0,-)<(\{\phi\},+) < (\cut,+)$. Similarly $\omega_{\phi,-}(0)=a$ if and only if $0< \{\phi\} \leq \cut$ if and only if 
	$(0,-)<(\{\phi\},-) < (\cut,+)$, and also $\omega_{\phi,o}(0)=a$ if and only if $(0,-)<(\{\phi\},o) < (\cut,+)$. We thus see that the preimage of $Z_0(a)$ is the clopen interval $((0,-),(\cut,+))$. Clearly $Z_0(b)$ is its complement, which is also clopen. The intertwining property now implies that also the preimages of the other cylinder sets are clopen. Hence $f$ is continuous. Since the orbit of $(0,-)$ is dense in $\Xi_{\theta, \cut}$ and mapped bijectively to the orbit of $\omega_0$, the image of $f$ must be the closure of this orbit. We thus see that the preimage map $f^{-1}$ is a bijection between sub-bases of $\overline{Orb(\omega_0)}$ and $\Xi_{\theta, \cut}$. As these spaces are Hausdorff this implies that $f$ is a homeomorphism.  
\end{proof}
\begin{remark}\label{cutandproject_remark}{\rmfamily
A powerful feature of the $C^*$-algebraic formulation of the topological invariants is, that different models can be treated on the same footing, if they are attached to the same abstract dynamical system. The above Kohmoto model is the simplest quasicrystal model.
Its hull $\Xi_{\theta, \cut}$ can equivalently be described via a cut \& project scheme and this leads naturally to other models which have the same hull. To fix ideas, we consider $\cut\le\theta<\frac{1}{2}$. Consider the integer lattice $\Z^2$ in $\R^2$, and parametrize the anti-diagonal $D$ of the unit cell $[0, 1]^2\subset \R^2$ by $[0, 1]: [0, 1]\ni\phi \mapsto i(\phi) \coloneqq \phi e_2 + (1 -\phi)e_1 \in D$. Let $L_\theta$ be the line in $\R^2$ which passes through the origin $0 \in\R^2$ and the point on $D$ parametrized by $\theta$. We define
		$$D'\coloneqq\{i(\phi)\in D|\theta\le\phi\le\theta+\cut\}$$
		and we choose an arbitrary point $v_\delta\in L_\theta$  at distance $\delta$ (small enough) from the origin. Then we define
		$$W'\coloneqq -\pi^\perp (\overline{D\setminus D'})\cup (-\pi^\perp (D')+v_\delta).$$
		
		To obtain a sequence with intercept $\phi \in [0, 1)$ consider the intersection of $L_\phi\coloneqq L_\theta+ i(\phi-\theta)$ with $\Z^2 + W'$. This pattern of points of $L_\phi$ can be ordered, as $L_\phi$ is a line, and so is given by a sequence $w(\phi) = (w_n (\phi))_{n\in\Z}$ of points. In figure \ref{CutProj_twoCuts} we can visualise the construction of the sequence for a choice of $\phi$. Each of these points can be uniquely specified by their $\varphi_n=\{\phi+n\theta\}$ coordinate on $D$. We call $NS$ the set of values for $\phi$ for which $L_\phi$ does not intersect the boundary of $\Z^2+W'$. When $\phi\in NS$, the difference vectors $\omega_\phi(n)=w_n - w_{ n-1}$ take four different values, namely
		$$\omega_\phi(n+1)=\left\{
		\begin{array}{llll}
			a+\delta\eqqcolon c & \mbox{if } \varphi_n<\cut\\
			a & \mbox{if } \cut<\varphi_n<\theta\\
			b-\delta\eqqcolon d & \mbox{if } \theta<\varphi_n<\theta+\cut\\
			b & \mbox{if } \varphi_n>\theta+\cut, 
		\end{array} \right.$$
		where $a$ is the distance between two consecutive segments in the vertical direction and $b$ in the horizontal direction.  
		We see that we can complete the subset of sequences $\omega_\phi(n)$ for $\phi\in NS$ in the topology of $\mathfrak{a} ^\Z$ defining for a singular value $\phi_s\notin NS$ the left and right limits
		$$\omega_{\phi_s}(n)=\lim_{\phi\rightarrow \phi_s^\pm}\omega_\phi(n)$$
		thus doubling, as above, the singular points. A natural tight binding Hamiltonian based on the sequence $\omega_\phi$ is one in which the hoping probabilities $H_{n\!-\!1\, n}$ depend on $w_n - w_{ n-1}$, while the potential may be constant. But from the topological perspective, such a model is equivalent to a Kohmoto model. In fact, a closer look reveals that the hulls and more generally the algebras of pattern equivariant operators of short range are isomorphic.
		
		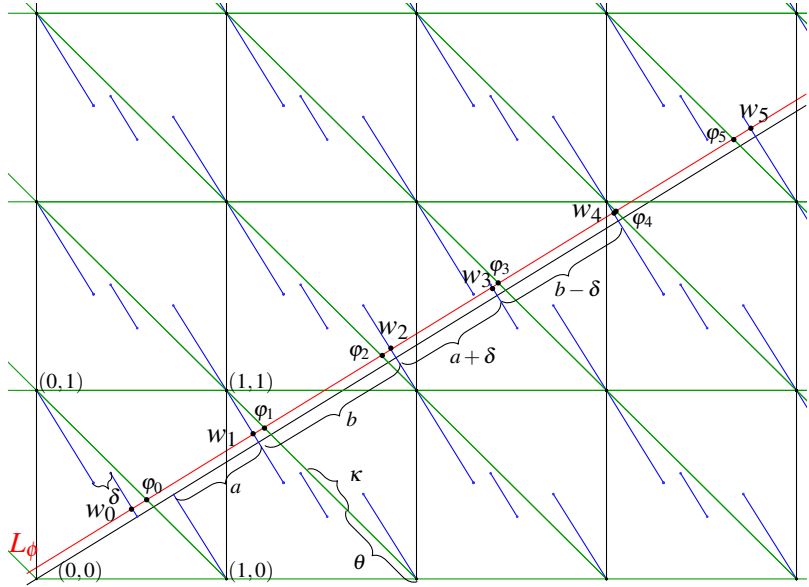
\begin{figure}
			\centering
			\definecolor{ttttff}{rgb}{0.2,0.2,1}
			\definecolor{qqzzqq}{rgb}{0,0.6,0}
			\begin{tikzpicture}[line cap=round,line join=round,x=2.5cm,y=2.5cm]
				\clip(-0.15,-0.05) rectangle (4.05,3.05);
				\draw [color=qqzzqq] (-1,1)-- (0,0);
				\draw [color=qqzzqq] (1,0)-- (0,1);
				\draw [domain=-0.05:4.05] plot(\x,{(-0--0.62*\x)/1});
				\draw [color=qqzzqq] (-1,1)-- (0,1);
				\draw [line width=0.4pt,color=blue] (0.72,0.45)-- (1,0);
				\draw [line width=0.4pt,color=blue] (0,1)-- (0.3,0.51);
				\draw [color=qqzzqq] (0,0)-- (1,0);
				\draw [line width=0.4pt,dash pattern=on 2pt off 2pt,color=blue] (0.3,0.51)-- (0.3,0.51);
				\draw [color=qqzzqq] (0,1)-- (1,0);
				\draw [color=qqzzqq] (2,0)-- (1,1);
				\draw [color=qqzzqq] (0,1)-- (1,1);
				\draw [line width=0.4pt,color=blue] (1.72,0.45)-- (2,0);
				\draw [line width=0.4pt,color=blue] (1,1)-- (1.3,0.51);
				\draw [color=qqzzqq] (1,0)-- (2,0);
				\draw [line width=0.4pt,dash pattern=on 2pt off 2pt,color=blue] (1.3,0.51)-- (1.3,0.51);
				\draw [color=qqzzqq] (1,1)-- (2,0);
				\draw [color=qqzzqq] (3,0)-- (2,1);
				\draw [color=qqzzqq] (1,1)-- (2,1);
				\draw [line width=0.4pt,color=blue] (2.72,0.45)-- (3,0);
				\draw [line width=0.4pt,color=blue] (2,1)-- (2.3,0.51);
				\draw [color=qqzzqq] (2,0)-- (3,0);
				\draw [line width=0.4pt,dash pattern=on 2pt off 2pt,color=blue] (2.3,0.51)-- (2.3,0.51);
				\draw [color=qqzzqq] (2,1)-- (3,0);
				\draw [color=qqzzqq] (4,0)-- (3,1);
				\draw [color=qqzzqq] (2,1)-- (3,1);
				\draw [line width=0.4pt,color=blue] (3.72,0.45)-- (4,0);
				\draw [line width=0.4pt,color=blue] (3,1)-- (3.3,0.51);
				\draw [color=qqzzqq] (3,0)-- (4,0);
				\draw [line width=0.4pt,dash pattern=on 2pt off 2pt,color=blue] (3.3,0.51)-- (3.3,0.51);
				\draw [color=qqzzqq] (3,1)-- (4,0);
				\draw [color=qqzzqq] (5,0)-- (4,1);
				\draw [color=qqzzqq] (3,1)-- (4,1);
				\draw [line width=0.4pt,color=blue] (4,1)-- (4.3,0.51);
				\draw [color=qqzzqq] (4,0)-- (5,0);
				\draw [color=qqzzqq] (4,1)-- (5,0);
				\draw [color=qqzzqq] (4,1)-- (5,1);
				\draw [color=qqzzqq] (-1,2)-- (0,1);
				\draw [color=qqzzqq] (1,1)-- (0,2);
				\draw [color=qqzzqq] (-1,2)-- (0,2);
				\draw [line width=0.4pt,color=blue] (0.72,1.45)-- (1,1);
				\draw [line width=0.4pt,color=blue] (0,2)-- (0.3,1.51);
				\draw [color=qqzzqq] (0,1)-- (1,1);
				\draw [line width=0.4pt,dash pattern=on 2pt off 2pt,color=blue] (0.3,1.51)-- (0.3,1.51);
				\draw [color=qqzzqq] (0,2)-- (1,1);
				\draw [color=qqzzqq] (2,1)-- (1,2);
				\draw [color=qqzzqq] (0,2)-- (1,2);
				\draw [line width=0.4pt,color=blue] (1.72,1.45)-- (2,1);
				\draw [line width=0.4pt,color=blue] (1,2)-- (1.3,1.51);
				\draw [color=qqzzqq] (1,1)-- (2,1);
				\draw [line width=0.4pt,dash pattern=on 2pt off 2pt,color=blue] (1.3,1.51)-- (1.3,1.51);
				\draw [color=qqzzqq] (1,2)-- (2,1);
				\draw [color=qqzzqq] (3,1)-- (2,2);
				\draw [color=qqzzqq] (1,2)-- (2,2);
				\draw [line width=0.4pt,color=blue] (2.72,1.45)-- (3,1);
				\draw [line width=0.4pt,color=blue] (2,2)-- (2.3,1.51);
				\draw [color=qqzzqq] (2,1)-- (3,1);
				\draw [line width=0.4pt,dash pattern=on 2pt off 2pt,color=blue] (2.3,1.51)-- (2.3,1.51);
				\draw [color=qqzzqq] (2,2)-- (3,1);
				\draw [color=qqzzqq] (4,1)-- (3,2);
				\draw [color=qqzzqq] (2,2)-- (3,2);
				\draw [line width=0.4pt,color=blue] (3.72,1.45)-- (4,1);
				\draw [line width=0.4pt,color=blue] (3,2)-- (3.3,1.51);
				\draw [color=qqzzqq] (3,1)-- (4,1);
				\draw [line width=0.4pt,dash pattern=on 2pt off 2pt,color=blue] (3.3,1.51)-- (3.3,1.51);
				\draw [color=qqzzqq] (3,2)-- (4,1);
				\draw [color=qqzzqq] (5,1)-- (4,2);
				\draw [color=qqzzqq] (3,2)-- (4,2);
				\draw [line width=0.4pt,color=blue] (4,2)-- (4.3,1.51);
				\draw [color=qqzzqq] (4,1)-- (5,1);
				\draw [color=qqzzqq] (4,2)-- (5,1);
				\draw [color=qqzzqq] (4,2)-- (5,2);
				\draw [color=qqzzqq] (-1,3)-- (0,2);
				\draw [color=qqzzqq] (1,2)-- (0,3);
				\draw [color=qqzzqq] (-1,3)-- (0,3);
				\draw [line width=0.4pt,color=blue] (0.72,2.45)-- (1,2);
				\draw [line width=0.4pt,color=blue] (0,3)-- (0.3,2.51);
				\draw [color=qqzzqq] (0,2)-- (1,2);
				\draw [line width=0.4pt,dash pattern=on 2pt off 2pt,color=blue] (0.3,2.51)-- (0.3,2.51);
				\draw [color=qqzzqq] (0,3)-- (1,2);
				\draw [color=qqzzqq] (2,2)-- (1,3);
				\draw [color=qqzzqq] (0,3)-- (1,3);
				\draw [line width=0.4pt,color=blue] (1.72,2.45)-- (2,2);
				\draw [line width=0.4pt,color=blue] (1,3)-- (1.3,2.51);
				\draw [color=qqzzqq] (1,2)-- (2,2);
				\draw [line width=0.4pt,dash pattern=on 2pt off 2pt,color=blue] (1.3,2.51)-- (1.3,2.51);
				\draw [color=qqzzqq] (1,3)-- (2,2);
				\draw [color=qqzzqq] (3,2)-- (2,3);
				\draw [color=qqzzqq] (1,3)-- (2,3);
				\draw [line width=0.4pt,color=blue] (2.72,2.45)-- (3,2);
				\draw [line width=0.4pt,color=blue] (2,3)-- (2.3,2.51);
				\draw [color=qqzzqq] (2,2)-- (3,2);
				\draw [line width=0.4pt,dash pattern=on 2pt off 2pt,color=blue] (2.3,2.51)-- (2.3,2.51);
				\draw [color=qqzzqq] (2,3)-- (3,2);
				\draw [color=qqzzqq] (4,2)-- (3,3);
				\draw [color=qqzzqq] (2,3)-- (3,3);
				\draw [line width=0.4pt,color=blue] (3.72,2.45)-- (4,2);
				\draw [line width=0.4pt,color=blue] (3,3)-- (3.3,2.51);
				\draw [color=qqzzqq] (3,2)-- (4,2);
				\draw [line width=0.4pt,dash pattern=on 2pt off 2pt,color=blue] (3.3,2.51)-- (3.3,2.51);
				\draw [color=qqzzqq] (3,3)-- (4,2);
				\draw [color=qqzzqq] (5,2)-- (4,3);
				\draw [color=qqzzqq] (3,3)-- (4,3);
				\draw [line width=0.4pt,color=blue] (4,3)-- (4.3,2.51);
				\draw [color=qqzzqq] (4,2)-- (5,2);
				\draw [color=qqzzqq] (4,3)-- (5,2);
				\draw [color=qqzzqq] (4,3)-- (5,3);
				\draw [color=qqzzqq] (-1,4)-- (0,3);
				\draw [color=qqzzqq] (1,3)-- (0,4);
				\draw [line width=0.4pt,color=blue] (0.72,3.45)-- (1,3);
				\draw [color=qqzzqq] (0,3)-- (1,3);
				\draw [color=qqzzqq] (0,4)-- (1,3);
				\draw [color=qqzzqq] (2,3)-- (1,4);
				\draw [line width=0.4pt,color=blue] (1.72,3.45)-- (2,3);
				\draw [color=qqzzqq] (1,3)-- (2,3);
				\draw [color=qqzzqq] (1,4)-- (2,3);
				\draw [color=qqzzqq] (3,3)-- (2,4);
				\draw [line width=0.4pt,color=blue] (2.72,3.45)-- (3,3);
				\draw [color=qqzzqq] (2,3)-- (3,3);
				\draw [color=qqzzqq] (2,4)-- (3,3);
				\draw [color=qqzzqq] (4,3)-- (3,4);
				\draw [line width=0.4pt,color=blue] (3.72,3.45)-- (4,3);
				\draw [color=qqzzqq] (3,3)-- (4,3);
				\draw [color=qqzzqq] (3,4)-- (4,3);
				\draw [color=qqzzqq] (5,3)-- (4,4);
				\draw [color=qqzzqq] (4,3)-- (5,3);
				\draw [color=qqzzqq] (4,4)-- (5,3);
				\draw [color=red,domain=-0.05:4.05] plot(\x,{(--0.06--0.62*\x)/1});
				\draw (0,0)-- (0,4);
				\draw (1,0)-- (1,4);
				\draw (2,0)-- (2,4);
				\draw (3,0)-- (3,4);
				\draw (4,0)-- (4,4);
				\draw [color=blue] (0.39,0.56)-- (0.53,0.33);
				\draw [color=blue] (0.39,1.56)-- (0.53,1.33);
				\draw [color=blue] (0.39,2.56)-- (0.53,2.33);
				\draw [color=blue] (1.39,0.56)-- (1.53,0.33);
				\draw [color=blue] (2.39,0.56)-- (2.53,0.33);
				\draw [color=blue] (3.39,0.56)-- (3.53,0.33);
				\draw [color=blue] (1.39,1.56)-- (1.53,1.33);
				\draw [color=blue] (1.39,2.56)-- (1.53,2.33);
				\draw [color=blue] (2.39,1.56)-- (2.53,1.33);
				\draw [color=blue] (2.39,2.56)-- (2.53,2.33);
				\draw [color=blue] (3.39,1.56)-- (3.53,1.33);
				\draw [color=blue] (3.39,2.56)-- (3.53,2.33);
				\begin{scriptsize}
					\draw[color=red] (-0.06,0.17) node {\normalsize $L_{\phi}$};
					\draw[color=black] (0.23,0.04) node {$(0, 0)$};
					\draw[color=black] (1.13,0.04) node {$(1, 0)$};
					\draw[color=black] (0.13,1.04) node {$(0, 1)$};
					\draw[color=black] (1.13,1.04) node {$(1, 1)$};
					
					\fill [color=black] (0,0) circle (0.5pt);
					\fill [color=black] (1,0) circle (0.5pt);
					\fill [color=black] (0,1) circle (0.5pt);
					\fill [color=black] (-1,1) circle (0.5pt);
					\fill [color=ttttff] (0.3,0.51) circle (0.5pt);
					\fill [color=black] (0,1) circle (0.5pt);
					\fill [color=black] (1,0) circle (0.5pt);
					\fill [color=black] (2,0) circle (0.5pt);
					\fill [color=black] (1,1) circle (0.5pt);
					\fill [color=black] (0,1) circle (0.5pt);
					\fill [color=black] (1,1) circle (0.5pt);
					\fill [color=ttttff] (1.72,0.45) circle (0.5pt);
					\fill [color=black] (2,0) circle (0.5pt);
					\fill [color=black] (1,1) circle (0.5pt);
					\fill [color=ttttff] (1.3,0.51) circle (0.5pt);
					\fill [color=black] (1,0) circle (0.5pt);
					\fill [color=black] (2,0) circle (0.5pt);
					\fill [color=ttttff] (1.3,0.51) circle (0.5pt);
					\fill [color=ttttff] (1.3,0.51) circle (0.5pt);
					\fill [color=black] (1,1) circle (0.5pt);
					\fill [color=black] (2,0) circle (0.5pt);
					\fill [color=black] (3,0) circle (0.5pt);
					\fill [color=black] (2,1) circle (0.5pt);
					\fill [color=black] (1,1) circle (0.5pt);
					\fill [color=black] (2,1) circle (0.5pt);
					\fill [color=ttttff] (2.72,0.45) circle (0.5pt);
					\fill [color=black] (3,0) circle (0.5pt);
					\fill [color=black] (2,1) circle (0.5pt);
					\fill [color=ttttff] (2.3,0.51) circle (0.5pt);
					\fill [color=black] (2,0) circle (0.5pt);
					\fill [color=black] (3,0) circle (0.5pt);
					\fill [color=ttttff] (2.3,0.51) circle (0.5pt);
					\fill [color=ttttff] (2.3,0.51) circle (0.5pt);
					\fill [color=black] (2,1) circle (0.5pt);
					\fill [color=black] (3,0) circle (0.5pt);
					\fill [color=black] (4,0) circle (0.5pt);
					\fill [color=black] (3,1) circle (0.5pt);
					\fill [color=black] (2,1) circle (0.5pt);
					\fill [color=black] (3,1) circle (0.5pt);
					\fill [color=ttttff] (3.72,0.45) circle (0.5pt);
					\fill [color=black] (4,0) circle (0.5pt);
					\fill [color=black] (3,1) circle (0.5pt);
					\fill [color=ttttff] (3.3,0.51) circle (0.5pt);
					\fill [color=black] (3,0) circle (0.5pt);
					\fill [color=black] (4,0) circle (0.5pt);
					\fill [color=ttttff] (3.3,0.51) circle (0.5pt);
					\fill [color=ttttff] (3.3,0.51) circle (0.5pt);
					\fill [color=black] (3,1) circle (0.5pt);
					\fill [color=black] (4,0) circle (0.5pt);
					\fill [color=black] (4,1) circle (0.5pt);
					\fill [color=black] (3,1) circle (0.5pt);
					\fill [color=black] (4,1) circle (0.5pt);
					\fill [color=black] (4,1) circle (0.5pt);
					\fill [color=black] (4,0) circle (0.5pt);
					\fill [color=black] (4,1) circle (0.5pt);
					\fill [color=black] (4,1) circle (0.5pt);
					\fill [color=black] (-1,2) circle (0.5pt);
					\fill [color=black] (0,1) circle (0.5pt);
					\fill [color=black] (1,1) circle (0.5pt);
					\fill [color=black] (0,2) circle (0.5pt);
					\fill [color=black] (-1,2) circle (0.5pt);
					\fill [color=black] (0,2) circle (0.5pt);
					\fill [color=ttttff] (0.72,1.45) circle (0.5pt);
					\fill [color=black] (1,1) circle (0.5pt);
					\fill [color=black] (0,2) circle (0.5pt);
					\fill [color=ttttff] (0.3,1.51) circle (0.5pt);
					\fill [color=black] (0,1) circle (0.5pt);
					\fill [color=black] (1,1) circle (0.5pt);
					\fill [color=ttttff] (0.3,1.51) circle (0.5pt);
					\fill [color=ttttff] (0.3,1.51) circle (0.5pt);
					\fill [color=black] (0,2) circle (0.5pt);
					\fill [color=black] (1,1) circle (0.5pt);
					\fill [color=black] (2,1) circle (0.5pt);
					\fill [color=black] (1,2) circle (0.5pt);
					\fill [color=black] (0,2) circle (0.5pt);
					\fill [color=black] (1,2) circle (0.5pt);
					\fill [color=ttttff] (1.72,1.45) circle (0.5pt);
					\fill [color=black] (2,1) circle (0.5pt);
					\fill [color=black] (1,2) circle (0.5pt);
					\fill [color=ttttff] (1.3,1.51) circle (0.5pt);
					\fill [color=black] (1,1) circle (0.5pt);
					\fill [color=black] (2,1) circle (0.5pt);
					\fill [color=ttttff] (1.3,1.51) circle (0.5pt);
					\fill [color=ttttff] (1.3,1.51) circle (0.5pt);
					\fill [color=black] (1,2) circle (0.5pt);
					\fill [color=black] (2,1) circle (0.5pt);
					\fill [color=black] (3,1) circle (0.5pt);
					\fill [color=black] (2,2) circle (0.5pt);
					\fill [color=black] (1,2) circle (0.5pt);
					\fill [color=black] (2,2) circle (0.5pt);
					\fill [color=ttttff] (2.72,1.45) circle (0.5pt);
					\fill [color=black] (3,1) circle (0.5pt);
					\fill [color=black] (2,2) circle (0.5pt);
					\fill [color=ttttff] (2.3,1.51) circle (0.5pt);
					\fill [color=black] (2,1) circle (0.5pt);
					\fill [color=black] (3,1) circle (0.5pt);
					\fill [color=ttttff] (2.3,1.51) circle (0.5pt);
					\fill [color=ttttff] (2.3,1.51) circle (0.5pt);
					\fill [color=black] (2,2) circle (0.5pt);
					\fill [color=black] (3,1) circle (0.5pt);
					\fill [color=black] (4,1) circle (0.5pt);
					\fill [color=black] (3,2) circle (0.5pt);
					\fill [color=black] (2,2) circle (0.5pt);
					\fill [color=black] (3,2) circle (0.5pt);
					\fill [color=ttttff] (3.72,1.45) circle (0.5pt);
					\fill [color=black] (4,1) circle (0.5pt);
					\fill [color=black] (3,2) circle (0.5pt);
					\fill [color=ttttff] (3.3,1.51) circle (0.5pt);
					\fill [color=black] (3,1) circle (0.5pt);
					\fill [color=black] (4,1) circle (0.5pt);
					\fill [color=ttttff] (3.3,1.51) circle (0.5pt);
					\fill [color=ttttff] (3.3,1.51) circle (0.5pt);
					\fill [color=black] (3,2) circle (0.5pt);
					\fill [color=black] (4,1) circle (0.5pt);
					\fill [color=black] (4,2) circle (0.5pt);
					\fill [color=black] (3,2) circle (0.5pt);
					\fill [color=black] (4,2) circle (0.5pt);
					\fill [color=black] (4,2) circle (0.5pt);
					\fill [color=black] (4,1) circle (0.5pt);
					\fill [color=black] (4,2) circle (0.5pt);
					\fill [color=black] (4,2) circle (0.5pt);
					\fill [color=black] (-1,3) circle (0.5pt);
					\fill [color=black] (0,2) circle (0.5pt);
					\fill [color=black] (1,2) circle (0.5pt);
					\fill [color=black] (0,3) circle (0.5pt);
					\fill [color=black] (-1,3) circle (0.5pt);
					\fill [color=black] (0,3) circle (0.5pt);
					\fill [color=ttttff] (0.72,2.45) circle (0.5pt);
					\fill [color=black] (1,2) circle (0.5pt);
					\fill [color=black] (0,3) circle (0.5pt);
					\fill [color=ttttff] (0.3,2.51) circle (0.5pt);
					\fill [color=black] (0,2) circle (0.5pt);
					\fill [color=black] (1,2) circle (0.5pt);
					\fill [color=ttttff] (0.3,2.51) circle (0.5pt);
					\fill [color=ttttff] (0.3,2.51) circle (0.5pt);
					\fill [color=black] (0,3) circle (0.5pt);
					\fill [color=black] (1,2) circle (0.5pt);
					\fill [color=black] (2,2) circle (0.5pt);
					\fill [color=black] (1,3) circle (0.5pt);
					\fill [color=black] (0,3) circle (0.5pt);
					\fill [color=black] (1,3) circle (0.5pt);
					\fill [color=ttttff] (1.72,2.45) circle (0.5pt);
					\fill [color=black] (2,2) circle (0.5pt);
					\fill [color=black] (1,3) circle (0.5pt);
					\fill [color=ttttff] (1.3,2.51) circle (0.5pt);
					\fill [color=black] (1,2) circle (0.5pt);
					\fill [color=black] (2,2) circle (0.5pt);
					\fill [color=ttttff] (1.3,2.51) circle (0.5pt);
					\fill [color=ttttff] (1.3,2.51) circle (0.5pt);
					\fill [color=black] (1,3) circle (0.5pt);
					\fill [color=black] (2,2) circle (0.5pt);
					\fill [color=black] (3,2) circle (0.5pt);
					\fill [color=black] (2,3) circle (0.5pt);
					\fill [color=black] (1,3) circle (0.5pt);
					\fill [color=black] (2,3) circle (0.5pt);
					\fill [color=ttttff] (2.72,2.45) circle (0.5pt);
					\fill [color=black] (3,2) circle (0.5pt);
					\fill [color=black] (2,3) circle (0.5pt);
					\fill [color=ttttff] (2.3,2.51) circle (0.5pt);
					\fill [color=black] (2,2) circle (0.5pt);
					\fill [color=black] (3,2) circle (0.5pt);
					\fill [color=ttttff] (2.3,2.51) circle (0.5pt);
					\fill [color=ttttff] (2.3,2.51) circle (0.5pt);
					\fill [color=black] (2,3) circle (0.5pt);
					\fill [color=black] (3,2) circle (0.5pt);
					\fill [color=black] (4,2) circle (0.5pt);
					\fill [color=black] (3,3) circle (0.5pt);
					\fill [color=black] (2,3) circle (0.5pt);
					\fill [color=black] (3,3) circle (0.5pt);
					\fill [color=ttttff] (3.72,2.45) circle (0.5pt);
					\fill [color=black] (4,2) circle (0.5pt);
					\fill [color=black] (3,3) circle (0.5pt);
					\fill [color=ttttff] (3.3,2.51) circle (0.5pt);
					\fill [color=black] (3,2) circle (0.5pt);
					\fill [color=black] (4,2) circle (0.5pt);
					\fill [color=ttttff] (3.3,2.51) circle (0.5pt);
					\fill [color=ttttff] (3.3,2.51) circle (0.5pt);
					\fill [color=black] (3,3) circle (0.5pt);
					\fill [color=black] (4,2) circle (0.5pt);
					\fill [color=black] (4,3) circle (0.5pt);
					\fill [color=black] (3,3) circle (0.5pt);
					\fill [color=black] (4,3) circle (0.5pt);
					\fill [color=black] (4,3) circle (0.5pt);
					\fill [color=black] (4,2) circle (0.5pt);
					\fill [color=black] (4,3) circle (0.5pt);
					\fill [color=black] (4,3) circle (0.5pt);
					\fill [color=black] (0,3) circle (0.5pt);
					\fill [color=black] (1,3) circle (0.5pt);
					\fill [color=black] (1,3) circle (0.5pt);
					\fill [color=black] (0,3) circle (0.5pt);
					\fill [color=black] (1,3) circle (0.5pt);
					\fill [color=black] (1,3) circle (0.5pt);
					\fill [color=black] (2,3) circle (0.5pt);
					\fill [color=black] (2,3) circle (0.5pt);
					\fill [color=black] (1,3) circle (0.5pt);
					\fill [color=black] (2,3) circle (0.5pt);
					\fill [color=black] (2,3) circle (0.5pt);
					\fill [color=black] (3,3) circle (0.5pt);
					\fill [color=black] (3,3) circle (0.5pt);
					\fill [color=black] (2,3) circle (0.5pt);
					\fill [color=black] (3,3) circle (0.5pt);
					\fill [color=black] (3,3) circle (0.5pt);
					\fill [color=black] (4,3) circle (0.5pt);
					\fill [color=black] (4,3) circle (0.5pt);
					\fill [color=black] (3,3) circle (0.5pt);
					\fill [color=black] (4,3) circle (0.5pt);
					\fill [color=black] (4,3) circle (0.5pt);
					\fill [color=black] (4,3) circle (0.5pt);
					\fill [color=black] (1.14,0.77) circle (1.0pt);

					\draw[color=black] (0.34,0.35) node {\normalsize $w_0$};
					\draw[color=black] (0.98,0.75) node {\normalsize $w_1$};
					\draw[color=black] (1.88,1.32) node {\normalsize $w_2$};
					\draw[color=black] (2.32,1.59) node {\normalsize $w_3$};
					\draw[color=black] (3.78,2.47) node {\normalsize $w_5$};
					\draw[color=black] (2.9,1.94) node {\normalsize $ w_4$};
					\draw[color=black] (0.61,0.49) node {$\varphi_0$};
					\draw[color=black] (1.20,0.88) node {$\varphi_1$};
					\draw[color=black] (1.7,1.2) node {$\varphi_2$};
					\draw[color=black] (2.44,1.65) node {$\varphi_3$};
					\draw[color=black] (3.19,1.9) node {$ \varphi_4$};
					\draw[color=black] (3.58,2.35) node {$\varphi_5$};
					
					\fill [color=black] (3.76,2.39) circle (1.0pt);
					\fill [color=black] (3.04,1.94) circle (1.0pt);
					\fill [color=black] (0.58,0.42) circle (1.0pt);
					\fill [color=black] (1.2,0.8) circle (1.0pt);
					\fill [color=black] (1.82,1.185) circle (1.0pt);
					\fill [color=black] (2.43,1.57) circle (1.0pt);
					\fill [color=black] (3.05,1.95) circle (1.0pt);
					\fill [color=black] (3.67,2.33) circle (1.0pt);
					\fill [color=black] (1.865,1.224) circle (1.0pt);
					
					\fill [color=ttttff] (0.39,0.56) circle (0.5pt);
					\fill [color=ttttff] (0.53,0.33) circle (0.5pt);
					\fill [color=ttttff] (0.39,1.56) circle (0.5pt);
					\fill [color=ttttff] (0.53,1.33) circle (0.5pt);
					\fill [color=ttttff] (0.39,2.56) circle (0.5pt);
					\fill [color=ttttff] (0.53,2.33) circle (0.5pt);
					\fill [color=ttttff] (1.39,0.56) circle (0.5pt);
					\fill [color=ttttff] (1.53,0.33) circle (0.5pt);
					\fill [color=ttttff] (2.39,0.56) circle (0.5pt);
					\fill [color=ttttff] (2.53,0.33) circle (0.5pt);
					\fill [color=ttttff] (3.39,0.56) circle (0.5pt);
					\fill [color=ttttff] (3.53,0.33) circle (0.5pt);
					\fill [color=ttttff] (1.39,1.56) circle (0.5pt);
					\fill [color=ttttff] (1.53,1.33) circle (0.5pt);
					\fill [color=ttttff] (1.39,2.56) circle (0.5pt);
					\fill [color=ttttff] (1.53,2.33) circle (0.5pt);
					\fill [color=ttttff] (2.39,1.56) circle (0.5pt);
					\fill [color=ttttff] (2.53,1.33) circle (0.5pt);
					\fill [color=ttttff] (2.39,2.56) circle (0.5pt);
					\fill [color=ttttff] (2.53,2.33) circle (0.5pt);
					\fill [color=ttttff] (3.39,1.56) circle (0.5pt);
					\fill [color=ttttff] (3.53,1.33) circle (0.5pt);
					\fill [color=ttttff] (3.39,2.56) circle (0.5pt);
					\fill [color=ttttff] (3.53,2.33) circle (0.5pt);
					\fill [color=black] (0.5,0.37) circle (1.0pt);
					\fill [color=black] (2.4,1.54) circle (1.0pt);
					
					\draw [decorate,decoration={brace,amplitude=5pt, mirror,raise=4ex}](0.19,0.675) -- (0.287,0.735);
					\draw[color=black] (0.41,0.44) node {$\delta$};

					\draw [decorate,decoration={brace,amplitude=5pt,mirror,raise=4ex}]
					(1.77,0.50) -- (2.144,0.125);
					\draw[color=black] (1.7,0.09) node {$\theta$};
					\draw [decorate,decoration={brace,amplitude=5pt,raise=4ex}](1.285,0.445) -- (1.495,0.225);
					\draw[color=black] (1.68,0.55) node {$\cut$};
					
					\draw [decorate,decoration={brace,amplitude=5pt, mirror,raise=4ex}]
					(0.63,0.5955) -- (1.08,0.870) node[midway,yshift=-3em]{};
					\draw[color=black] (1.05,0.48) node {$a$};
					\draw [decorate,decoration={brace,amplitude=5pt, mirror,raise=4ex}]
					(1.1,0.8755) -- (1.81,1.310375) node[midway,yshift=-3em]{};
					\draw[color=black] (1.66,0.85) node {$b$};
					
					\draw [decorate,decoration={brace,amplitude=5pt, mirror,raise=4ex}]
					(1.82,1.315) -- (2.34,1.6335) node[midway,yshift=-3em]{};
					\draw[color=black] (2.3,1.18) node {$a+\delta$};
					\draw [decorate,decoration={brace,amplitude=5pt, mirror,raise=4ex}]
					(2.345,1.642) -- (2.975,2.027875) node[midway,yshift=-3em]{};
					\draw[color=black] (2.85,1.54) node {$b-\delta$};
				\end{scriptsize}
			\end{tikzpicture}
			\caption{The cut \& project scheme for the model $\Xi_{\theta, \cut}$.}
			\label{CutProj_twoCuts}
		\end{figure}
	}    
\end{remark}

\subsection{\texorpdfstring{The crossed product $C^*$-algebra and its Toeplitz extension}{}}
It is very useful to describe the algebra of short range pattern equivariant operators $\mathcal{A}_{\omega_0}$ more abstractly, that is, to see it as a faithful representation of a $C^*$-algebra whose other representations yield other microscopically different but macroscopically indistinguishable realisations of the solid. This algebra is in general a groupoid $C^*$-algebra, but in one dimension it reduces to the crossed product algebra 
$C(\Xi_{\omega_0})\rtimes_\alpha\Z$. Besides, such crossed products with $\Z$ have a natural extension, the so-called Toeplitz extension, which plays an fundmental role in the bulk edge correspondence.

As we will need this further down, let us briefly describe the construction of 
the crossed product algebra and its Toeplitz extension in the case of a $\Z$-action $\alpha:X\to X$ by homeomorphisms on a metrizable space $X$. 
Then $\alpha(f)=f\circ\alpha^{-1}$ is a $\Z$-action on the algebra $C_0(X)$ of continuous functions complex valued functions which vanish at infinity. The algebraic crossed product $C_0(X)_\alpha\Z$ of $C_0(X)$ with $\Z$ is defined as follows. The elements of this algebra are non-commutative Laurent polynomials in one variable $u$ with coefficients in $C_0(X)$ (hence finite linear combinations of formal expressions $fu^n$, $f \in C_0(X)$, $n\in \Z$) which satisfy the relation
\begin{equation}
	\label{eq:1.1}
	uf=\alpha(f)u.
\end{equation}
and $fu^0=f$. 
The algebra is equipped with the adjoint map 
$$(fu^n)^* = u^{-n}  \bar{f}$$
and so it is a complex $*$-algebra and $u$ is a unitary\footnote{if $X$ is not compact so that $C_0(X)$ not unital then $u$ is not in the algebra but $f u^{-1} = f u^*$}.

Any point $x\in X$ gives rise to a representation $\rho_x$ of $C_0(X)_\alpha \Z$ on $\ell^2(\Z)$
$$\rho_{x}(f u^n)\psi(k) = f(\alpha^k(x)) \psi(k-n).$$
The $C^*$-crossed product $C_0(X)\rtimes_\alpha\Z $ is the completion of $C_0(X)_\alpha\Z$ in the norm
$$\lVert a \rVert=\sup_{x\in X}\lVert \rho_{x}(a) \rVert_{\mathcal{B}(\ell^2(\Z))}.$$
Stated differently, the family $\{\rho_x\}_{x\in X}$ is a faithful family of representations, and $C_0(X)\rtimes_\alpha\Z $ is isomorphic to the direct sum representation. 

Given an element $h$ of the algebra, we get a covariant family of operators $H_X=\{H_x\}_{x\in X}$, where $H_x=\rho_x(h)$. 
Covariance means here that $H_{\alpha(x)} = U H_{x} U^*$ where $U$ is the translation operator on $\ell^2(\Z)$. Furthermore, the spectrum of $h$ is the spectrum of this family,
$$\mathrm{spec}(h)=\mathrm{spec}(H_X):=\bigcup_{x\in X} \mathrm{spec}(H_x).$$
If $X$ has a transitive point $x$ (a point with a dense orbit) then $\rho_x$ is faithful and
$\mathrm{spec}(H_X)= \mathrm{spec}(H_x)$. But in general none of the representations $\rho_x$ need to be faithful and so the individual Hamiltonians $H_x$ might have a smaller spectrum. 

The crossed product algebra has a natural extension. For that we consider 
the $*$-algebra $C_0(X)_\alpha \N$ given by non-commutative polynomials in $\hat u$ and its adjoint $\hat u^*$ with coefficients in $C_0(X)$. These satisfy the relations
$$\hat{u}f  = \alpha(f)\hat{u},\quad \hat{u}^*\hat{u} = 1,\quad  
\hat{u}\hat{u}^* = 1 - \hat{p},\quad \hat{p}\hat{p}=\hat{p}, \quad \hat{p}^*=\hat{p},\quad \hat{p} f= f\hat{p} .$$
Its elements are finite sums of elements of the form
$$f \hat u^n \hat p^l ({\hat u^*})^m$$ 
where $n,m\in\N$ and $l=0,1$. 

The map $\pi:C_0(X)_\alpha \N\to C_0(X)_\alpha \Z$,
$$\pi(f \hat u^n \hat p^l ({\hat u^*})^m) = \delta_{l 0}f u^{n-m}$$
is a surjective $*$-algebra morphism. Its kernel are the elements of the form $f \hat u^n \hat p ({\hat u^*})^m$. This kernel is isomorphic to 
$C_0(X)\otimes F(\ell^2(\N))$ where $F(\ell^2(\N))$ are the finite rank operators. An isomorphism between $C_0(X)\otimes \mathcal{F}(\ell^2(\N))$ and $\ker \pi$ is given by  $\b(f \otimes E_{nm}) = f \hat{u}^n \hat p (\hat{u}^*)^m$ where $E_{nm}$ is the elementary matrix when we identify $F(\ell^2(\N))$ with half sided infinite matrices which have only finitely many non-zero entries.

The representations $\rho_x$ restrict to 
representations $\hat\rho_x$ of $C_0(X)_\alpha \N$ on $\ell^2(\N)$
$$\hat\rho_{x}(f)\psi(k) = f(\alpha^n(x)) \psi(k-n), \quad \hat\rho_{x}(\hat u)\psi(k) = \psi(k-1)$$
where $\psi(k-1)=0$ if $k-1=0$. This implies that $\hat\rho_{x}(\hat u^*)\psi(k) = \psi(k+1)$ and hence $\hat\rho_{x}(\hat p)\psi(k) = \delta_{k0} \psi(k)$. 
The completion of  $C_0(X)_\alpha \N$ in the norm
$\lVert a \rVert=\sup_{x\in X}\lVert \hat\rho_{x}(a) \rVert_{\mathcal{B}(\ell^2(\N))}$ is called the Toeplitz extension of $C_0(X)_\alpha \Z$
and denoted by $\mathcal T(C_0(X),\alpha)$. The epimorphism $\pi$ extends to the closures, $\pi:\mathcal T(C_0(X),\alpha)\to C_0(X)_\alpha \Z$ and its kernel is isomorphic to $C_0(X)\otimes\mathcal K(\ell^2(\N))$,
where $\mathcal K$ denotes the compact operators. 

In the case $X=\Xi_{\omega_0}$, the representation$$\rho_{\omega_0}(f u^n)\psi(k) = f (\omega_0-k)\psi(n-k)$$
is faithful, as $\omega_0$ has a dense orbit by construction.
As the cylinder sets $Z_n(a)$ 
form a sub-base of the topology of $\Xi$, 
$C(\Xi_{\omega_0})_\alpha\Z$ is given by finite sums of products of the form $c\chi_{Z_n(a)}u^m$, $c\in\C$, $a\in\mathfrak a$, $n,m\in\Z$, and these are pattern equivariant with range $n$ and interaction range $m$. Thus $\rho_{\omega_0}(C(\Xi_{\omega_0})\rtimes_\alpha\Z)=\mathcal A_{\omega_0}$.
In particular, if $\Xi = \Xi_{\theta,\cut}$ is the hull defined in \eqref{eq-hull} then 
\begin{equation}
	\label{our_model}
	h = h(\lambda) := u+u^*+\lambda(1-2\chi_{[0, \cut]}).
\end{equation}
with $\lambda$ a coupling constant and 
$\chi_{[0, \cut]}$ the characteristic function on 
$[(0,+),(\cut,-)]$, is an example of an operator whose representative $\rho_{0,-}(h)$ is a pattern equivariant operator of short range and $\rho_{\omega_0}(h(-\frac12)+\frac12)$ with $\omega_0$ as above yields (up to an additive constant) the Hamiltonian of the introduction.

\subsection{$K$-theory and topological phases}\label{sec-K-theory}
In the $C^*$-algebraic approach, the topological invariants are the elements of the $K$-groups of the algebra $\mathcal A$ of pattern equivariant operators with short range. In this approach, two physical systems described by Hamiltonians $h_1$, $h_2$ of the algebra of pattern equivariant operators of short range, which have a gap at the Fermi energy, belong to the same topological phase if they are homotopic in that algebra along a path along which the gap does not close. Adding constants, we may always arrange the situation to be such that the gaps are around $0$ and then a such Hamiltonians are invertible self-adjoint elements of $\mathcal A$. We make one compromise which is essential to apply the algebraic topology methods we need for the bulk edge correspondence, namely we allow for stacked systems. Stacking of two systems can be interpreted as addition in the $K$-group of $\mathcal A$. Let us recall the essential notions of the description of topological phases of insulators using van Daele's picture of $K$-theory \cite{kellendonk2017c}.

Given a unital $C^*$-algebra $A$ we denote by $GL_n^{s.a.}(A)$ the invertible self-adjoint elements in $M_n(A)$. This is a topological space equipped with the norm topology of $M_n(A)$.
We view $GL_n^{s.a.}(A)$ as a subset of $GL_{n+m}^{s.a.}(A)$, namely every $n\times n$ matrix $h$ of $GL_n^{s.a.}(A)$ can be completed with the identity $m\times m$ matrix $1_m$ to the element $h\oplus 1_m\in GL_{n+m}^{s.a.}(A)$. Denoting by $GL_\infty^{s.a.}(A)/\sim$ the union over all this spaces modulo homotopy we define addition of homotopy classes as 
$$[h_1]\oplus [h_2]=[h_1\oplus h_2]=\left[\begin{pmatrix}
h_1 & 0\\
0 & h_2
\end{pmatrix}\right].$$
It turns out that $h_1\oplus h_2$ is homotopic to $h_2\oplus h_1$ and so $GL_\infty^{s.a.}(A)/\sim$ is an abelian semigroup. Physically, this addition can be interpreted as stacking, $\begin{pmatrix}
h_1 & c\\
c^* & h_2
\end{pmatrix}$ is a two layer system, one layer described by $h_1$ the other by $h_2$, and these two layers may weakly interact, $c$ denoting a coupling. If $c$ is small enough so that $\begin{pmatrix}
h_1 & c\\
c^* & h_2
\end{pmatrix}$ is homotopic to $\begin{pmatrix}
h_1 & 0\\
0 & h_2
\end{pmatrix}$
then the two layer system is in the same topological phase as the sum of the two systems. Furthermore, a system with a Hamiltonian $h$ which has only strictly positive energy states is in the topological phase of $1$ and the latter is trivial from a physical point of view, as the Fermi energy is below the energies of the states of the Hamiltonian. And indeed, $[1]$ is the $0$-element of the semigroup and $h\oplus 1$ in the same class as $h$.
Elements of the class $[-1]$ are Hamiltonians which have only negative energy states. They are always homotopic to so-called atomic limits, that is, operators without kinetic part. In the tight binding approximation these correspond to systems with filled localised orbitals. In physics they are also considered as topologically trivial, although they have non-zero integrated density of states and define non-trivial elements in the above semigroup. 

The $K_0$-group of $A$ is obtained via Grothendieck's construction and defines relative topological phases: Given two pairs $([h_1],[h'_1])$ and 
$([h_2],[h'_2])$ of $GL_\infty^{s.a.}(A)/\sim\times GL_\infty^{s.a.}(A)/\sim$, say that they are equivalent in the sense of Grothendieck if there is a class $[k]\in GL_\infty^{s.a.}(A)/\sim$ such that
$$[h_1]+[h'_2] + [k] = [h'_1]+[h_2] + [k].$$
The possibility of adding a class $[k]$ in the above definition is important, as usually the semigroup is not cancellative, but it can be shown that it suffices to require $[k]$ to be a sum of classes $[1]$ and $[-1]$.
So the elements of $K_0(A)$ are equivalence classes of pairs $[[h],[h']]$. 
Addition is component wise $[[h_1],[h'_1]]+[[h_2],[h'_2]] = [[h_1+h_2],[h'_1+h_2']]$ and $-[[h],[h']]=[[h'],[h]]$.
This looks quite complicated, but it turns out that any class can be brought into the form $[[h],[1_m\oplus -1_{n-m}]]$ where $h$ is a self-adjoint invertible element of $M_n(A)$ and $m\leq n$. In particular, the negative of $[[h],[1_n]]$ is 
$[[1_n],[h]] = [[-h],[-1_n]]$. 
We simply denote $[[h],[1_n]]$ by $[h]_0$.  Let us mention that the group structure on $K_0(A)$ will be important in what follows, as we will have to subtract elements. 

The definition of $K_1(A)$ is similar to that of $K_0(A)$, except one drops the requirement that the elements are self-adjoint. As $U_n(A)$, the unitary elements of $M_n(A)$,  are a deformation retract of the invertible elements, we can work with $U_\infty(A)$.
Furthermore, $u\oplus v$ is homotopic to $uv\oplus 1$ and therefore $U_\infty(A)/\sim$ is already an abelian group, $-[u]_1 = [u^{-1}]_1$. It follows that the Grothendieck construction does not add anything new so that we can see $K_1(A)$ as this  group $U_\infty(A)/\sim$. We denote its elements by $[u]_1$. 

Given a unital $C^*$-algebra morphism $f:A\to B$ between unital $C^*$-algebras, we define $f_*:K_i(A)\to K_i(B)$ through $f_*([x]) = [f(x)]$, where either $[x] = [h]_0$ or $[x]=[u]_1$. 

The definition of the $K$-groups becomes more subtle if $A$ is not unital. One first adds a unit to $A$. This gives $A^+:= A\times\C$ with multiplication $(a,\lambda)(b,\mu) = (ab+\mu a+\lambda b,\lambda\mu)$ and  $(a,\lambda)^* =  (a^*,\bar \lambda)$. The map $(a,\lambda)\mapsto \lambda$ is a unital epimorphism $s:A^+\to \C$ whose kernel is $A$. By definition, $K_i(A) = \ker s_*$. 

In the standard literature on $K$-theory, like \cite{rordam2000introduction}, $K_0(A)$ is defined using projections. This is equivalent to the above approach, as the map which assigns to an invertible self-adjoint element
$h$ its projection onto the negative spectral states yields a bijection between $GL_n^{s.a.}(A)/\sim$ and $Proj_n(A)/\sim$ where $Proj_n(A)$ denotes the space of orthogonal projections of $M_n(A)$. We will mostly follow this formulation, which in the physical context means that we consider homotopy classes of Fermi projections.

One of the fundamental results of importance for us is the six-term exact sequence.
\begin{theorem}[Six-Term exact sequence]
Let
\begin{equation}\label{SES}
    \begin{tikzcd}
	0 & I & A & B & 0.
	\arrow[from=1-1, to=1-2]
	\arrow["\varphi", from=1-2, to=1-3]
	\arrow["\psi", from=1-3, to=1-4]
	\arrow[from=1-4, to=1-5]
\end{tikzcd}
\end{equation}
be a short exact sequence of $C^*$-algebras.
    There are group homomorphisms $\exp : K_0 (B) \rightarrow K_1 (I)$ and $\mathrm{ind} : K_1 (B) \rightarrow K_0 (I)$ such that the six-term sequence
    $$\begin{tikzcd}
	{K_0(I)} && {K_0(A)} && {K_0(B)} \\
	\\
	{K_1(B)} && {K_1(A)} && {K_1(I)}
	\arrow["{\varphi_*}", from=1-1, to=1-3]
	\arrow["{\psi_*}", from=1-3, to=1-5]
	\arrow["{\exp}", from=1-5, to=3-5]
	\arrow["{\psi_*}", from=3-3, to=3-1]
	\arrow["{\varphi_*}", from=3-5, to=3-3]
    \arrow["{\mathrm{ind}}", from=3-1, to=1-1]
\end{tikzcd}$$
    of Abelian groups is exact.
\end{theorem}
We will need the expression of the exponential map $\exp$. As we will need it only in this case, we suppose that the $C^*$-algebra $A$ (hence also $B$) in \eqref{SES} is unital. We extend $\varphi$ to a unital morphism $\bar{\varphi}:I^+\rightarrow A$, $\bar{\varphi}(x,\alpha)=\varphi(x)+\alpha 1$. 
Let $p$ be an orthogonal  projection in $Proj_n(B)$, and let $a$ be a self-adjoint element in $M_n(A)$ for which $\psi(a)=p$. Then, $\bar{\varphi}(u)=\exp(2\pi\text{i}a)$ for precisely one unitary element $u$ in $U_n(I^+)$, and $\exp([p]_0)=-[u]_1$.

Recall that the suspension $SA$ and cone $CA$ of a $C^*$-algebra $A$ are defined to be
    \begin{align*}
        SA &=\{f\in C([0,1], A)|f(0)=f(1)=0 \}\\
        CA &=\{f\in C([0,1], A)|f(0)=0 \}.
    \end{align*}
They are connected by the exact sequence 
$$\begin{tikzcd}
	0 & SA & CA & A & 0.
	\arrow[from=1-1, to=1-2]
	\arrow["i", from=1-2, to=1-3]
	\arrow["\text{ev}", from=1-3, to=1-4]
	\arrow[from=1-4, to=1-5]
\end{tikzcd}$$
where $i$ is the inclusion and $\text{ev}$ the evaluation at $0$. The cone $CA$ has trivial $K$-groups. The six-term exact sequence implies therefore that $\mathrm{ind}$ and $\exp$ are isomorphisms:
$$K_1(A)\cong K_0(SA),\quad K_0(A)\cong K_1(SA).$$
The second isomorphism is given by
\begin{equation}\label{Bottmap}
    K_0(A)\ni [p]_0-[s(p)]_0 \mapsto [(1_n-p)+\exp(2\pi\text{i}t)p]_1\in K_1(SA),
\end{equation}
for $p\in Proj_n(A^+)$ where $s:A^+\to\C$ is as above.

Another fundamental result which we will make use of is the six-term exact sequence of Pimsner and Voiculescu \cite{blackadar1998k}[Thm.~10.2.2] 
\begin{theorem}[Pimsner-Voiculescu Exact Sequence]\label{PVES_thm}
    Let $A$ be a $C^*$-algebra and $\alpha \in \text{Aut}(A)$. Then there is a cyclic six-term exact sequence
    \begin{equation}\label{PVES}
        \begin{tikzcd}
	{K_0(A)} && {K_0(A)} && {K_0(A\rtimes_\alpha\Z)} \\
	\\
	{K_1(A\rtimes_\alpha\Z)} && {K_1(A)} && {K_1(A)}
	\arrow["{1-\alpha_*}", from=1-1, to=1-3]
	\arrow["{\iota_*}", from=1-3, to=1-5]
	\arrow[from=1-5, to=3-5]
	\arrow["{\iota_*}", from=3-3, to=3-1]
	\arrow["{1-\alpha_*}", from=3-5, to=3-3]
    \arrow[from=3-1, to=1-1]
\end{tikzcd}
    \end{equation}
\end{theorem}
This sequence can be split into two exact sequences, $i=0,1$
\begin{equation}\label{eq-PV}
0\to \mathrm{Coinv}_{\alpha_*}K_i(A) \to
K_i(A\rtimes_\alpha\Z) \to
\mathrm{Inv}_{\alpha_*} K_{i-1}(A) \to  0
\end{equation}
where
$\mathrm{Coinv}_{\alpha_*}(K_i(A)) = \mathrm{coker}(\alpha_*-\mathrm{id})=K_i(A)/\mathrm{Im}(\alpha_*-\mathrm{id})$ and 
$    \mathrm{Inv}_{\alpha_*}(K_i(A)) = \mathrm{ker}(\alpha_*-\mathrm{id})
$.

\section{Augmentation and bulk-edge correspondence}
In this section we describe the general setup in the framework of $K$-theory which we use to relate the integrated density of states to other topological invariants which, in the one-dimensional situation, manifest themselves in spectral flows. 
As already observed in \cite{kellendonk_prodan}, the simple approach to the bulk edge correspondence, which is based on the Toeplitz extension of the bulk algebra, yields a trivial correspondence if the hull is totally disconnected, as the $K_1$-group of the edge algebra is trivial. The process of augmentation remedies this. We discuss in this section the general framework and specify in the following section to the two cases of particular interest which we already mentioned in the introduction.

Consider a crossed product $C^*$-algebra $A\rtimes_\alpha\Z$. As we have seen, such algebras arise in the context of aperiodic models with $A=C(\Xi)$, where $\Xi$ is the hull, but the added generality will be useful for future work, for instance on higher dimensional systems. The group $K_0(A\rtimes_\alpha\Z)$ then corresponds to the set of topological phases of the bulk of the system. By an augmentation of the system we mean an extension  
of $A$, that is, a $C^*$-algebra $\tilde{A}$ together with a surjective morphism $q: \tilde{A}\to A$. This is usually summarized in a short exact sequence (SES)
\begin{equation}\label{eq-SES1}
0 \rightarrow I \xrightarrow{i}  \tilde{A} \xrightarrow{q} A \rightarrow 0.
\end{equation}
where $I$ is the ideal given by $\ker q$ ($i$ is just the inclusion map). When we focus on $A=C(\Xi)$ then the augmentation comes from an augmented hull $\tilde\Xi$. We suppose that also the action $\alpha$ extends to $\tilde A$ (denoted by the same symbol $\alpha$ or, if more clarity is needed, by $\tilde\alpha$) in such a way that it preserves the ideal $I$. We then obtain an exact sequence of crossed product algebras. 
\begin{equation}\label{eqn:firstSES}
	0\rightarrow I\rtimes_\alpha \Z \xrightarrow{i}  \tilde{A}\rtimes_\alpha \Z \xrightarrow{q} A\rtimes_\alpha \Z \rightarrow0.
\end{equation}

On the other hand, the Toeplitz extension of $\tilde{A}\rtimes_\alpha \Z $ provides us two more exact sequences. 
The short exact sequence
\begin{equation}\label{eqn:secondSES}
	0 \rightarrow \tilde{A} \otimes \mathcal{K} \xrightarrow{\b} \mathcal{T}(\tilde{A}, \alpha) \xrightarrow{\pi} \tilde{A} \rtimes_\alpha \mathbb{Z} \rightarrow 0
\end{equation}
and the augmented exact sequence which arises if we
compose the two surjective maps, $\mathcal{T}(\tilde{A}, \alpha) \stackrel{\pi}\rightarrow \tilde{A}\rtimes_\alpha \Z \stackrel{q}\rightarrow A\rtimes_\alpha \Z $
\begin{equation}\label{eqn:thirdSES}
	0 \rightarrow J \xrightarrow{j} \mathcal{T}(\tilde{A}, \alpha) \xrightarrow{\tilde{\pi}} A \rtimes_\alpha \mathbb{Z} \rightarrow 0
\end{equation}
where $J$ is by definition the kernel of $\tilde{\pi}\coloneqq q \circ \pi$. 
We get the following diagram. 
\begin{equation}\label{eqn:diagram}
	\begin{tikzcd}
		&&&&&& 0 \\
		\\
		&& 0 && 0 && {I \rtimes_\alpha \mathbb{Z}} \\
		\\
		0 && {\tilde{A} \otimes \mathcal{K}} && {\mathcal{T}(\tilde{A}, \alpha)} && {\tilde{A} \rtimes_\alpha \mathbb{Z}} && 0 \\
		&&&&  \\
		0 && J && {\mathcal{T}(\tilde{A}, \alpha)} && {A \rtimes_\alpha \mathbb{Z}} && 0 \\
		\\
		&& {I \rtimes_\alpha \mathbb{Z}} && {0} && 0 \\
		\\
		&& 0
		\arrow[from=1-7, to=3-7]
		\arrow[from=3-3, to=5-3]
		\arrow[from=3-5, to=5-5]
		\arrow["i", from=3-7, to=5-7]
		\arrow[from=5-1, to=5-3]
		\arrow["\b", from=5-3, to=5-5]
		\arrow["\b", from=5-3, to=7-3]
		\arrow["\pi", from=5-5, to=5-7]
		\arrow[from=5-7, to=5-9]
		\arrow["q", from=5-7, to=7-7]
		\arrow[from=7-1, to=7-3]
		\arrow["j", from=7-3, to=7-5]
		\arrow["\pi", from=7-3, to=9-3]
		\arrow["{\tilde{\pi}}", from=7-5, to=7-7]
		\arrow[from=7-5, to=9-5]
		\arrow[from=7-7, to=7-9]
		\arrow[from=7-7, to=9-7]
		\arrow[from=9-3, to=11-3]
		\arrow[equal, from=5-5, to=7-5]
	\end{tikzcd}
\end{equation}
The first (upper) horizontal sequence is the SES (\ref{eqn:secondSES}), the second is the SES (\ref{eqn:thirdSES}) and the right vertical sequence is (\ref{eqn:firstSES}). 
The left vertical SES follows from the snake lemma, which says that $\mathrm{coker}(\b)\cong I \rtimes_\alpha \mathbb{Z}$. 

The algebra $A\rtimes_\alpha \mathbb{Z}$ describes the bulk, while $J$ is the analogue of the edge algebra in the simple bulk edge correspondence. Indeed, if $\tilde A = A$ then $J=A$ and the diagram reduces to the Toeplitz exact sequence of $A\rtimes_\alpha\Z$ which is the exact sequence underlying the simple bulk edge correspondence. However, if $A=C(\Xi)$ with totally disconnected $\Xi$, which is the case for systems with finite local complexity, then $K_1(A)=0$. With a proper augmentation $J$ is, of course, no longer the same as $A$. It is the algebra whose $K_1$-group hosts the topological invariants to which we want to relate the bulk invariants.
We will see that, in contrast to the case of \cite{kellendonk_prodan}, $J$ can contain invariants which are of bulk type. 

We apply the $K$-functor to obtain the following diagram which also includes the exponential maps arising from the short exact sequences and
further maps which we will explain below.
\begin{equation}\label{eqn:focus_on_K-diagram}
	\begin{tikzcd}
		&& {K_0(I \rtimes_\alpha \mathbb{Z})}  \\
		\\
		{K_0(\mathcal{T}(\tilde{A}, \alpha))} && {K_0(\tilde{A} \rtimes_\alpha \mathbb{Z})} && {K_1(\tilde{A})} && {\mathbb{C}} \\
		\\
		{K_0(\mathcal{T}(\tilde{A}, \alpha))} && {K_0(A \rtimes_\alpha \mathbb{Z})} && {K_1(J)} \\
		\\
		\C && {K_1(I \rtimes_\alpha \mathbb{Z})} && {K_1(I \rtimes_\alpha \mathbb{Z})}  && {\mathbb{C}} 
		\arrow["{i_*}", from=1-3, to=3-3]
		\arrow["{\mathrm{exp}}", from=3-3, to=3-5]
		\arrow["{q_*}", from=3-3, to=5-3]
		\arrow[dashed, "{\che}", from=3-5, to=3-7]
		\arrow["{\b_*}", from=3-5, to=5-5]
		\arrow["{\mathrm{exp}_J}", from=5-3, to=5-5]
		\arrow[color={rgb,10:red,0;green,0;blue,10}, "{\beta}", from=5-3, to=3-5]
		\arrow[color={rgb,10:red,10;green,0;blue,0},"{\gamma}"', from=5-3, to=7-5]
		\arrow["{\hat{\mathrm{exp}}}", from=5-3, to=7-3]
		\arrow["{\pi_*}", from=5-5, to=7-5]
		\arrow["s"', bend right=18, from=7-5, to=5-3]
		\arrow[equal,  from=7-3, to=7-5]
		\arrow[dashed, "{\chab}", from=7-5, to=7-7]
		\arrow[dashed, "{\chb}", from=5-3, to=7-1]
		\arrow["{\tilde{\pi}_*}", from=5-1, to=5-3]
		\arrow[equal, from=3-1, to=5-1] 
		\arrow["{\pi_*}", from=3-1, to=3-3]
	\end{tikzcd}
\end{equation}
An abstract insulator is a self-adjoint element $h$ of $A\rtimes_\alpha\Z$ (or, if it is matrix valued, of $M_N(A\rtimes_\alpha\Z)$) with a gap at the Fermi energy $E_F$. Its Fermi projection $P_F(h)$ is the projection onto the states below that energy. It is important to note that $P_F(h)$ is a continuous function of $h$. More precisely, by  
a density with support in a bounded interval $\Delta$ we mean a positive continuous function $\varphi:\R\to \R^+$ with support in $\Delta$ and $\int_{-\infty}^{+\infty}\varphi(x) \mathrm{d} x = 1$. Let
\begin{equation}\label{eq-g}
   g(t)=\int_{-\infty}^t\varphi(x)\mathrm{d}x.
\end{equation}   
Then $P_F(h) = 1 - g(h)$. 
The $K_0$-class $[P_F(h)]_0$ describes the topological phase. This class is mapped  
by $\exp_J$ to an element of $K_1(J)$ which we want to relate to spectral flow. $K_1(J)$ itself is part of the right vertical exact sequence in the above diagram.
Our plan is therefore to split the elements of the image of $\exp_J$ into two parts, which correspond to the images of the two maps
\begin{align}\label{beta}
	\begin{split}
		\beta:K_0(A \rtimes_\alpha \mathbb{Z}) \cap \mathrm{ker}(\hat{\mathrm{exp}}) & \rightarrow K_1(\tilde{A})/\mathrm{exp}\circ i_*(K_0(I \rtimes_\alpha \mathbb{Z}))\\
		x &\mapsto \b_*^{-1}\circ \mathrm{exp}_J(x)=\mathrm{exp}\circ q_*^{-1}(x) 
	\end{split} \\
	\begin{split}
		\gamma:K_0(A \rtimes_\alpha \mathbb{Z})  & \rightarrow K_1(I \rtimes_\alpha \mathbb{Z}))\\
		x&\mapsto \pi_*\circ \mathrm{exp}_J(x)=\hat \exp (x) 
	\end{split}
\end{align}
Note that, to compute $\beta$ and $\gamma$ on the $K_0$-class $[P_F(h)]_0$ defined by $h$ we need to look for a pre-image $\tilde h$ of $h$ under $q$. We call such a pre-image $\tilde h$ an augmentation of $h$. Furthermore, assuming that $E_F<1$, we may also take a pre-image of $1\oplus h \in M_2(A\rtimes_\alpha\Z))$, because $P_F(1\oplus h) = 0 \oplus P_F(h)$ and this defines the same $K_0$-class as $P_F(h)$. Through the choice of pre-image $\beta$ is defined only up to an element in $\mathrm{exp}\circ i_*(K_0(I \rtimes_\alpha \mathbb{Z}))$. This introduces an ambiguity. 
Moreover, we choose a section $s$ for $\gamma$, that is, a group morphism $s:\mathrm{im}\gamma\to K_0(A \rtimes_\alpha \mathbb{Z})$
which satisfies $\gamma\circ s=\mathrm{Id}_{\mathrm{Im}\gamma}$. We will see that such a section exists indeed in our models, but in general we need to assume its existence. 

The next step is to define so-called Chern-cocycles
$\chb$, $\che$ and $\chab$ which give rise to morphisms from the $K$-groups of interest into $\C$ via Connes pairing \cite{connes1994noncommutative,prodan2016bulk}, and which we can give a physical interpretion. The superscripts $b$, $e$ and $ab$ mean \textit{bulk}, \textit{edge} and \textit{augmented bulk} respectively. Our attention lies on the middle vertical sequence
\begin{equation}\label{eqn:K-vertical}
\begin{matrix}
K_0(I \rtimes_\alpha \mathbb{Z}) & \stackrel{i_*}\to &
 K_0(\tilde A \rtimes_\alpha \mathbb{Z}) &  \stackrel{q_*}\to &
K_0(A \rtimes_\alpha \mathbb{Z}) &  \stackrel{\hat\exp}\to & 
K_1(I \rtimes_\alpha \mathbb{Z})\\
 & &\:\:\:\:\:\:\:\:\: \downarrow \che\!\circ\exp & &\:\:\:\:\:\: \downarrow \chb & &\:\:\:\:\:\:\:\:\: \downarrow \chab  \\
 & &\!\!\!\!\!\!\!\!\!\!\!\!\!\!\!\!\!\!\!\!\!\C  & & \!\!\!\!\!\!\!\!\!\!\!\!\!\C  & & \!\!\!\!\!\!\!\!\!\!\!\!\C
\end{matrix}
\end{equation}
with the aim to establish a relation of the form, for $x\in K_0(A \rtimes_\alpha \mathbb{Z})$, 
\begin{equation}\label{theFormula}
	\chb(x) \equiv \che(\tilde\beta(x))-\chab(\gamma(x))\mod{\mathrm{Amb}_e+\mathrm{Amb}_b}
\end{equation}
where 
$$\tilde\beta  = \beta\circ (\mathrm{id}-s\circ \gamma) $$
and
\begin{align*}
	\mathrm{Amb}_e &:= \che(\mathrm{exp}\circ i_*(K_0(I \rtimes_\alpha \mathbb{Z})))\\
	\mathrm{Amb}_b &:= \chb(\tilde{\pi}_*(K_0(\mathcal{T}(\tilde{A}, \alpha))))
\end{align*}
are what we call the ambiguity groups of the edge and of the bulk respectively. The first ambiguity comes from the definition of $\beta$. The second one is due to the fact that $\beta$ and $\gamma$ do not see the elements $\tilde{\pi}_*(K_0(\mathcal{T}(\tilde{A}, \alpha)))$, as the two horizontal lines in (\ref{eqn:focus_on_K-diagram}) are exact.  
The formula says that, up to an ambiguity, the bulk numerical invariant
$\chb(x)$--which will have the interpretation as value of the IDS if $x$ is the class of the Fermi projection--can be computed as a sum of two other numerical invariants which are obtained from a direct sum decomposition of $\exp_J(x)\in K_1(J)$. Indeed, any element $w$ in the image of $\exp_J$ can be uniquely written as a sum $w=y+z$ with $y\in \b_*(K_1(\tilde A))$ and $z\in s'(K_1(I\rtimes_\alpha\Z)\cap \mathrm{im} \gamma)$, where $s' = \exp_J\circ s$.
We then use a cocycle $\che$ on $\tilde A$ and a cocycle 
$\chab$ on $K_1(I\rtimes_\alpha\Z)$ to pair with the two summands in the decomposition of $\exp_J(x)$.

\subsection{Chern cocycles}
In general, a Chern cocycle of degree $n$ over a $C^*$-algebra $B$ can be defined with the help of $n$ commuting derivations $\delta_1,\dots, \delta_n$ with domains $\mathcal{D}_{\delta_1}, \dots \mathcal{D}_{\delta_n}$ and a positive trace $\tau$ with domain $\mathcal{D}_{\tau}$ satisfying the following properties. 
\begin{enumerate}
	\item The domains $\mathcal{D}_{\delta_1},\dots \mathcal{D}_{\delta_n}$ are dense sub-algebras of $B$.
	\item $\mathcal{D}_{\tau}$ is an ideal of $B$.
	\item $\tau$ is cyclic, i.e. $\tau(ab)=\tau(ba)$, for all $a\in \mathcal{D}_{\tau}$ and $b\in B$.
	\item $\lvert \tau(ab) \rvert \le \lVert b\rVert \tau(\lvert a\rvert)$, for all $a\in \mathcal{D}_{\tau}$ and $b\in B$.
	\item The set $\{a\in A|\delta_i(a)\in \mathcal{D}_{\tau}\}$ is a dense sub-algebra.
	\item $\tau\circ \delta_i(a)=0$ for all $a\in \mathcal{D}_{\delta_i}$ such that $\delta_i(a)\in \mathcal{D}_{\tau}$.
	\item There exists a dense Fréchet sub-algebra $\mathcal B$ which is closed under holomorphic functional calculus and such that the inclusion $i:\mathcal{B}\rightarrow A $ is continuous and 
	$$\mathcal{B}^{n+1}\ni (a_0,\dots a_n)\xmapsto{\eta} \sum_{\sigma\in \mathrm{Perm}_n}\mathrm{sgn}(\sigma)\tau(a_0\delta_{\sigma(1)}(a_1)\dots\delta_{\sigma(n)}(a_1))$$
	is well defined.    
\end{enumerate}
Then, we can define the Chern cocycles of degree $n$, $\mathrm{ch}_n:K_n(B)\rightarrow\C$. If $n$ is even $\mathrm{ch}_n([p]_0)=c_n \eta(p,\dots, p)$,
while if $n$ odd $\mathrm{ch}_n([u]_1)=c_n \eta((u^*-1),u,u^*,\dots, u)$, where $c_n\in\C$ are constants. In these formulas we need to evaluate the chern cocycle on a representative which belongs to $\mathcal B$; there is always such a representative as the $K$-groups of $\mathcal B$ and $B$ are isomorphic. 
In our case, we will see that $\mathcal{B}$ is a Banach algebra such that $\lVert \cdot\rVert_\mathcal{B}\ge \lVert \cdot\rVert_B$, so that the inclusion is continuous and stable by holomorphic functional calculus.

We apply this to $B=A\rtimes_\alpha \Z$. 
Given an $\alpha$-invariant trace $\tau_A$ on $A$ we obtain a trace $\hat{\tau}_A$ on (a dense subalgebra of) $A\rtimes_\alpha \Z$ via
\begin{equation}\label{tau_hat}
	\hat{\tau}_A(au^n)=\delta_{0n}\tau_A(a), \quad a\in A.
\end{equation}
It also induces a trace on $M_m(A\rtimes_\alpha \Z)$ (for an arbitrary $m$) simply by composing with the standard trace for $m\times m$ matrices $\mathrm{Tr}_m$. To simplify notations, the composition with $\mathrm{Tr}_m$ remains implicit.

We define $\chb$ on $ K_0(A\rtimes_\alpha \Z)$ as the $0$-cocycle $\ch_0$ defined by the trace $\hat \tau_A$:
\begin{equation}\label{chern_bulk}
	\chb([p]_0)=\hat{\tau}_A(p).
\end{equation}

Let us consider now $1$-cocycles $\ch_1$ defined by a trace $\tau$ and derivation $\delta$ satisfying the properties (1)-(6) above. The following result will be useful to interprete the cocycles $\che$ and $\chab$  for the models under consideration. 
\begin{prop}\label{useful_formula_for_trace}
	Let $B$ be a $C^*$-algebra, $\tau$ a positive trace on $B$ and 
$\delta$ be a derivation on $B$ satisfying the properties (1)-(6) above. 
	Let $\varphi$ be a density with support in a bounded interval $\Delta$ and $g$ as in \eqref{eq-g}. We choose a self-adjoint element $b\in B$ and we set $U=\mathrm{e}^{2\pi\mathrm{i}g(b)}$. If $U$ is in $\mathcal{D}_\delta$ and $(U^*-1)\delta U$ is in $\mathcal{D}_\tau$, then
	$$\tau((U^*-1) \delta U)=-\tau(\varphi(b)\delta b).$$

\end{prop}
\begin{proof}
	The statement is a slight generalisation of that of proposition 7.1.3 of \cite{prodan2016bulk} since our hypothesis is a little bit weaker, but the proof follows that same lines. 
	
	It is well known that the map $U\mapsto \nu(U):=\tau((U^*-1) \delta U)$ is invariant under homotopies defined by differentiable paths $t\mapsto U(t)$ where $U(t)$ is a unitary in $B$ (if necessary with a unit attached) and such that $\delta U(t)\in \mathcal{D}_\delta$ and $(U(t)^*-1)\delta U(t)\in \mathcal{D}_\tau$. Extending $\tau$ and $\delta$ to matrices over $B$ in the standard way one can use that $\nu(U^k) = k \nu(U)$ (Whitehead lemma).
	
	Let $b$ be a self-adjoint element such that $U$ satisfies the hypothesis of the theorem.
	By Duhamel formula, we have
	$$\delta U^k=2\pi\mathrm{i}k\int_0^1U^{(1-s)k}\delta g(b)U^{sk}\mathrm{d}s$$
	and by cyclicity of $\tau$, we get
	$$\nu(U)=\tau((1-U^k)\delta g(b)).$$
	Let $\tilde g$ be the Fourier transform of the function $E\mapsto (1-g(E))\mathrm{e}^E$ (which is well-defined as the function is exponentially decaying), hence
	$$(1-g(E))\mathrm{e}^E=-\frac{1}{2\pi}\int_{-\infty}^\infty \tilde{g}(-t)\mathrm{e}^{-\mathrm{i}Et}\mathrm{d}t$$
	and therefore
	$$g'(E)=\frac{1}{2\pi}\int_{-\infty}^\infty \tilde{g}(-t)\frac{\mathrm{d}}{\mathrm{d}E}\mathrm{e}^{-(1+\mathrm{i}t)E}\mathrm{d}t=-\frac{1}{2\pi}\int_{-\infty}^\infty \tilde{g}(-t)(1+\mathrm{i}t)\mathrm{e}^{-(1+\mathrm{i}t)E}\mathrm{d}t.$$
	Replacing $\delta g(b)$, by the Duhamel formula and cyclicity of the trace we get
	\begin{align*}
		\begin{split}
			\nu(U)&=\frac{1}{2\pi}\int_{-\infty}^\infty \tilde{g}(-t)\tau\left((1-U^k)\delta \mathrm{e}^{-(1+\mathrm{i}t)b}\right)\mathrm{d}t\\
			&=\frac{1}{2\pi}\int_{-\infty}^\infty \tilde{g}(-t)(1+\mathrm{i}t)\tau\left((U^k-1) \int_0^1 \mathrm{e}^{-q(1+\mathrm{i}t)b}\delta b\mathrm{e}^{-(1-q)(1+\mathrm{i}t)b}\mathrm{d}q\right)\mathrm{d}t\\
			&= \tau\left(\frac{1}{2\pi}\int_{-\infty}^\infty \tilde{g}(-t)(1+\mathrm{i}t)\mathrm{e}^{-(1+\mathrm{i}t)b}(U^k-1)\mathrm{d}t\delta b \right)\\
			&= \tau\left((1-U^k)g'(b)\delta b \right)\\
			&= \tau\left((1-U^k)\varphi(b)\delta b \right)
		\end{split}
	\end{align*}
	where we used the fact that $U^k-1$ commutes with functions in the variable $b$.
	Now, let $\phi$ be a function in $C^2([0,1])$ such that the support of $\phi$ is contained in $]0,1[$
	and let $\{c_k\}_{k\in \Z}$ be the Fourier coefficients of $\phi$, namely
	$$S_\phi^{(N)}(x)=\sum_{\lvert n\rvert \le N}c_k\mathrm{e}^{-2\pi\mathrm{i}kx}, \quad c_k=\int_0^1\phi(t)\mathrm{e}^{2\pi\mathrm{i}kt}\mathrm{d}t,$$
	as $\phi$  is $C^2$, $S_\phi^{(N)}$ converges normally to $\phi$. As $\phi(0)=0$,
	$\sum_{k\in\Z}c_k=0$ hence $\sum_{k\in\Z\setminus\{0\}}c_k=-c_0$.
	Therefore,
	\begin{align*}
		\begin{split}
			c_0\nu(U)&=-\sum_{k\in\Z\setminus\{0\}}c_k\nu(U)\\
			&= \sum_{k\in\Z}c_k \tau\left((U^k-1)\varphi(b)\delta b \right)\\
			&= \left(\sum_{k\in\Z}c_k\right) \tau\left(\varphi(b)\delta b \right)-  \tau\left(\sum_{k\in\Z}\left(c_k\mathrm{e}^{2\pi\mathrm{i}kg(b)}\right)\varphi(b)\delta b \right)\\
			&= -\tau\left(\lim_{N\rightarrow\infty} S_\phi^{(N)}\circ g (b)\varphi(b)\delta b\right)\\ 
			&= -\tau\left(\phi\circ g (b)\varphi(b)\delta b\right).
		\end{split}
	\end{align*}
	Finally, we let $\phi$ converge to the indicator function of $[0, 1]$. Then $c_0 \rightarrow 1$, while on the other hand $\phi\circ g (b)\varphi(b) \rightarrow\varphi(b)$ (the Gibbs phenomenon is damped). This concludes the proof.
\end{proof}
\subsection{Interpretation of $\chb$ for aperiodic chains}\label{sec-IDS}
In the context of our physical models, when $P_F(h)$ is the Fermi projection of an abstract insulator $h$ whose Fermi energy lies in the gap, the number $\chb([P_F(h)]_0)$ has a well-known interpretation as integrated density of states up to the Fermi energy \cite{Bellissard1992}. 
We recall some details.

Let $\Xi$ be the hull of the system with its $\Z$-action $\alpha$. 
To obtain a formula for $\chb([P_F(h)]_0)$ which involves the Hamiltonians $H_\xi$, $\xi\in\Xi$, which arise through the representations $\rho_\xi$ of $C(\Xi)\rtimes_\alpha\Z$, we express $\chb$ in such a representation. For that we assume that the measure $\mu$ is ergodic w.r.t.\ to the action $\alpha$. Such measures always exist, and are sometimes even unique. This is for instance the case for our system with one or two cuts when the rotation angle is irrational. 

By Birkhoff's theorem, for $\mu$-almost all $\xi_0\in\Xi$ and $b = f u^n$, $f\in C(\Xi)$, we have 
$$\hat\tau_A(b)=\lim_{L\rightarrow \infty}\frac{1}{2L+1}\delta_{n0} \sum_{-L\le n\le L}f(\xi_0-n\theta) = \lim_{L\rightarrow \infty}\frac{\mathrm{Tr}(\rho_{\xi_0}(b)|_{\Lambda_L})}{2L+1}$$
where $\Lambda_L\subset \ell^2(\Z)$ are the $2L+1$ sites around the origin. The limit on the r.h.s.\ is called the trace per unit volume of $\rho_{\xi_0}(b)$.

Let $\Delta$ be a gap in the spectrum of $h$. If the Fermi energy $E_F$ is in the gap $\Delta$, then 
$P_F(H_\xi):=1-g(H_\xi)$ is the Fermi projection of $H_\xi$, where $g$ is as in \eqref{eq-g}.  
As $P_F(H_\xi) = \rho_\xi(P_F(h))$ we obtain 
$$\hat{\tau}_A(P_F(h))=\lim_{L\rightarrow \infty}\frac{\mathrm{Tr}(P_F(H_\xi)|_{\Lambda_L})}{2L+1}=\lim_{L\rightarrow \infty} \frac{\mathrm{Tr}(P_F(H_\xi|_{\Lambda_L}))}{2L+1}$$
where the second equality is Shubin's formula which is valid, as $h$ can be approximated by an operator of finite interaction range. 
Therefore, $\hat{\tau}_A(P_F(h))$ corresponds for $\mu$-almost $\xi$ to the number of states of $H_\xi$ below the Fermi energy $E_F$ per unit volume. This quantity is also referred to as the Integrated Density of States at the Fermi level and denoted by $\mathrm{IDS}(E_F)$. Thus
\begin{equation}\label{eq-int-chb}
\chb([P_F(h)]_0) = \mathrm{IDS}(E_F)
\end{equation}
\section{Augmentation via the mapping torus}
We discuss here an augmentation procedure which is based on the mapping torus construction. Applied to $A=C(\Xi)$ we obtain a tight binding version of a bulk edge  correspondence which is known for models described by differential operators \cite{ZoisKellendonk2005rotation}: the correspondence relates the integrated density to the force induced on the edge states by moving the boundary of the system. In one dimension it manifests itself through a spectral flow of edge states which is induced by the moving of the boundary. Upon using a particular choice of augmentation we obtain an alternative proof of the equality between Bellissard's and Johnson's gap labelling groups.

Given a $C^*$-algebra $A$ with an action $\alpha$ of $\Z$, its mapping torus is the algebra
$$M_\alpha A=\{f\in C([0,1], A)|f(1)=\alpha(f(0)) \}.$$
The sequence
\begin{equation*}
	0 \rightarrow SA \xrightarrow{i}  M_\alpha A \xrightarrow{q} A \rightarrow 0.
\end{equation*}
is exact when we consider $q: f\mapsto f(0)$. The action $\alpha$ induces an action $\beta$ of $\R$ on $M_\alpha A$ defined as
$$(\beta^s(f))(t)=\alpha^{\lfloor t+s\rfloor}(f(\{t+s\})), \quad \forall s\in\R, \forall t\in [0,1].$$
By restriction of $s$ to $\Z$ we get the action $\tilde{\alpha}$ of $\Z$ on $M_\alpha A$. 

\begin{lemma}\label{lem-split-exact}
	The boundary maps of the central vertical sequence of \eqref{eqn:focus_on_K-diagram} are zero. 
\end{lemma}
\begin{proof}
	We easily check that we have the isomorphism $M_\alpha A \rtimes_{\tilde{\alpha}} \mathbb{Z}\cong M_{\bar{\alpha}}(A \rtimes_\alpha \mathbb{Z}) $ where $\bar{\alpha}$ is given by $\bar{\alpha}(f(t) u^n)=\alpha(f(t))u^n= u f(t) u^{n-1} $ and therefore $\bar{\alpha}=\mathrm{Ad}_u$. Hence,
	$$[x]=[uxu^*]=[\bar{\alpha}(x)]=\bar{\alpha}_*[x],$$
	where the equalities are in $K_i(A \rtimes_\alpha \mathbb{Z})$. In particular, $\bar\alpha$ induces the identity on $K_i(M_\alpha A)$. From  (\cite{blackadar1998k} proposition 10.4.1), we have that the boundary maps $\partial_i:K_i(A\rtimes_\alpha \Z) \rightarrow K_i(A\rtimes_\alpha \Z)$ are of the form $1-\bar{\alpha}_*=0$. Hence $\hat{\exp}=0$. 
\end{proof}

It follows that $\gamma=0$ so that $\chab$ is irrelevant and (\ref{eqn:K-vertical}) simplifies for $\tilde A = M_\alpha A$ as follows 
\begin{equation}\label{eqn:K-vertical_mappingTorus}
\begin{matrix}
K_0(SA \rtimes_\alpha \mathbb{Z}) & \stackrel{i_*}\to &
 K_0(M_{\alpha}(A) \rtimes_{\tilde\alpha} \mathbb{Z}) &  \stackrel{q_*}\to &
K_0(A \rtimes_\alpha \mathbb{Z}) &  \stackrel{0}\to & 
K_1(SA \rtimes_\alpha \mathbb{Z})\\
 & &\:\:\:\:\:\:\:\:\: \downarrow \che\!\circ\exp & &\:\:\:\:\:\: \downarrow \chb & &  \\
 & &\!\!\!\!\!\!\!\!\!\!\!\!\!\!\!\!\!\!\!\!\!\C  & & \!\!\!\!\!\!\!\!\!\!\!\!\!\C  & & 
 \end{matrix}
\end{equation}

\subsection{Chern cocycles}
We now turn our attention to the Chern cocycles. Let $\tau_A$ be an $\alpha$-invariant trace on $A$. Its extension $\hat{\tau}_A$ (see \eqref{tau_hat}) defines $\chb$ as in \eqref{chern_bulk}. 
It also extends to 
a trace $\tau_{\tilde{A}}$ on $M_\alpha A$:  for $f\in M_\alpha A$
$$\tau_{\tilde{A}}(f)=\int_0^1 \tau_A(f(t))\mathrm{d}t$$
and $\mathcal{D}_\tau=M_\alpha A$. Define the derivation $\delta_1$  
$$\delta_1(f)=\dot f  =\partial_t f$$
on the domain $\mathcal{D}_\delta $ of functions $f \in M_\alpha A$ for which $\delta_1(f)\in M_\alpha A$. Then $\mathcal{D}_\delta $ is a dense subalgebra of $M_\alpha A$ which is a Banach algebra w.r.t.\ the norm
$$\lVert f\rVert_\mathcal{A}= \sup_t \lVert f(t)\rVert_A+ \sup_t \lVert \partial_t f(t)\rVert_A.$$
We can verify the hypothesis on $\mathcal{D}_\tau$ and $\mathcal{D}_\delta$ so that
$$ \che([u]_1)= \frac{1}{2\pi\mathrm{i}}\tau_{\tilde{A}}((u^*-1)\partial_t u)  $$
is a well defined $1$-cocycle. Extend $\delta_1$ to the crossed product $M_\alpha A\rtimes_{\tilde\alpha}\Z\cong M_{\bar{\alpha}}(A \rtimes_\alpha \mathbb{Z})$ through $\delta_1 u = 0$ and define a second derivation on (a dense subalgebra of) that algebra through
	$$\delta_2(fu^n)= (fu^n)' = \mathrm{i}nfu^n.$$
The following result is a special case of a general result for crossed product algebras \cite{nest1988cyclic,elliott1988cyclic}. A proof can also be found in \cite{KelRichBal,prodan2016bulk}. It makes use of the homotopy invariance of cyclic cocycles which allows to reduce it to a direct calculation.
\begin{prop}\label{prop-ENN}
Let $P$ be a projection in $M_\alpha A\rtimes_{\tilde\alpha}\Z$ and $U$ be a representative for the class $\exp([p]_0)\in K_1(M_\alpha A)$ then
	\begin{equation}\label{winding_number_b}
	\frac{1}{2\pi\mathrm{i}}\tau_{\tilde{A}}((U^*-1)\partial_t U) 
=\frac1{\mathrm{i}}\hat{\tau}_{\tilde{A}}(P[\delta_1(P), \delta_2(P)]).
	\end{equation}
\end{prop}
\renewcommand{\sc}{\sigma}
It is a consequence of Lemma~\ref{lem-split-exact} that $q_*$ is surjective. Moreover, we can define a right inverse $\sc:K_0(A\rtimes_\alpha \Z)\to K_0(M_\alpha A\rtimes_{\tilde\alpha}\Z)$ to $q_*$. Given a projection $p\in A \rtimes_\alpha \mathbb{Z}$ let 
\begin{equation}\label{eq-sc}
\sc[p]_0 = [P]_0
\end{equation} 
where
$P:[0,1]\rightarrow M_2(A \rtimes_\alpha \mathbb{Z})$ is given by 
	$$P(t)=U_t \begin{pmatrix}
		0 & 0\\
		0 & p
	\end{pmatrix} U_t^*$$
	where 
	$$U_t=R_t\begin{pmatrix}
		u & 0\\
		0 & 1
	\end{pmatrix}R_t^*,\quad R_t=\begin{pmatrix}
		\cos{(\frac{\pi}{2}t)} & -\sin{(\frac{\pi}{2}t)}\\
		\sin{(\frac{\pi}{2}t)} & \cos{(\frac{\pi}{2}t)}
	\end{pmatrix}.$$
Clearly $P$ is a projection in 
$M_2(M_{\bar{\alpha}} (A \rtimes_\alpha \Z) )$ 
which satisfies $q(P)=0 \oplus p$. So $q_*\circ \sigma = \mathrm{id}$.

\begin{lemma}\label{lem-tr-ch}
Given a projection $p\in A \rtimes_\alpha \mathbb{Z}$ let $P$ be as above.
Then 
\begin{equation}\label{eq-tr-ch}
\frac1{\mathrm{i}}\hat{\tau}_{\tilde{A}}(P[\delta_1(P), \delta_2(P)])
=\hat{\tau}_A(p).
\end{equation}
\end{lemma}
\begin{proof} 
Clearly $P(0)=\begin{pmatrix}
		0 & 0\\
		0 & p
	\end{pmatrix}$ and $P(1)=\begin{pmatrix}
		0 & 0\\
		0 & \bar{\alpha}(p)
	\end{pmatrix}$ so that $P$ is a projection in $M_2(M_{\bar{\alpha}}(A \rtimes_\alpha \mathbb{Z}))$ and $q(P)=0 \oplus p$. 
	Denote $c_t=\cos{(\frac{\pi}{2}t)}$ and $s_t=\sin{(\frac{\pi}{2}t)}$. A direct calculation leads to (we write shorter $\dot f=\delta_1(f)$ and $f' = \delta_2(f)$): 
	\begin{align*} 
		\dot{P}_t &=  \frac{\pi}{2}U_t (A_1+A_2)U_t^*,\quad A_1=p\begin{pmatrix}
			0 & 1\\ 
			1 & 0
		\end{pmatrix}, \quad A_2=\begin{pmatrix}
			0 & -(c_t^2u^*+s_t^2u)p\\ 
			-p(c_t^2u+s_t^2u^*) & c_t s_t[p,u^*-u]
		\end{pmatrix}\\
		P'_t &= U_t (B_1+B_2)U_t^*,\quad B_1=p'\begin{pmatrix}
			0 & 0\\ 
			0 & 1
		\end{pmatrix}, \quad B_2=\mathrm{i}c_t s_t p\begin{pmatrix}
			0 & 1\\ 
			-1 & 0
		\end{pmatrix}.
	\end{align*}
	Hence
	$$\frac{1}{\mathrm{i}}\hat{\tau}_{\tilde{A}}(P[\dot{P}, P'])=\frac{\pi}{2\mathrm{i}}\int_0^1\hat{\tau}_A\left( U_t \begin{pmatrix}
		0 & 0\\ 0 & p
	\end{pmatrix} ([A_1,B_1]+ [A_1,B_2]+[A_2,B_1]+[A_2,B_2])U_t^*\right)\mathrm{d}t.$$
	We compute the four terms. We see easily that the first term is zero and second one is equal to $\hat{\tau}_A(p)$. For the third term, we get
	\begin{align*}
		\mathrm{3\textsuperscript{rd} term}&= \frac{\pi}{2\mathrm{i}}\int_0^1\hat{\tau}_A\left( \begin{pmatrix}
			0 & 0\\ 0 & p
		\end{pmatrix} [A_2,B_1]\right)\mathrm{d}t\\ &= 
		\frac{1}{2\mathrm{i}}\hat{\tau}_A(p(u^*-u)p'-p(u^*-u)pp'-pp'p(u^*-u)+p'(u^*-u)p)\\
		&= \frac{1}{2\mathrm{i}}\hat{\tau}_A((p(u^*-u)p)'-p(u^*-u)'p)\\
		&= \frac{1}{2}\hat{\tau}_A(p(u^*+u)p)),
	\end{align*}
	where we used the relation $pp'p=0$, the cyclicity of the trace and that it vanishes on a total derivative. 
	Finally
	$$\mathrm{4\textsuperscript{th} term}=-\frac{1}{2}\hat{\tau}_A(p(u^*+u)p)$$
	follows quickly. Thus only the second term contributes yielding the result.
\end{proof}
\begin{coro}\label{cor-chb=che}
Let $h$ be a self-adjoint element of $A\rtimes_\alpha\Z$ which has a gap in its spectrum around $0$. Let $P_-(h)$ be the projection onto its negative energy states. Define 
$$\tilde h(t) = U_t \begin{pmatrix}
		1 & 0\\
		0 & h
	\end{pmatrix} U_t^*.$$
Then $\tilde h$ is a self-adjoint element in $M_2(M_{\bar{\alpha}}(A \rtimes_\alpha \mathbb{Z}))$ which satisfies $q(\tilde h)=1 \oplus h$ and
$$\chb([P_-(h)]_0) = \che(\exp[P_-(\tilde h)]_0).$$
\end{coro} 
\begin{proof}
Observe that $P_-(\tilde h(t)) = U_t P_-(1\oplus h) U^*_t = U_t(0 \oplus P_-(h))U^*_t$ and apply 
Prop.~\ref{prop-ENN} and Lemma~\ref{lem-tr-ch}.
\end{proof}
\begin{lemma}\label{lem-amb}
We have
	$$\mathrm{Amb}_e=
	\{\tau_A(p)|[p]_0\in \mathrm{Inv}_{\alpha_*} (K_0(A))\}.$$
	Furthermore $\mathrm{Amb}_b\subset \mathrm{Amb}_e$ and 
	$\mathrm{Amb}_e\subset \mathrm{im}\exp \circ \sc$.
\end{lemma}
\begin{proof}
We focus on the first two lines of diagram 
(\ref{eqn:focus_on_K-diagram}) completing the first line with the exponential map $\mathrm{exp}_I$ coming from the Toeplitz extension exact sequence. 
	\begin{equation*}
		\begin{tikzcd}
			{K_0(SA \rtimes_\alpha \mathbb{Z})} && K_1(SA) && K_1(SA) \\
			\\
			{K_0(M_{\bar{\alpha}}(A \rtimes_\alpha \mathbb{Z}))} && {K_1(M_\alpha A)} && \C  
			\arrow["{i_*}", from=1-1, to=3-1]
			\arrow["{i_*}", from=1-3, to=3-3]
			\arrow["{\exp}", from=3-1, to=3-3]
			\arrow["{\mathrm{exp}_I}", from=1-1, to=1-3]
			\arrow[dashed, "{\che}", from=3-3, to=3-5]
			\arrow["{1-\alpha_*}", from=1-3, to=1-5]
		\end{tikzcd}
	\end{equation*}
	By commutativity of this diagram the ambiguity is given by 
	\begin{equation}
		\mathrm{Amb}_e= \{\che(i_*([u]_1)):[u]_1\in \mathrm{Inv}_{\alpha_*} (K_1(SA))\}.
	\end{equation}
	Using the Bott isomorphism \eqref{Bottmap},
	$$K_0(A)\ni [p]_0\mapsto [f_p:t\mapsto\exp(2\pi\mathrm{i}pt)]_1\in K_1(SA),$$
	which maps $\mathrm{Inv}_{\alpha_*} (K_0(A))$ to $\mathrm{Inv}_{\alpha_*} (K_1(SA))$, and
	$$\che(i_*([f_p]_1))=\frac{1}{2\pi\mathrm{i}}\int_0^1 \tau_A(\exp (-2\pi\mathrm{i}t)2\pi\mathrm{i}p \exp (2\pi\mathrm{i}t))\mathrm{d}t= \tau_A(p)$$
	we obtain the result.
	
	The group $\mathrm{Amb}_b$ is generated by values $\chb([p]_0)$ where $p=\pi(\tilde P)$ for some projection $\tilde P$ in $M_n(M_\alpha A)$ whose class defines an element of $\pi_*(K_0(\mathcal T(M_\alpha A,\tilde\alpha)))$. Let $P$ be the above lift of $p$, then $[\tilde P]_0-[P]_0\in \ker q_*$ and thus 
	$\chb([p]_0) = -\che(\exp([\tilde P]_0-[P]_0))\in \mathrm{Amb}_e$.
	
Finally, $\tau_A(p) = \che(\exp \circ \sc([p]_0))$ showing the last statement.
\end{proof}
The following well-known result is a byproduct of our approach. We include it here, as our proof is more explicit the its usual proof (sketched in \cite{Bellissard1992}) which is based on  
Connes' Thom isomorphism \cite{connes1981analogue} together with the result that $M_\alpha A \rtimes_\beta \R$ is Morita equivalent to $A \rtimes_\alpha \Z$. 
Specified to the algebra $A=C(\Xi)$ for some hull $\Xi$ of a one-dimensional aperiodic Schr\"odinger operator, it states that the $K$-theoretic gap labelling group of Bellissard coincides with the gap labelling group of Johnson. We will comment on that below.
\begin{coro}\label{cor-B=J}
We have $\chb(K_0(A\rtimes_\alpha\Z)) = \che(K_1(M_\alpha A))$.
\end{coro}
\begin{proof}
Let $f\in M_\alpha A$ be a unitary. Then $[0,1]\ni s\mapsto f(t+s)$ is a homotopy of unitaries between $f$ and $\tilde\alpha(f)$ showing that $\mathrm{id}-\tilde\alpha_*$ is the zero map on $K_1(M_\alpha A)$. It follows that $\exp:K_0(M_\alpha A\rtimes_{\tilde\alpha}\Z) \to K_1(M_\alpha A)$ is surjective. Hence $K_1(M_\alpha A)=\mathrm{im}\exp$. 
As $\che(\mathrm{im}\exp)=\che(\mathrm{im}\exp \circ \sc) + \mathrm{Amb}_e = \che(\mathrm{im}\exp \circ \sc)$ (the second equality follows from Lemma~\ref{lem-amb}) we deduce from Cor.~\ref{cor-chb=che} the result.
\end{proof}

The self-adjoint element $\tilde h$ used for Cor.~\ref{cor-chb=che} has two advantages. First, the spectrum of $h(t)$ is independent of $t$, and second, the identity of Cor,~\ref{cor-chb=che} does not involve an ambiguity. Its choice is nevertheless a bit special in the way the two stacked insulators $1$ and $h$ are coupled. From a physical point of view it is more natural to augment the system using an interpolation procedure. This comes with the price that we cannot guarantee that the spectrum of the augmented Hamiltonian does not fill the bulk gap, neither can we easily control the ambiguity. The first problem may (but need not) keep us from seeing the spectral flow (while the edge invariant is still well-defined, we loose only this interpretation). The second has the consequence that our result is only true up to an ambiguity. 
\begin{coro}[BEC for mapping torus, general case]\label{BBC_mapping_torus}
	Given $x\in K_0(A \rtimes_\alpha \mathbb{Z})$, then
	$$\chb(x)\equiv \che(\beta(x))
	\mod{	\{\tau_A(p)|[p]_0\in \mathrm{Inv}_{\alpha_*} (K_0(A))\}}.$$
\end{coro}
\begin{proof} The formula of Cor.~\ref{cor-chb=che} holds only for a special choice of augmentation. Any other choice of augmentation may change the value of its r.h.s.\ by an element of $\mathrm{Amb}_e$ which we computed in Lemma~\ref{lem-amb}.
\end{proof}
\subsection{The case $A=C(\Xi)$}
We specialise the previous construction to the case needed for our models, in which $A=C(\Xi)$. Here $\Xi$ is the hull of a material, which can be any compact metrizable space with a $\Z$ action $\alpha$ which has a dense orbit. 
We choose a $\alpha$-invariant (Borel) 
probability measure $\mu$ on $\Xi$ to define a trace $\tau_A$ on $A$ as above through  $\tau_A(f)=\int_\Xi f(\xi)\mathrm{d}\mu(\xi)$.  
\begin{coro}\label{cor-chb=che1}
	Let $\Xi$ be a compact metrizable space with a topologically transitive action $\alpha$ by $\Z$ and $\mu$ an $\alpha$-invariant probability measure on $\Xi$. Then, for $x\in K_0(C(\Xi) \rtimes_\alpha \mathbb{Z})$ 
	$$\chb(x)\equiv \che(\beta(x)) \mod{1} .$$
\end{coro}

\begin{proof} Let $p$ be a projection in $M_N(C(\Xi))$ which satisfies $\alpha_*([p]_0)= [\alpha(p)]_0 = [p]_0$. Perhaps after adding $0$ projections this implies $\alpha(p)=upu^*$ for some unitary $u\in M_N(C(\Xi))$. Then $\mathrm{Tr}(p(\alpha(\xi))) = \mathrm{Tr} (up(\xi)u^*) = \mathrm{Tr} (p(\xi))$. Hence $\xi\mapsto \mathrm{Tr} (p(\xi)) $ is a continuous integer-valued $\alpha$-invariant function on $\Xi$. As $\Xi$ has a dense orbit this function must be constant. It follows that 
$\tau_A (p) = \int_\Xi \mathrm{Tr} (p(\xi)) d\mu(\xi)\in \Z$ and we conclude with Cor.~\ref{BBC_mapping_torus}.
\end{proof} 
The mapping torus of $C(\Xi)$ can be identified with $\tilde A = C(\tilde\Xi)$, where $\tilde \Xi=\Xi\times\R/\sim_\alpha$ is the suspension space of $\Xi$, namely $(\xi,t+1) \sim_\alpha (\alpha(\xi),t)$. The extension $\tau_{
\tilde A}$ of the trace $\tau_A$ to $\tilde A$ is given by integration w.r.t.\ the extended measure $\mu\times \lambda$, where $\lambda$ is the Lebesgue measure on $\R$. 
The derivation $\delta$ to define the $1$-cocycle $\che$ is derivation along the flow, that is, w.r.t.\ $t$. 

Let $h$ be a self-adjoint element in $C( \Xi)\rtimes_\alpha\Z$ which has a spectral gap $\Delta$, say around $0$. 
To evaluate $\che(\beta([P_F(h)]_0))$ we can procede as follows: We choose an augmentation $\tilde h$ of $h$ which also has a gap around $0$. This can be an element of $C(\tilde\Xi )\rtimes_\alpha \Z$ such that $q(\tilde h) = h$ but since we are calculating in $K$-theory it is sufficient to find 
$\tilde h$ in $M_N(C(\tilde\Xi )\rtimes_\alpha \Z)$ such that $q(P_F(\tilde h))$ defines the same $K_0$-class as $P_F(h)$. We could, for instance, take the augmentation which we used for Cor.~\ref{cor-chb=che} but for physical reasons this might not the best choice. 
Then, $\beta([P_F(h)]_0)$ is $\exp([P_F(\tilde h)]_0)$ modulo $\mathrm{Amb}_e$. To obtain $\exp([P_F(\tilde h)]_0)$ we choose a pre-image $\hat h\in \mathcal{T}(C(\tilde\Xi), \bar{\alpha})$ of $\tilde{h}$ under $\pi$ (perhaps matrix valued). 
In this way, $1-g(\hat h)$ ($g$ as in proposition \ref{useful_formula_for_trace}) is a lift of the projection $P_F(\tilde h)$.
Hence $U=\mathrm{e}^{2\pi\mathrm{i}g(\hat h)}$ is a unitary in $\mathcal{T}(C(\tilde\Xi), \bar{\alpha}) $ such that $U-1$ lies in the ideal 
$\tilde{\mathcal E}$ of the Teoplitz extension 
and we have 
$\exp ([1-g( \tilde h)]_0)=[\psi^{-1}(U)]_1$. If $\psi^{-1}(U)-1$ is trace class we get, with Prop.~\ref{useful_formula_for_trace}
\begin{equation}\label{eq-che-mapping}
\che(\beta([P_F(h)]_0)) 
\equiv - \tau_{\tilde{\mathcal  E}}\big( \varphi(\hat h)\delta \hat h\big) \mod{1}
\end{equation}
where $\tau_{\tilde{\mathcal  E}} = \tau_{\tilde A}\circ\psi^{-1}$.
What happens if our augmentation has no gap at $0$? Then $1-g(\tilde h)$ is no longer a projection so that we cannot argue as above with $\exp$. Nevertheless $U=\mathrm{e}^{2\pi\mathrm{i}g(\hat h)}$ is still a unitary in $\mathcal{T}(C(\tilde\Xi), \bar{\alpha}) $ but now $U-1$ only lies in the (larger) ideal $J$. We thus have  $\exp_J ([1-g( h)]_0)=[U]_1$. Moreover $\pi_*\circ \exp_J=0$ so that $[U]_1$ lies in the image of $\b_*$. The $K_1$-class $[\psi^{-1}(U)]_1$ has therefore a representative $U'$ to which we can apply $\ch^{(e)}$. However, it is not clear whether $\psi^{-1}(U)$ itself lies in the domain of $\ch^{(e)}$. For the following physical interpretation we therefore have to assume that $\psi^{-1}(U)$ lies in that domain. Having said that, our simulations below show that it does not seem necessary that the augmentation of $h$ has a gap around zero for the spectral flow to be visible.

\subsection{Interpretation of $\che$}\label{sec-interpretation-mapping}
We provide a interpretation of formula \eqref{eq-che-mapping} in a physical context, in which $\Xi$ is the hull of a one dimensional aperiodic material. For that we reformulate the chern cocycle on the level of the representations $\hat\rho_{\tilde\xi}$ of 
$\tilde {\mathcal E} \subset $
$\mathcal{T}(C(\tilde\Xi), \bar{\alpha})$. These representations are defined on $\ell^2(\N)$. Given
$b = f \hat u^n\hat p (\hat{u}^*)^m\in \tilde {\mathcal E}$ with $f\in C(\tilde\Xi)$, then
$$\hat\rho_{\tilde\xi}(b) \psi (k) = \delta_{0,k-n} f(\alpha^{-k}(\tilde\xi)) \psi(k-(n-m))$$
with $\psi(k)=0$ if $k<0$. The trace of such an element is given by $\mathrm{Tr}_{\ell^2(\N)}(\hat\rho_{\tilde\xi}(b))=\delta_{n,m} f(\alpha^{-n} (\tilde\xi))$.
From this we see, using $\alpha$-invariance of the measure, that
\begin{equation}\label{eq-tau-m}
\tau_{\tilde{\mathcal  E}}(b) = \delta_{n,m}  \int_{\tilde\Xi }f(\tilde{\xi})\mathrm{d}\tilde \mu(\tilde{\xi}) =  \int_{\tilde \Xi} \mathrm{Tr}_{\ell^2(\N)}(\hat\rho_{\tilde\xi}(b)) \mathrm{d}\tilde\mu(\tilde\xi).
\end{equation}
The derivation $\delta$ on $C(\tilde \Xi)\otimes \mathcal K$ is derivation along the flow (namely, $\partial_t$).
We thus obtain 
\begin{equation}\label{eq-che-m}
\tau_{\tilde{\mathcal  E}}(\varphi(\hat h)\delta \hat h) = \int_{\Xi}\int_0^1   \mathrm{Tr}_{\ell^2(\N)}(\varphi(\hat H_{\xi,t})\partial_t \hat H_{\xi,t} ) \mathrm{d} t \mathrm{d}\mu(\xi)
\end{equation}
where $\hat H_{\tilde\xi}=\hat\rho_{\tilde\xi}(\hat h)$. 
Assuming that $\mu$ is ergodic, the ergodicity theorem implies that for $\mu$-almost $\xi_0$ 
\begin{equation*}
	\int_{\Xi}\mathrm{Tr}_{\ell^2(\N)}(\varphi(\hat{H}_{\xi,t})\partial_t \hat{H}_{\xi,t} ) \mathrm{d}\mu(\xi)
	=  \lim_{L\to +\infty} \frac1{L} \sum_{l=0}^{L-1} \mathrm{Tr}_{\ell^2(\N)}(\varphi(\hat{H}_{\alpha^{-l}(\xi_0),t})\partial_t \hat{H}_{\alpha^{-l}(\xi_0),t} )
\end{equation*}
so that
$$\tau_{\tilde{\mathcal  E}}\big( \varphi(\hat h)\delta \hat h\big)= \lim_{L\to +\infty} \frac1{L} \int_0^{L}   \mathrm{Tr}_{\ell^2(\N)}(\varphi(\hat{H}_{\beta^{-t}(\xi_0,0)})\partial_t \hat{H}_{\beta^{-t}(\xi_0,0)} ) \mathrm{d}t .$$
We may choose $\hat h \in {\mathcal T}(C(\tilde\Xi),\alpha))$  
by replacing $u$ and $u^*$ with $\hat u$ and $\hat u^*$ in the expression for $\tilde h$. Then, perhaps up to topologically irrelevant terms,  
$\hat H_{\tilde\xi}$ is the restriction of $H_{\tilde\xi}$ to a half line, that is, to $\ell^2(\N)$. 
The derivative $\partial_t \hat{H}_{\beta^{-t}(\xi_0,0)}$ can be understood as a generalized force which arises when we move the system relative to the half line. This becomes clearer if we take as augmentation of $H_\xi$ the operator which we obtain if we interpolate the potential $V_\xi$, which is to begin with only defined at integer values, to real values. Then 
$\partial_t \hat{H}_{\alpha^{-n}(\xi_0),t}=\partial_t V_{\alpha^{-n}(\xi_0),t}$ is the force exhibited on the particles of the system when
shifting the potential against the boundary (edge) of the half line.
$\varphi(\hat{H}_{\tilde\xi}$) can be understood as a density matrix of states with energy in the gap. It has the physical dimension of inverse energy. It may be taken to approximate $\frac1{|\Delta|}$ times the characteristic function on the gap $\Delta$ so that $\varphi(\hat{H}_{\tilde\xi})$ becomes $\frac1{|\Delta|}$ times the spectral projection $P_{\Delta}(\hat{H}_{\tilde\xi})$ onto the states of $\hat{H}_{\tilde\xi}$ in $\Delta$. 
Consequently, 
$$ \mathrm{Tr}_{\ell^2(\N)}(\varphi(\hat{H}_{\beta^{-t}(\xi_0,0)})\partial_t \hat{H}_{\beta^{-t}(\xi_0,0)} )= \frac1{|\Delta|}\mathrm{Tr}_{\ell^2(\N)}(P_{\Delta}(\hat{H}_{\beta^{-t}(\xi_0,0)})\partial_t \hat{H}_{\beta^{-t}(\xi_0,0)} )$$ 
is the expectation value of this force w.r.t.\ to the edge states in the gap divided by the width of the gap. Finally the integration over the interval $[0,L]$ yields the work exhibited on the system if we move the system relative to the edge over a distance $L$. Thus 
$|\Delta| \tau_{\tilde{\mathcal  E}}(\varphi(\hat h)\delta \hat h) $ is the work 
per unit length exhibited on the system if we move the system relative to the edge. The fact that $\tau_{\tilde{\mathcal  E}}(\varphi(\hat h)\delta \hat h) $ does not depend on the choice of density function $\varphi$ means that we can let the width of the support of this function go to zero so that   
$\tau_{\tilde{\mathcal  E}}(\varphi(\hat h)\delta \hat h) $ is the work at energy $0$ (or any other energy in the gap) per unit length and energy which is exhibited on the augmented system, restricted to a half line, if we move the potential against the boundary of the half line. Our generalised bulk edge correspondence \eqref{eq-che-mapping} thus says that this quantity coincides modulo an integer to (minus) the number of states per unit length  below energy $0$ of the system without boundary.

Let us also give an interpretation of $\tau_{\tilde{\mathcal  E}}(\varphi(\hat h)\delta \hat h) $ through spectral flow. The spectrum of $\hat H_{\tilde\xi}$ in the gap of 
$H_{\tilde\xi}$ consists of eigenvalues $E_i(\xi,t)$, $i$ indexing these eigenvalues, whose associated eigenvectors are localized at the edge. If $h$ is an operator of finite interaction range (polynomial in $u$ and $u^*$) then there are only finitely many such eigenvalues and $(U^*-1)$ 
is trace class. As 
$\mathrm{Tr}_{\ell^2(\N)}(\varphi(\hat H_{\tilde\xi})\partial_t \hat H_{\tilde\xi} ) = \sum_i \varphi(E_i({\tilde\xi}))\partial_t E_i({\tilde\xi})$ we obtain
\begin{align*}\tau_{\tilde{\mathcal  E}}(\varphi(\hat h)\delta \hat h)  & = - \int_{\Xi}\sum_i\mathrm{SF}(E_i(\xi,t);t\in [0,1]) \mathrm{d}\mu(\xi)\\
& = -\lim_{L\to +\infty} \frac1{L} \sum_i\mathrm{SF}(E_i(\xi_0,t);t\in [0,L])\end{align*}
where
$$\mathrm{SF}(E_i(\xi,t);t\in [0,L]) = \int_0^L\varphi(E_i({\xi,t}))\partial_t E_i({\xi,t})) \mathrm{d} t $$
is the spectral flow of the eigenvalue $E_i(\xi,t)$ though the gap when $t$ varies between $0$ and $L$. With Cor.~\ref{cor-chb=che1}, \eqref{eq-che-mapping} and \eqref{eq-int-chb} we thus get
\begin{equation}\label{eq-che-mapping-int}
\mathrm{IDS}(E_F)  
\equiv \lim_{L\to +\infty} \frac1{L}\sum_i \mathrm{SF}(E_i(\xi,t);t\in [0,L])
\mod{1}
\end{equation}

\subsection{Choice of augmented Hamiltonian}
Given an abstract Hamiltonian $h\in C(\Xi)\rtimes_\alpha\Z$, the augmented Hamiltonian is an element $\tilde h$ of $C(\tilde\Xi)\rtimes_\alpha\Z$, or more generally $M_N(\tilde h\in C(\tilde\Xi)\rtimes_\alpha\Z)$, such that $q_*([P_F(\tilde h)]_0)=[P_F(h)]_0$. 
There are two natural choices which we label with A and B.

The more physically motivated augmentation is obtained by interpolation, for instance, by linear interpolation. 
Consider $h = \sum_{k} h_k u^k$ where $h_k\in C(\Xi)$. We define the augmentation $\tilde h^A$ by extending each $h_k$ to $\tilde h_k\in C(\tilde\Xi)$ as follows
\begin{equation}\label{eq-augA}
\tilde h^A: = \sum_{k} \tilde h_k u^k,\quad \tilde h_k(t) = (1-t) h_k + t \alpha(h_k) .
\end{equation}
In particular, if $h = u + u^{-1} + v$ with $v\in C(\Xi)$ is a Hamiltonian of Kohmoto type, then the covariant family of augmented Hamiltonians is $H_{\xi,t}^A=\tilde\rho_{\xi,t}(\tilde h^A)$, $(\xi,t)\in \tilde\Xi$ where
\begin{equation}\label{eq-augAb}
H_{\xi,t}^A\psi(n) = \psi(n+1)+\psi(n-1) + V_{\xi,t}(n)\psi(n), \quad V_{\xi,t}(n) = (1-t) V_\xi(n) + t V_\xi(n+1)
\end{equation}
where $ V_{\xi}(n) = v(\alpha^{-n}(\xi))$. 

The other augmentation, which was used in Cor.~\ref{cor-chb=che} is more mathematically motivated. It has the advantage that it has the same gap in its spectrum as $h$ and that the bulk edge correspondence holds for it without ambiguity. As addition in $K$-theory corresponds to the stacking of insulators it can be described as follows. 
We may assume that the gap of $h$ in question is around $0$. The identity operator $h_0=1$ can be understood as the trivial abstract insulator. It has a gap at $0$ energy and so its Fermi projection is zero. The direct sum $1\oplus h$ is therefore a stacked insulator which is in the same topological phase as $h$, indeed, $P_F(1\oplus h) = 0 \oplus P_F(h)$.
Now we couple the two layers in a peculiar way,
\begin{equation}\label{eq-augB}
\tilde h^B(t) := U_t \begin{pmatrix}
		1 & 0\\
		0 & h
	\end{pmatrix} U_t^*.
\end{equation}	
By construction, $\tilde h^B(t)$ is homotopic to $\tilde h^B(0) = 1\oplus h$ and unitary equivalent, so for given $t$, $	\tilde h^B(t)$ and $h$ are in the same topological phase when regarded as elements of the bulk algebra. However, this has to be distinguished from the topological phase of $\tilde h^B$ seen as an element of the augmented bulk algebra.  

The main interest in working with this augmentation is that it allows to give a simple proof of Cor.~\ref{cor-B=J}. Specified to the algebra $A=C(\Xi)$ for some hull $\Xi$ of a one-dimensional aperiodic Schr\"odinger operator, the corollary states that the $K$-theoretic gap labelling group of Bellissard coincides with the gap labelling group of Johnson and Moser \cite{johnson1982rotation,johnson1986exponential}. Indeed, the elements of $K_1(M_\alpha C(\Xi))$ are represented by homotopy classes of continuous functions $f:\tilde\Xi\to U_n(\C)$ and  
Johnson and Moser's gap labelling group is the subgroup of $\R$ obtained from the 
$\tilde\mu$-averages of $\frac1{2\pi i} \det(f)^{-1} \partial_t\det(f)$ where $\partial_t$ is derivative along the $\R$-flow.
Our way of proving this equality seems more explicit and it allows for an interpretation of the gap labels of Johnson and Moser in terms of a spectral flow, see below and Figure~\ref{spectrum_edge_MT_proofslift} where we visualise this with the example of a 1-cut Sturmian system.

\section{Augmentation with arcs}
We consider an augmentation of the hull $\Xi = \Xi_{\theta,\cut}$ of the  quasi-periodic Kohmoto models described in the introduction.
 It is obtained by joining the doubled points $(\phi,\epsilon)$, $\epsilon = \pm 1$ of $\Xi$ with an arc $\{(\phi,t):t\in [0,1]\}$. To get an intuition of what this means, the reader may think on the Cantor set which is obtained by taking out iteratively the middle third interval from $[0,1]$. What we do here is the inverse procedure, we add intervals to obtain a circle from a Cantor set. 
This is indicated in Figure~\ref{fig:1.2} where the inserted intervals appear as arcs.

\begin{figure*}
    \centering
    \begin{subfigure}[t]{0.45\textwidth}
        \centering
    \begin{tikzpicture}[scale=1.2, line cap=round,line join=round,>=triangle 45,x=2.000cm,y=2.000cm]
    	\clip(-1.500,-1.500) rectangle (1.500,1.500);
    	\draw [line width=1.5pt] (0.,0.) circle (2.cm);
    	\begin{scriptsize}
    		\SquareMarker{1}{0}
    		\draw(0.906,0.077) node {\textbf{0}};
    		\SquareMarker{0.540}{0.841}
    		\draw(0.437,0.781) node {\textbf{1}};
    		\SquareMarker{-0.416}{0.909}
    		\draw(-0.388,0.819) node {\textbf{2}};
    		\SquareMarker{-0.990}{0.141}
    		\draw (-0.866,0.099) node {\textbf{3}};
    		\SquareMarker{-0.654}{-0.757}
    		\draw (-0.600,-0.650) node {\textbf{4}};
    		\SquareMarker{0.284}{-0.959}
    		\draw (0.225,-0.794) node {\textbf{5}};
    		\SquareMarker{0.960}{-0.279}
    		\draw (0.831,-0.181) node {\textbf{6}};
    		\SquareMarker{0.754}{0.657}
    		\draw (0.657,0.592) node {\textbf{7}};
    		\SquareMarker{-0.146}{0.989}
    		\draw(-0.139,0.864) node {\textbf{8}};
    		\SquareMarker{0.622}{0.783}
    		\draw (0.558,0.690) node {\textit{0}};
    		\SquareMarker{-0.323}{0.946}
    		\draw (-0.290,0.834) node {\textit{1}};
    		\SquareMarker{-0.973}{0.240}
    		\draw(-0.858,0.251) node {\textit{2}};
    		\SquareMarker{-0.7237}{-0.684}
    		\draw(-0.676,-0.575) node {\textit{3}};
    		\SquareMarker{0.187}{-0.982}
    		\draw(0.119,-0.840) node {\textit{4}};
    		\SquareMarker{0.927}{-0.374}
    		\draw (0.800,-0.294) node {\textit{5}};
    		\SquareMarker{0.816}{0.578}
    		\draw (0.725,0.493) node {\textit{6}};
    		\SquareMarker{-0.046}{0.999}
    		\draw (-0.063,0.872) node {\textit{7}};
    	\end{scriptsize}
    \end{tikzpicture}
\end{subfigure}%
~ 
\begin{subfigure}[t]{0.45\textwidth}
\centering
\begin{tikzpicture}[scale=1.2, line cap=round,line join=round,>=triangle 45,x=2.000cm,y=2.000cm]
	\clip(-1.500,-1.500) rectangle (1.500,1.500);
	\draw [line width=1.5pt] (0.,0.) circle (2.cm);
	\draw [shift={(0.609,0.946)},line width=2pt, color=black]  plot[domain=-0.845:2.851,variable=\t]({1.*0.049*cos(\t r)+0.*0.049*sin(\t r)},{0.*0.049*cos(\t r)+1.*0.049*sin(\t r)});
	\draw [line width=2pt,color=black] (0.514,0.863)-- (0.563,0.960);
	\draw [line width=2pt,color=black] (0.568,0.83)-- (0.641,0.910);
	\draw [shift={(-0.467,1.024)},line width=2pt,color=black]  plot[domain=0.155:3.851,variable=\t]({1.*0.049*cos(\t r)+0.*0.049*sin(\t r)},{0.*0.049*cos(\t r)+1.*0.049*sin(\t r)});
	\draw [shift={(-1.114,0.160)},line width=2pt,color=black]  plot[domain=1.155:4.851,variable=\t]({1.*0.049*cos(\t r)+0.*0.049*sin(\t r)},{0.*0.049*cos(\t r)+1.*0.049*sin(\t r)});
	\draw [shift={(-0.737,-0.851)},line width=2pt,color=black]  plot[domain=2.155:5.851,variable=\t]({1.*0.049*cos(\t r)+0.*0.049*sin(\t r)},{0.*0.049*cos(\t r)+1.*0.049*sin(\t r)});
	\draw [shift={(0.318,-1.080)},line width=2pt,color=black]  plot[domain=-3.129:0.568,variable=\t]({1.*0.049*cos(\t r)+0.*0.049*sin(\t r)},{0.*0.049*cos(\t r)+1.*0.049*sin(\t r)});
	\draw [shift={(1.080,-0.316)},line width=2pt,color=black]  plot[domain=-2.129:1.568,variable=\t]({1.*0.049*cos(\t r)+0.*0.049*sin(\t r)},{0.*0.049*cos(\t r)+1.*0.049*sin(\t r)});
	\draw [shift={(0.849,0.738)},line width=2pt,color=black]  plot[domain=-1.129:2.568,variable=\t]({1.*0.049*cos(\t r)+0.*0.049*sin(\t r)},{0.*0.049*cos(\t r)+1.*0.049*sin(\t r)});
	\draw [shift={(-0.162,1.114)},line width=2pt,color=black]  plot[domain=-0.129:3.568,variable=\t]({1.*0.049*cos(\t r)+0.*0.049*sin(\t r)},{0.*0.049*cos(\t r)+1.*0.049*sin(\t r)});
	\draw [shift={(0.701,0.881)},line width=1pt,color=black]  plot[domain=-0.945:2.751,variable=\t]({1.*0.049*cos(\t r)+0.*0.049*sin(\t r)},{0.*0.049*cos(\t r)+1.*0.049*sin(\t r)});
	\draw [shift={(-0.363,1.066)},line width=1pt,color=black]  plot[domain=0.055:3.751,variable=\t]({1.*0.049*cos(\t r)+0.*0.049*sin(\t r)},{0.*0.049*cos(\t r)+1.*0.049*sin(\t r)});
	\draw [shift={(-1.093,0.271)},line width=1pt,color=black]  plot[domain=1.055:4.751,variable=\t]({1.*0.049*cos(\t r)+0.*0.049*sin(\t r)},{0.*0.049*cos(\t r)+1.*0.049*sin(\t r)});
	\draw [shift={(-0.818,-0.773)},line width=1pt,color=black]  plot[domain=2.055:5.751,variable=\t]({1.*0.049*cos(\t r)+0.*0.049*sin(\t r)},{0.*0.049*cos(\t r)+1.*0.049*sin(\t r)});
	\draw [shift={(0.209,-1.106)},line width=1pt,color=black]  plot[domain=-3.229:0.468,variable=\t]({1.*0.049*cos(\t r)+0.*0.049*sin(\t r)},{0.*0.049*cos(\t r)+1.*0.049*sin(\t r)});
	\draw [shift={(1.043,-0.422)},line width=1pt,color=black]  plot[domain=-2.229:1.468,variable=\t]({1.*0.049*cos(\t r)+0.*0.049*sin(\t r)},{0.*0.049*cos(\t r)+1.*0.049*sin(\t r)});
	\draw [shift={(0.919,0.650)},line width=1pt,color=black]  plot[domain=-1.229:2.468,variable=\t]({1.*0.049*cos(\t r)+0.*0.049*sin(\t r)},{0.*0.049*cos(\t r)+1.*0.049*sin(\t r)});
	\draw [shift={(-0.050,1.124)},line width=1pt,color=black]  plot[domain=-0.229:3.468,variable=\t]({1.*0.049*cos(\t r)+0.*0.049*sin(\t r)},{0.*0.049*cos(\t r)+1.*0.049*sin(\t r)});
	\draw [line width=1pt,color=black] (0.83,0.56)-- (0.936,0.605);
	\draw [line width=1pt,color=black] (0.8,0.59)-- (0.881,0.680);
	\draw [line width=2pt,color=black] (0.779,0.633)-- (0.870,0.695);
	\draw [line width=2pt,color=black] (0.74,0.685)-- (0.809,0.765);
	\draw [line width=1pt,color=black] (0.637,0.77)-- (0.729,0.841);
	\draw [line width=1pt,color=black] (0.612,0.795)-- (0.656,0.899);
	\draw [line width=1pt,color=black] (-0.025,0.999)-- (-0.003,1.113);
	\draw [line width=1pt,color=black] (-0.068,0.996)-- (-0.096,1.109);
	\draw [line width=2pt,color=black] (-0.115,0.995)-- (-0.114,1.108);
	\draw [line width=2pt,color=black] (-0.176,0.989)-- (-0.207,1.094);
	\draw [line width=1pt,color=black] (-0.314,1.068)-- (-0.303,0.95);
	\draw [line width=1pt,color=black] (-0.402,1.038)-- (-0.343,0.936);
	\draw [line width=2pt,color=black] (-0.39,0.92)-- (-0.419,1.031);
	\draw [line width=2pt,color=black] (-0.445,0.90)-- (-0.504,0.992);
	\draw [line width=1pt,color=black] (-1.069,0.313)-- (-0.965,0.260);
	\draw [line width=1pt,color=black] (-1.091,0.222)-- (-0.973,0.220);
	\draw [line width=2pt,color=black] (-1.094,0.205)-- (-0.990,0.17);
	\draw [line width=2pt,color=black] (-1.107,0.112)-- (-0.998,0.11);
	\draw [line width=1pt,color=black] (-0.841,-0.730)-- (-0.741,-0.668);
	\draw [line width=1pt,color=black] (-0.776,-0.798)-- (-0.711,-0.698);
	
	\draw [line width=2pt,color=black] (-0.763,-0.810)-- (-0.68,-0.735);
	\draw [line width=2pt,color=black] (-0.693,-0.871)-- (-0.637,-0.785);
	
	\draw [line width=1pt,color=black] (0.160,-1.102)-- (0.167,-0.985);
	\draw [line width=1pt,color=black] (0.252,-1.084)-- (0.207,-0.979);
	
	\draw [line width=2pt,color=black] (0.269,-1.080)-- (0.255,-0.968);
	\draw [line width=2pt,color=black] (0.359,-1.054)-- (0.315,-0.955);
	
	\draw [line width=1pt,color=black] (0.919,-0.389)-- (1.014,-0.460);
	\draw [line width=1pt,color=black] (0.930,-0.357)-- (1.048,-0.374);
	
	\draw [line width=2pt,color=black] (0.955,-0.31)-- (1.055,-0.357);
	\draw [line width=2pt,color=black] (0.97,-0.25)-- (1.080,-0.267);
	\draw [shift={(1.126,-0.001)},line width=2pt,color=black]  plot[domain=-1.845:1.851,variable=\t]({1.*0.049*cos(\t r)+0.*0.049*sin(\t r)},{0.*0.049*cos(\t r)+1.*0.049*sin(\t r)});
	\draw [line width=2pt,color=black] (1.112,0.045)-- (1.,0.03);
	\draw [line width=2pt,color=black] (1.112,-0.048)-- (1.,-0.03);
	\begin{scriptsize}
		\SquareMarker{1}{0}
		\draw[color=black] (0.906,0.077) node {\textbf{0}};
		\SquareMarker{0.540}{0.841}
		\draw[color=black] (0.437,0.781) node {\textbf{1}};
		\SquareMarker{-0.416}{0.909}
		\draw[color=black] (-0.388,0.819) node {\textbf{2}};
		\SquareMarker{-0.990}{0.141}
		\draw[color=black] (-0.866,0.099) node {\textbf{3}};
		\SquareMarker{-0.654}{-0.757}
		\draw[color=black] (-0.600,-0.650) node {\textbf{4}};
		\SquareMarker{0.284}{-0.959}
		\draw[color=black] (0.225,-0.794) node {\textbf{5}};
		\SquareMarker{0.754}{0.657}
		\draw[color=black] (0.657,0.592) node {\textbf{7}};
		\SquareMarker{0.622}{0.783}
		\draw[color=black] (0.558,0.690) node {\textit{0}};
		\SquareMarker{0.960}{-0.279}
		\draw[color=black] (0.831,-0.181) node {\textbf{6}};
		\SquareMarker{-0.323}{0.946}
		\draw[color=black] (-0.290,0.834) node {\textit{1}};
		\SquareMarker{-0.973}{0.240}
		\draw[color=black] (-0.858,0.251) node {\textit{2}};
		\SquareMarker{-0.7237}{-0.684}
		\draw[color=black] (-0.676,-0.575) node {\textit{3}};
		\SquareMarker{0.187}{-0.982}
		\draw[color=black] (0.119,-0.840) node {\textit{4}};
		\SquareMarker{0.927}{-0.374}
		\draw[color=black] (0.800,-0.294) node {\textit{5}};
		\SquareMarker{0.816}{0.578}
		\draw[color=black] (0.725,0.493) node {\textit{6}};
		\SquareMarker{-0.046}{0.999}
		\draw[color=black] (-0.063,0.872) node {\textit{7}};
		\SquareMarker{-0.146}{0.989}
		\draw[color=black] (-0.139,0.864) node {\textbf{8}};
	\end{scriptsize}
\end{tikzpicture}
    \end{subfigure}
    \caption{Graphical representation of $\Xi_{\theta,\cut}$ (left) and its augmented version $\tilde{\Xi}_{\theta,\cut}$ (right). The bold (resp. italic) numbers correspond to the points of the orbit of $\theta$ (resp. $\cut$). 
In the right figure the set $X_0$ is depicted with thick arcs and $X_\cut$ with thin arcs.     
    Both figures represent only an approximation to the topological spaces in that they show only finitely many cuts or added arcs, respectively. For irrational $\theta$,  $\Xi_{\theta,\cut}$ is actually a Cantor set.}
\end{figure*}
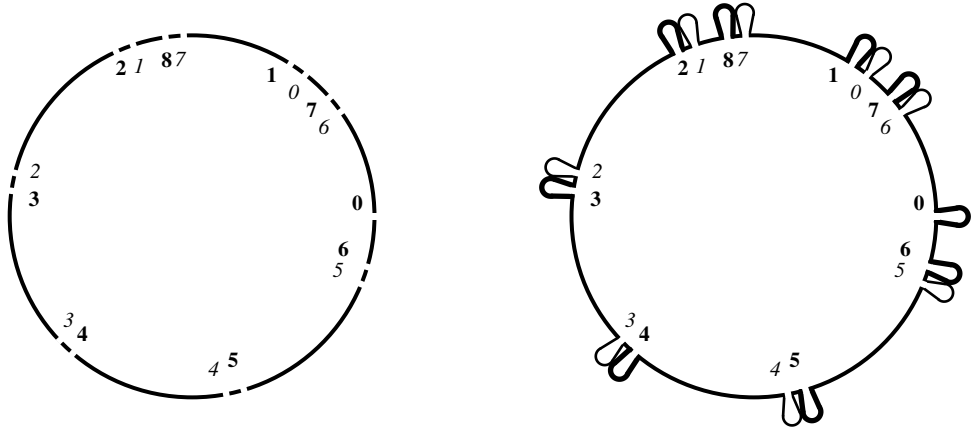\label{fig:1.2}
More rigorously, we define 
\begin{equation*}\label{cut_up_circle_augmented}
	\tilde{\R}_{\theta, \cut}\coloneqq \{(\phi,t) | \phi\in \R,t \in [0,1]\mbox{ if }\phi\in
	\{0,\cut\}+\Z+\theta\Z, 
	t = o \mbox{ otherwise}\} 
\end{equation*}
with the topology given by the total order 
$$  (\phi,t) < (\phi',t') \Leftrightarrow \phi < \phi' \mbox{ or } 
\phi = \phi', t <t' $$
and the augmented hull it its quotient by  
the action of $\Z$ on $\tilde{\R}_{\theta, \cut}$ by translation of the first coordinate $(\phi,t)\mapsto (\phi+1,t)$:
$$\tilde{\Xi}=\tilde{\Xi}_{\theta, \cut}\coloneqq \tilde{\R}_{\theta, \cut}/\Z.$$ 
By identifying $(\phi,-)$ with $(\phi,0)$ and $(\phi,+)$ with $(\phi,1)$ we see that $\tilde{\Xi}_{\theta, \cut}$ contains ${\Xi}_{\theta, \cut}$ as a closed subset. 
The action $\alpha$ of $\Z$ on $\tilde\Xi$ is given by $(\phi, t)\mapsto (\phi+\theta, t)$. 
We claim that $\tilde{\Xi}_{\theta, \cut}$ is homeomorphic to $S^1$. If $\theta$ is rational, this is clear, as only finitely many arcs are added to ${\Xi}_{\theta, \cut}$. If $\theta$ is irrational then $\tilde{\Xi}_{\theta, \cut}$ can be obtained through an inverse limit of circles in which we only add the arcs one by one. 
Figure~\ref{fig:1.2} shows only the addition of finitely many arcs.
We stress, however, that no such homeomorphism can conjugate the $\Z$-action with the rotation action on the circle, as the augmented dynamical system $(\Xi_{\theta,\cut},\alpha)$ is not minimal, while for irrational $\theta$ the rotation action on the circle is minimal.

With $\tilde{A}=C(\tilde{\Xi})$ the exact sequence (\ref{eq-SES1}) becomes
\begin{equation}\label{eq-SES2cut}
	0 \rightarrow C_0(\mathcal{S}) \xrightarrow{i}  C(\tilde{\Xi}) \xrightarrow{q} C({\Xi})
	\rightarrow 0.
\end{equation}
where $q$ is the restriction map and 
$\mathcal{S}\coloneqq\tilde{\Xi}\setminus\Xi$. We can write $\mathcal S=X_{0} \cup X_{\cut}$, where 
$$X_{\phi'} \coloneqq \{(\phi,t) | \phi\in \phi'+\theta\Z,t \in (0,1)
\}/\Z
$$ corresponding to the set of the arcs along the orbit of the cut point $\phi'$.
$X_0$ and $X_{\cut}$ coincide if $\cut\in \theta\Z+\Z$ but are otherwise disjoint.

Our goal is to analyze diagram (\ref{eqn:K-vertical}) which here becomes
\begin{equation}\label{eqn:K-vertical-cut}
\begin{matrix}
K_0(C_0(\Ss) \rtimes_\alpha \mathbb{Z}) & \stackrel{i_*}\to &
 K_0(C(\tilde \Xi) \rtimes_\alpha \mathbb{Z}) &  \stackrel{q_*}\to &
K_0(C(\Xi) \rtimes_\alpha \mathbb{Z}) &  \stackrel{\hat\exp}\to & 
K_1(C_0(\Ss) \rtimes_\alpha \mathbb{Z})\\
 & &\:\:\:\:\:\:\:\:\: \downarrow \che\!\circ\exp & &\:\:\:\:\:\: \downarrow \chb & &\:\:\:\:\:\:\:\:\: \downarrow \chab  \\
 & &\!\!\!\!\!\!\!\!\!\!\!\!\!\!\!\!\!\!\!\!\!\C  & & \!\!\!\!\!\!\!\!\!\!\!\!\!\C  & & \!\!\!\!\!\!\!\!\!\!\!\!\C
\end{matrix}
\end{equation}
The result will depend on whether when $\theta$ is irrational or not and on the number $\ell$ of orbits of cut points of $\Xi_{\theta,\cut}$. 
Here $\ell=1$ if $\cut$ is a multiple of $\theta$ and otherwise $\ell=2$. 
This could be easily generalized to models with more that $2$ cuts, but we leave that to the reader. 

Given two cut points $\phi_1$, $\phi_2$,
we write more briefly $\chi_{[\phi_1,\phi_2]}$ for 
the indicator function on the clopen subset $[(\phi_1,+),(\phi_2,-)]$. 
$\chi_{[\phi_1,\phi_2]}$ is a projection in $C(\Xi)$ and hence also in
$C(\Xi)\rtimes_\alpha\Z$.

To construct a pre-image in $C(\tilde\Xi)\rtimes_\alpha\Z$ for $\chi_{[0, \theta]}$ under $q$ we adapt the Rieffel projection (\cite{rieffel1981c}) to our context. Define the two real continuous function $f,g$ on $\tilde{\Xi}$
\begin{equation}\label{eq-Rieffel}
f(\phi,t) = \left\{\begin{array}{ll}
1 & \mbox{if } 0<\phi<\theta \\
0 & \mbox{if } \theta<\phi<1 \\
t & \mbox{if } \phi = 0 \\
1-t & \mbox{if } \phi = \theta 
\end{array}\right.
\qquad
g(\phi,t) = \left\{\begin{array}{ll}
\sqrt{t(1-t)} & \mbox{if } \phi =\theta \\
0 & \mbox{if } \phi\neq \theta 
\end{array}\right.
\end{equation}
and set 
		$$P_\theta = gu+f +u^* g.$$
		Clearly $P_\theta$ is self-adjoint and $P_\theta^2=P_\theta$ as
		$$P_\theta^2=(f^2+g^2+\alpha^{-1}(g^2))+(fg+g\alpha(f))u+u^*(\alpha(f)g+gf)=f+gu+u^*g$$
		and therefore $P_\theta$ is a projection in $C(\tilde{\Xi})\rtimes_\alpha\Z$. Moreover, it is clear that $q(g) = 0$ and $q(f ) = \chi_{[0 ,\theta]}$, so $P_\theta$ is a lift of $\chi_{[0 ,\theta]}$.

Let us compute $\exp([P_\theta]_0)$ where $\exp:K_0(C(\tilde\Xi)\rtimes_\alpha\Z)\to K_1(C(\tilde \Xi)$ is the exponential map of the Toeplitz extension exact sequence. For that we take as a lift of $P_\theta$ 
	$$\hat P_\theta = g\hat u+f +\hat u^* g.$$
It is easily seen that $\hat P_\theta^2 = \hat P_\theta - g^2 \hat p$.
We thus have
$$\hat\rho_{\phi,t}(\hat P_\theta^2) = \hat\rho_{\phi,t}(\hat P_\theta) - g^2(\phi,t) \hat\rho_{\phi,t}(\hat p).$$
Thus, if $\phi\neq \theta$ then $g^2(\phi,t)=0$ and $\hat\rho_{\phi,t}(\hat P_\theta)$ is a projection. On the other hand, we have $\hat\rho_{\theta,t}(\hat u^* g) = \alpha^{-1}(g)(\theta,t) \hat\rho_{\theta,t}(\hat u^*) = 0$ so that
$$\hat\rho_{\theta,t}(\hat P_\theta \hat p) = f(\theta,t)  \hat\rho_{\theta,t}( \hat p) = (1-t) \hat\rho_{\theta,t}( \hat p)$$
while $\hat\rho_{\theta,t}(\hat P_\theta^2 \hat p^\perp) = \hat\rho_{\theta,t}(\hat P_\theta \hat p^\perp)$. It follows that
$$e^{-2\pi \imath \hat P_\theta} - 1 = (e^{-2\pi \imath \hat P_\theta} - 1)\hat p$$
and, upon the identification of $\ker \pi_*$ with $C(\tilde\Xi)\otimes\mathcal K$, this becomes the function $U\in C(\tilde\Xi)\otimes\mathcal K^+$
\begin{equation}\label{eq-U}
(\phi,t)\mapsto  U(\phi,t)-1 = \delta_{\phi,\theta}(e^{2\pi \imath t}  - 1) \otimes e_{00}.
\end{equation}
The class of $U$ represents $\exp([P_\theta]_0)$.
\begin{lemma}\label{lem-K-aug}
	Let $\tilde\Xi$ be the augmentation of $\Xi=\Xi_{\theta,\cut}$ as described above. 
$K_1(C(\tilde\Xi))\cong \Z$ and $\exp([P_\theta]_0)$ is a generator of $K_1(C(\tilde\Xi))$.
Moreover, $[1]_0$ and $[P_\theta]_0$ are generators of $K_0(C(\tilde\Xi)\rtimes_\alpha\Z)\cong\Z^2$.
\end{lemma}
\begin{proof} 
We saw that $\tilde{\Xi}$ is homeomorphic to $S^1$. Hence $K_0(C(\tilde{\Xi}))\cong K_1(C(\tilde{\Xi}))\cong\Z$. Since $U(\phi,t)$ has winding number $1$ it represents a generator of $K_1(C(\tilde{\Xi}))$. Whatever the action $\alpha$, if it preserves orientation we have $\alpha_*=\mathrm{id}$. This implies with \eqref{eq-PV} that $K_0(C(\tilde{\Xi})\rtimes_\alpha\Z)\cong \Z^2$ with generators $[1]_0$ and a pre-image of the generator of  	$K_1(C(\tilde{\Xi}))$ under $\exp$. $[P_\theta]_0$ is such a pre-image.
\end{proof}

\subsection{Calculations of the $K$-groups}
\begin{lemma}\label{lem-K-irrational}
	Let $\Xi=\Xi_{\theta,\cut}$ with action $\alpha$ by rotation by $\theta$ as above. We denote by $\ell$ the number of orbits of cut points of $\Xi_{\theta,\cut}$.
	Let $\tilde\Xi$ be the augmentation of $\Xi$ as described above. 
\begin{enumerate}
\item If $\theta$ is irrational then
the exact sequence \eqref{eqn:K-vertical-cut} is given by 
\begin{equation}\label{eq-KES2cuti}\nonumber
0\xrightarrow{}  \Z^2 \xrightarrow{} \Z^{\ell+1}
	\xrightarrow{} \Z^{\ell}  .
\end{equation}
Moreover,  $[1]_0=q_*([1]_0)$, $[\chi_{[0,\theta]}]_0=q_*([P_\theta]_0)$, $[\chi_{[0,\cut]}]_0$ generate $K_0(C(\Xi)\rtimes_\alpha\Z)\cong\Z^{\ell+1}$. If $\ell=1$ then $[\chi_{[0,\cut]}]_0$ is an integer linear combination of $[1]_0$ and $[\chi_{[0,\theta]}]_0$. If $\ell=2$ then $\hat\exp([\chi_{[0,\cut]}]_0) = (-1,1)$.
\item If $\theta=\frac{\p}{\q}$ with $\p,\q$ coprime then the exact sequence \eqref{eqn:K-vertical-cut} is given by 
\begin{equation}\label{eq-KES2cutq}\nonumber
\Z^\ell \xrightarrow{}  \Z^2 \xrightarrow{1} \Z^{\ell}
	\xrightarrow{} \Z^{\ell}  
\end{equation}
 and $q_*$ has rank $1$.
Moreover, 
$[\chi_{[0,\frac1{\q}]}]_0$ and $[\chi_{[0,\cut]}]_0$ generate $K_0(C(\Xi)\rtimes_\alpha\Z)\cong\Z^{\ell}$. If $\ell=1$ then $[\chi_{[0,\cut]}]_0$ is a multiple of $[\chi_{[0,\frac1{\q}]}]_0$. If $\ell=2$ then $\hat\exp([\chi_{[0,\cut]}]_0) = (-1,1)$.
\end{enumerate}
\end{lemma}
\begin{proof} We consider first (1), that $\theta$ is irrational. 
We have $C_0(\Ss)=C_0((0,1)\times \Z)^\ell$ with action $\alpha(t,n) = (t,n-1)$. As $K_i(C_0((0,1)\times \Z)) \cong K_{i-1}(C_0(\Z))$
and  $K_{1}(C_0(\Z)) = 0$ and $K_{0}(C_0(\Z)) \cong C_0(\Z,\Z)$ with essentially the same formula for the action we deduce
$$K_0(C_0((0,1)\times \Z)\rtimes_\alpha\Z) \cong \mathrm{Inv}_\alpha C_0(\Z,\Z) = 0,\quad
K_1(C_0((0,1)\times \Z)\rtimes_\alpha\Z) \cong \mathrm{Coinv}_\alpha C_0(\Z,\Z) = \Z$$
which proves that the groups at the left and right end of \eqref{eqn:K-vertical-cut} are as stated.

As $\Xi$ is totally disconnected we can apply results of \cite{Bellissard1992} which show that $K_0(C(\Xi) \rtimes_\alpha \mathbb{Z})$ is isomorphic to $\mathrm{Coinv}_{\alpha}(C(\Xi,\Z))$ and under this isomorphism $[\chi_I]_0$ is mapped to the coinvariance class $[\chi_I]_\alpha$. All coinvariance classes are integer linear combinations of classes $[\chi_I]_\alpha$ where $I$ is a clopen subset of $\Xi$. The Boolean algebra of clopen subsets of $\Xi$ is generated by the intervals $[(\phi,+),(\phi',-)]$ where $\phi<\phi'$ are cut points.
It follows that the classes $[\chi_{[0,\phi]}]_\alpha$ 
		with $\phi\in \{0,\cut\}+\Z+\theta\Z\cap [0,1]$ generate $\mathrm{Coinv}_{\alpha}(C(\Xi,\Z))$. If $\phi = \{n\theta\}$ (rational part) then $[\chi_{[(0,+),(\phi,-)]}]_\alpha = n [\chi_{[(0,+),(\theta,-)]}]_\alpha + \lfloor n \theta \rfloor [1]_\alpha$. If $\phi = \{\cut+n\theta\}$ then $[\chi_{[0,\phi]}]_\alpha = [\chi_{[0,\cut]}]_\alpha + n [\chi_{[0,\theta]}]_\alpha + \lfloor \cut+ n \theta \rfloor [1]_\alpha$. This shows that $[1]_0$, $ [\chi_{[0,\cut]}]_0$, and $  [\chi_{[0,\theta]}]_0 $ generate $K_0(C(\Xi)\rtimes_\alpha\Z)$. To see that $[1]_0$, $ [\chi_{[0,\theta]}]_0$ are independent we apply $\chb$ to them obtaining $1$ and $\theta$ which are rationally independent. We have already seen above that $q(P_\theta) = \chi_{[0,\theta]}$ and clearly $q(1)=1$.

If $\ell = 1$ then $\cut\in \Z+\theta\Z$ which implies that   $[\chi_{[0,\cut]}]_\alpha$ is an integer linear combination of $[1]_\alpha$ and $  [\chi_{[0,\theta]}]_\alpha $. 

Let $\ell=2$. Define $\tilde{\chi}_{[0, \cut]}\in C(\tilde{\Xi}) \rtimes_\alpha \mathbb{Z}$ as the function which, in the arc parametrized by $0 \times (0, 1)$, increases linearly from $0$ to $1$, in the arc parametrized by $\cut \times (0, 1)$, decreases linearly from $1$ to $0$, and is otherwise constant. Then $q(\tilde\chi_{[0, \cut]})={\chi}_{[0, \cut]}$.
We compute
$$	
		\exp(-2\pi\text{i}\tilde{\chi}_{[0, \cut]})(\phi,t)= \left\{
		\begin{array}{lll}
			1 & \mbox{if } \phi \ne 0, \cut \\
			\text{e}^{-2\pi \text{i}t} &\mbox{if } \phi=0\\ \text{e}^{2\pi \text{i}t} &\mbox{if } \phi=\cut
		\end{array} \right. $$
so that $\hat\exp([{\chi}_{[0, \cut]}]_0)$ is  the class in $K_1(C_0(\Ss)\rtimes_\alpha\Z)$ represented by the function
\begin{equation}\label{gamma_g3}	
(t_1\otimes n_1,t_2\otimes n_2)\mapsto ((\text{e}^{-2\pi \text{i}t_1}-1)\otimes \delta_0(n_1)+1\otimes1, (\text{e}^{2\pi \text{i}t_2}-1)\otimes \delta_0(n_2)+1\otimes1),
\end{equation}
 $\delta_0\in C_0(\Z)$ taking value $1$ in $0$, and value $0$ elsewhere. Here we have identified 
$C_0(\mathcal{S}) $ with 
$(C_0(\left]0,1\right[)\otimes C_0( \Z))\oplus (C_0(\left]0,1\right[)\otimes C_0( \Z))$ and note that the representative of the class lies in the subalgebra $C_0(\mathcal{S})$ of $C_0(\Ss)\rtimes_\alpha\Z$.

We now consider the case (2), that $\theta = \frac{\p}{\q}$.
Now $C_0(\Ss)=C_0((0,1)\times \Z/\q\Z)^\ell$ with action $\alpha(t,n) = (t,n-1)$. As 
$$K_i(C_0((0,1)\times \Z/\q\Z)) \cong K_{i-1}(C(\Z/\q\Z))\cong \left\{\begin{array}{ll} 0 & \mbox{for } i = 0 \\
C(\Z/\q\Z,\Z) & \mbox{for } i = 1 \end{array}\right.$$
we get 
$$K_0(C_0((0,1)\times \Z/\q\Z)\rtimes_\alpha \Z) \cong \mathrm{Inv}_\alpha C(\Z/\q\Z,\Z) \cong \Z$$ 
$$K_1(C_0((0,1)\times \Z/\q\Z)\rtimes_\alpha \Z) \cong \mathrm{Coinv}_\alpha C(\Z/\q\Z,\Z) \cong \Z$$
which proves that the groups at the left and right end of \eqref{eqn:K-vertical-cut} are as stated.

While $\Xi$ is not totally disconnected, it is contractible to a finite set so that  $K_0(C(\Xi) \rtimes_\alpha \mathbb{Z})\cong \mathrm{Coinv}_{\alpha}(C(\Xi,\Z))$. We conclude as in the irrational case that the classes $[1]_0$, $ [\chi_{[0,\cut]}]_0$, and $  [\chi_{[0,\frac{\p}{\q}]}]_0 $ generate $K_0(C(\Xi)\rtimes_\alpha\Z)$. Let $n',N\in\Z$ be a solution to $n'\frac{\p}{\q} + N = \frac{1}{\q}$. Then  $[\chi_{[0,\frac{1}{\q}]}]_\alpha   = n'  [\chi_{[0,\frac{\p}{\q}]}]_\alpha + N [1]_\alpha$ showing that $[\chi_{[0,\frac{1}{\q}]}]_0$ and  $[\chi_{[0,\cut]}]_0$ generate  $\mathrm{Coinv}_{\alpha}(C(\Xi,\Z))$.

If $\ell = 1$ then $\cut\in \frac1{\q}\Z$ which implies that   $[\chi_{[0,\cut]}]_\alpha$ is a multiple of $  [\chi_{[0,\frac1{\q}]}]_\alpha $. 

If $\ell=2$ the same calculation as in the irrational case yields $\hat\exp([\chi_{[0,\cut]}]_0) = (-1,1)$.
\end{proof}
For the rest of this section we restrict our attention to the case that $\cut$ is not a multiple of $\theta$, hence $\ell=2$ and $\gamma$ has rank $1$. The case that $\cut=\theta$ has been analyzed for irrational $\theta$ in \cite{kellendonk_prodan} and the remaining cases are similar to that.
 
We make a choice of section $s$ for the remainder of this section, namely we take $s:\mathrm{im}\gamma \to K_0(C(\Xi)\rtimes_\alpha\Z)$ to be
\begin{equation}\label{eq-section}
s(\gamma([\chi_{[0, \cut]}]_0)) = [\chi_{[0, \cut]}]_0.
\end{equation}
\subsection{Chern cocycles}
To define the cocycles we consider $\alpha$-invariant measures on the spaces $\Xi$ and $\mathcal S$. On $\Xi$ we take the Borel probability measure 
$$\mu([(\phi,\epsilon),(\phi',\epsilon')])=\phi'-\phi,\quad\mbox{if } 0\leq\phi<\phi'\leq 1.$$
This defines a the trace $\tau_A$  on $C(\Xi)$ and $\hat \tau_A$ on the crossed product. 
Consequently the cocycle for the bulk algebra is 
$  \chb([p]_0)=\hat{\tau}_A(p) $ as in \eqref{chern_bulk}.

Note that $\mathcal S$ is a countable disjoint union of open intervals (the arcs). We consider on $\mathcal S$ the $\sigma$-finite Borel measure which restricts on each of these open intervals to the Lebesgue measure. This measure is $\alpha$-invariant, but, if $\theta$ is irrational, not finite. The corresponding densely defined trace on $I=C_0(\mathcal S)$ is 
\begin{equation}\label{eq-trI}
\tau_I(f) :=  \int_{\mathcal S}f(\tilde{\xi})\mathrm{d}\tilde{\mu}(\tilde{\xi}) = \sum_{\phi\in \{0,\cut\}+\theta\Z} \int_0^1 f(\phi,t)\mathrm{d}t,
\end{equation}
with domain $\mathcal{D}_{\tau_I}=L^1(\tilde\Xi,\tilde\mu)\cap C(\tilde\Xi)$, which is an ideal.
Any function  
of $C_0(\mathcal S)$ whose support is contained in only finitely many of the intervals is trace class. 

Define the derivation
$$\delta_1 f(\phi,t) = 
\left\{\begin{array}{ll} 
	0 & \mbox{if } t=o,0,1 \\
	\partial_t f(\phi,t) &
	\mbox{otherwise}
\end{array}\right.$$
with domain $\mathcal{D}_{\delta_1}=\{f\in C_0(\mathcal S)|\delta_1 f \in C_0(\mathcal S)\}$. We define the dense sub-algebra $\mathcal{A}=\{f\in \mathcal{D}_{\delta_1}|{\delta_1}f\in \mathcal{D}_{\tau_I}\}$ with norm 
$$\lVert f\rVert_\mathcal{A}= \lVert f\rVert_\infty +\lVert \delta_1f\rVert_\infty+\lVert \delta_1f\rVert_1.$$
We easily check that the hypothesis on $\mathcal{D}_{\delta_1}$, $\mathcal{D}_{\tau_I}$ and $\mathcal{A}$ (which is a Banach algebra) are verified. We obtain thus two densely defined $1$-cocycles on $\mathcal A\subset C_0(\mathcal S)$ 
through the formula, $\phi' \in \{0,\cut\}$,
$$ \mathrm{ch}^{(\phi')}(f,g) = \frac1{2\pi i} \int_{X_\varphi} f\delta_1 g \mathrm{d}\tilde\mu = \frac1{2\pi i} \sum_{\phi\in \varphi+\theta\Z} \int_{0}^1 f(\phi,t)
\partial_t g(\phi,t) d\mathrm{d}t.$$
We set
$$\mathrm{ch}^{(\pm)}=\mathrm{ch}^{(\cut)}\pm\mathrm{ch}^{(0)}.$$
While we have defined the $1$-cocycle $\mathrm{ch}^{(+)}$ on $C_0(X_0)\oplus C_0(X_{\cut})= C_0(\mathcal S)$ it extends to a $1$-cocycle on $C(\tilde\Xi)$. Indeed, as we have seen, $\tilde\Xi$ is homeomorphic to $S^1$ and $K_1(C(\tilde\Xi))$ is isomorphic to $\Z$ the  isomorphism assigning to the class of a unitary function $f\in M_n(C(\tilde\Xi))$ the winding number of its determinant. If the restriction of the derivative $\partial_t f$ to $\mathcal S$ is $\tilde\mu$-integrable then this winding number is exactly $\mathrm{ch}^{(+)}(f^{-1}, f)$ (as above, we absorb the matrix trace in the definition of $\mathrm{ch}$). In this way we obtain the morphism 
\begin{equation}
\che:=\theta \mathrm{ch}^{(+)}:K_1(C(\tilde\Xi))\to \C
\end{equation}

On the other hand, we can extend $\mathrm{ch}^{(-)}$ to the crossed product $C_0(\mathcal S)\rtimes_\alpha\Z$ by taking the trace $\hat \tau_I$ from \eqref{eq-trI} and extending the derivation $\delta_1$ trivially on the unitary $u$ which induces the action $\alpha$. Indeed, this is possible, as the derivation commutes with the action. We thus obtain the morphism 
\begin{equation}\label{eq-chab-cut}
\chab=\frac{\cut}{2}\mathrm{ch}^{(-)}:K_1(C_0(\mathcal S)\rtimes_\alpha\Z)\to\C
\end{equation}
\begin{prop}[BEC for extension with arcs, $\theta$ irrational] 
\label{propositionBBC_TwoCuts_irrational}
Let $\Xi=\Xi_{\theta,\cut}$ as above with irrational $\theta$. 
Given $x\in K_0(C(\Xi) \rtimes_\alpha \mathbb{Z})$ there are $N,n_1,n_2\in \Z$ such that 
	\begin{align*}
		\chb(x)&=N+n_1\theta+n_2\cut\\
		\che(\beta(x-s(\gamma(x))))&=-n_1\theta\\
		\chab(\gamma(x))&=-  n_2\cut.
	\end{align*}
	In particular, we have
	$$\chb(x)+\che(\beta(x-s(\gamma(x))))+ \chab(\gamma(x))\in \mathrm{Amb}_b=\Z$$
	while $\mathrm{Amb}_e = 0$.
\end{prop}

\begin{proof} We provide the proof for $\cut$ not being a multiple of $\theta$, the other case is similar and can also be found in \cite{kellendonk_prodan}. 
	To verify the formulas, we check their validity on the three generators of $K_0(C(\Xi) \rtimes_\alpha \mathbb{Z})$. 
	
	Applying $\chb$ to the generators, we get
	$$\chb([1]_0)=1,\quad \chb([\chi_{[0, \theta]}]_0)=\theta,\quad \chb([\chi_{[0, \cut]}]_0)=\cut$$
	and hence the first statement of the proposition. 
	
	Clearly, $\beta([1]_0)=\gamma([1]_0) = 0$ .
	
	We consider the generator $[\chi_{[0, \theta]}]_0$. Let $P_\theta$ be the Rieffel projection as defined above such that $q_*([P_\theta]_0)=[\chi_{[0, \theta]}]_0$. Hence $\gamma([\chi_{[0, \theta]}]_0)=0$ so that $\chab(\gamma([\chi_{[0, \theta]}]_0))=0$. Furthermore, $\beta([\chi_{[0, \theta]}]_0)=\exp ([P_\theta]_0)$ which is represented by the unitary $U$ from \eqref{eq-U}. Hence
$$\che(\beta([\chi_{[0, \theta]}]_0)) =\theta \ch^{+}(U^* , U) = -\theta $$	
as only the term $\phi = \theta$ contributes to the sum.	

	Finally, we consider  the generator $[\chi_{[0, \cut]}]_0$. Then $\che(\beta([\chi_{[0, \cut]}]_0-s(\gamma([\chi_{[0, \cut]}]_0))))=0$, because of our choice for $s$. In the previous proof we computed 
	$\gamma([\chi_{[0 ,\cut]}]_0)=\hat\exp([{\chi}_{[0, \cut]}]_0)$ (see \eqref{gamma_g3}).
	Hence 
	\begin{align*}
		\chab(\gamma([\chi_{[0 ,\cut]}]_0))& =\frac{\kappa}{\pi\mathrm{i}}
		\int_0^1 \exp(2\pi\mathrm{i}t))\partial_t\exp(-2\pi\mathrm{i}t)) =-\kappa.
	\end{align*}
	We remark that the group ambiguity $\mathrm{Amb}_e=\che(\mathrm{exp}\circ i_*(K_0(C_0(\mathcal{S}) \rtimes_\alpha \mathbb{Z})))$ is trivial (as $i_*=0$), 
	while $$\mathrm{Amb}_b=\chb(\tilde{\pi}_*(K_0(\mathcal{T}(C(\tilde{\Xi}), \alpha))))=\Z,$$ 
	as 
	$K_0(\mathcal{T}(C(\tilde{\Xi}), \alpha))\cong K_0(C(\tilde{\Xi}))\cong \Z$
	and the image of its generator under $\tilde\pi_*$ is $[1]_0$.
	
	Taking linear integer combinations of the three generators we get the statement of the proposition.
\end{proof}
\begin{prop}[BEC for extension with arcs, $\theta$ rational]
	\label{propositionBBC_TwoCuts_rational}
Let $\Xi=\Xi_{\theta,\cut}$ as above with $\theta=\frac{\p}{\q}$, $\p$ and $\q$ coprime. 
Given $x\in K_0(C(\Xi) \rtimes_\alpha \mathbb{Z})$ there are $N,n_1,n_2\in \Z$ such that
	\begin{align*}
		\chb(x)&=N+n_1\frac{\p}{\q}+n_2\cut\\
		\che(\beta(x-s(\gamma(x))))&\equiv -n_1\frac{\p}{\q} \mod{p}\\
		\chab(\gamma(x))&=-  n_2\cut.
	\end{align*}
	In particular, we have
	$$\chb(x)+\che(\beta(x-s(\gamma(x))))+ \chab(\gamma(x))  \equiv 0 \mod{1}.$$
\end{prop}
\begin{proof}
Everything can be proven exactly as in the irrational case, except that now $\mathrm{Amb}_e =\che(\mathrm{exp}\circ i_*(K_0(C_0(\mathcal{S}) \rtimes_\alpha \mathbb{Z})))$ is not trivial. To determine it we look again at the first two lines of diagram (\ref{eqn:focus_on_K-diagram}) completing the first line with the exponential map $\mathrm{exp}_I$:
	\begin{equation*}
		\begin{tikzcd}
			{K_0(C_0(\mathcal{S}) \rtimes_\alpha \mathbb{Z})} && K_1(C_0(\mathcal{S}))&& K_1(C_0(\mathcal{S})) \\
			\\
			{K_0(C(\tilde{\Xi})\rtimes_\alpha \mathbb{Z})} && {K_1(C(\tilde{\Xi}))} && \C  \\
			\arrow["{i_*}", from=1-1, to=3-1]
			\arrow["{i_*}", from=1-3, to=3-3]
			\arrow["{\exp}", from=3-1, to=3-3]
			\arrow["{\mathrm{exp}_I}", from=1-1, to=1-3]
			\arrow[dashed, "{\che}", from=3-3, to=3-5]
			\arrow["{1-\alpha_*}", from=1-3, to=1-5]
		\end{tikzcd}
	\end{equation*}	
By commutativity of the first square $\mathrm{Amb}_e = \che(i_* \mathrm{Inv}_{\alpha_*} K_1(C_0(\mathcal S)))$. We have seen that $K_1(C_0(\mathcal S))\cong \Z^q\oplus \Z^q$ and under this isomorphism $\ch^{(+)}\circ i_*$ maps $(n,m)$ to $n+m$. Hence $\ch^{(+)}(i_* \mathrm{Inv}_{\alpha_*} K_1(C_0(\mathcal S)))=q\Z$ which implies $\mathrm{Amb}_e = p\Z$. Therefore $\mathrm{Amb}_e \subset \mathrm{Amb}_b=\Z$.
\end{proof}
\subsection{Diophantine equation and $q$-ambiguitiy in the rational case}
If  $1$, $\theta$ and $\cut$ are rationally independent then the three numbers $N$, $n_1$, $n_2$ are uniquely determined by $\chb(x)$. 
For rational $\theta$ this is, of course, never possible. 

Let us consider $\theta=\frac{\p}{\q}$. 
We saw that then $\che(x)$ is only determined modulo $\p$, or equivalently, $n_1$ is determined through $\che(x)$ modulo $\q$.
We call this the $q$-\emph{ambiguity} of $n_1$. It comes about as there is no canonical lift of $x$ under $q_*$. In the physical context, $x$ is defined by the Fermi projection of a Hamiltonian $h$ and there might be other, physically motivated principles which make the choice of augmentation $\tilde h$ unique up to homotopy. For instance, it seems reasonable to require that $\tilde h$ is obtained from $h$ by interpolating the potential. Then $\che([P_F(h)]_0)$ is uniquely determined by $h$. 
But the equation of Prop.~\ref{propositionBBC_TwoCuts_rational} between the numerical invariants holds only up to an integer. The $q$-ambiguity can therefore also be understood as the non-uniqueness of the solution to this equation. Indeed, as $x-s(\gamma(x))\in \mathrm{im} q_*$ this element must be a multiple of $[\chi_{[0,\frac1{\q}]}]_0$ and hence 
$\chb(x) - \chab(x) = \chb(x-s(\gamma(x))) \in \frac{1}{\q}\Z$. We therefore get for $n_1$ the Diophantine equation (in which $n_1$ and $N$ are unkowns) 
\begin{equation}\label{eq-dioph}
N \q + \p n_1 = \q(\chb(x) - \chab(x))
\end{equation}
which has the solution $n_1 \equiv \p^{-1}\q(\chb(x) - \chab(x))\mod{\q}$ where $\p^{-1}$ is the inverse of $\p$ in $\Z/\q\Z$.
 
 \subsection{\texorpdfstring{Interpretation of the invariants}{}}
We interprete the numerical invariants in the physical context in which $x$ is defined by the Fermi projection of a Hamiltonian. For that we formulate these invariants in representations of the abstract $C^*$-algebras. 
The interpretation of  $\chb([P_F(h)]_0)$ as value of the IDS at the gap is always the same and has been given in Section~\ref{sec-IDS}.

\subsubsection{Interpretation of $\che$ for primary gaps}
For primary gaps we have $\gamma([P_F(h)]_0))=0$ and therefore the boundary invariant is given by $\che(\beta([P_F(h)]_0))$. We
rewrite this expression 
in a representation $\hat\rho_{\xi,t}$ of $\tilde{\mathcal E} \subset \mathcal{T}(C(\tilde\Xi),\alpha)$ on $\ell^2(\N)$. 
Here $\Xi$ is the 2 cut model based on a rotation by $\theta$, $\tilde\Xi$ its augmentation by arcs and the measure $\tilde\mu$ is supported on the added arcs $\mathcal S$ and on each arc equal to the Lebegue  measure. This measure is infinite if $\theta$ is irrational. It induces the trace $\tau_{I}$ which can be seen as a trace on $\tilde A=C(\tilde\Xi)$ whose domain is only dense in the ideal $I$. Defining $\tau_{\tilde{\mathcal E}}=\tau_I\circ \Psi^{-1}$ we find as in Section~\ref{sec-interpretation-mapping}
\begin{equation}\label{eq-trace-cute}
\tau_{\tilde{\mathcal E}}(b) = \int_{\mathcal S} \mathrm{Tr}_{\ell^2(\N)}(\hat\rho_{\tilde \xi}(b)) \mathrm{d}\tilde\mu(\tilde\phi)=\sum_{\phi\in{\mathcal C}}\int_0^1
 \mathrm{Tr}_{\ell^2(\N)}\big(\hat\rho_{\phi,t}(b) \big) \mathrm{d} t
 \end{equation}
for $b\in\tilde{\mathcal E}$ with $\Psi(b)=f\otimes E_{nm}$ with $\tilde\mu$-measurable $f$ ($\mathcal C$ is the set of cut points). 

$\delta$ is derivation along the arcs (and zero elsewhere). 

We first interpret $\che(\beta([P_F(h)]_0))$ in the case that $\gamma(h)=0$ (primary gap). 
This holds, for instance, if the augmentation $\tilde h$ of $h$ still has a gap around $0$. We then can define $U= e^{2\pi \imath g(\hat h)}$ as in Section~\ref{sec-interpretation-mapping} using the lift $\hat h$ of $\tilde h$ obtained by replacing $u$ by $\hat u$. We cannot expect that $U^*-1$ is trace class. However, we argue that $\delta U$ is trace class if we choose a good lift. For that we use the freedom to deform the bulk Hamiltonian $h$ in such a way that it is a Laurent polynomial in $u$, that is $h = \sum_{n\in F} h_n u^n$ for some finite subset $F\subset \Z$, whose coefficients  $h_n$ belong to $C(\Xi_N)$ where $\Xi_N$ is an approximation of $\Xi$ by a circle with $N$ cuts. Lift $h$ to $\tilde h=\sum_{n\in F} \tilde h_n u^n$ where $\tilde h_n$ belongs to  $C(\tilde \Xi_N)$ where $\tilde\Xi_N$ is the augmentation of $\Xi_N$ obtained by gluing $N$ arcs into the cuts. Compressing $\tilde h$ to the half-line to obtain $\hat h$ we see that $\partial_t \hat h$ has support only on $N$ arcs. But then also $\partial_t U$ has support only on finitely many arcs and hence is trace class. We thus can apply Prop.~\ref{useful_formula_for_trace} to obtain, with the same computations as Section~\ref{sec-interpretation-mapping}, 
$$\che(\beta([P_F(h)]_0)) \equiv - \theta \tau_{\tilde{\mathcal E}}(\varphi(\hat h) \delta \hat h) \mod{1} $$
where 
\begin{equation}\label{eq-che-cut}
\tau_{\tilde{\mathcal E}}(\varphi(\hat h) \delta \hat h)
=\sum_{\phi\in{\mathcal C}} \int_0^1   \mathrm{Tr}_{\ell^2(\N)}\big(\varphi(\hat H_{\phi,t})\partial_t \hat H_{\phi,t}\big) \mathrm{d}t. 
\end{equation}
Compared to \eqref{eq-che-m} we see that integration over $\Xi$ is replaced by summing over the cut points. This has an impact on the interpretation. While the derivative $\partial_t \hat{H}_{\xi,t}$ can still be understood as a generalized force which arises when we vary $t$, and 
$ \mathrm{Tr}_{\ell^2(\N)}\big(P_{\Delta}(\hat H_{\phi,t})\partial_t \hat H_{\phi,t}\big)$ is the expectation value of this force,
this variation does not come from a relative shift of the augmented potential against the boundary of the half-line, but from, what one calls in quasicrystal physics the phason motion. \eqref{eq-che-cut} says that 
$|\Delta|\tau_{\tilde{\mathcal E}}(P_\Delta(\hat h) \delta \hat h)$
is the work the system exhibits on the edge states of $\hat H_{\phi,t}$ with energy in the gap $\Delta$ which is induced by moving $t$ through all the arcs (even infinitely many if $\theta$ is irrational). This is the same as the work we obtain if the parameter $\tilde\xi$ varies around $\tilde\Xi$, which is a circle, since $\delta$ is supported on these arcs only and, as can be shown as in \cite{kellendonk_prodan}, that the Lebesgue measure of the set of eigenvalues of $H_{\phi,o}$ in the gap for $\phi\notin {\mathcal C}$ is zero. The work is different from the work in Section~\ref{sec-interpretation-mapping}, but our results show that they are related (modulo $1$) by a factor of $\theta$.

As in Section~\ref{sec-interpretation-mapping} we also have a spectral flow intepretation. Also here the spectrum of $\hat H_{\tilde\xi}$ in the gap of 
$H_{\tilde\xi}$ consists of eigenvalues $E_i(\xi,t)$ whose associated eigenvectors are localized at the edge. Thus with 
$\mathrm{Tr}_{\ell^2(\N)}(\varphi(\hat H_{\tilde\xi})\partial_t \hat H_{\tilde\xi} ) = \sum_i \varphi(E_i({\tilde\xi}))\partial_t E_i({\tilde\xi})$ 
we obtain
\begin{equation}\label{eq-SF-edge}
\tau_{\tilde{\mathcal  E}}(\varphi(\hat h)\delta \hat h)   = \sum_{\phi\in{\mathcal C}} \sum_i\mathrm{SF}(E_i(\phi,t);t\in [0,1])
\end{equation}
which is the spectral flow of the eigenvalues $E_i(\phi,t)$ though the gap when $t$ varies through all the arcs. 
We visualize this in Figure~\ref{primary_gaps_WN}. This is simple to imagine if $\theta$ is rational, because then we have only finitely many cut points. That the result is finite also in the irrational case comes about as the edge states do not feel potential changes which are far out in the chain. 

\subsubsection{Interpretation of $\che$ for secondary gaps}\label{sec-secondary-gaps}
We now look at the case in which $\gamma(h)\neq 0$, which means that the $n_2$ from above is not $0$ (secondary gap). In that case, the quantity to compute is $\che(\beta(x))$ where $x=[P_F(h)]_0-s\circ\gamma([P_F(h)]_0)$.   
With our choice of section \eqref{eq-section} we have $s\circ\gamma([P_F(h)]_0 = n_2[\chi_{[0,\cut]}]_0$ so that $x=[P_F(h)]_0 - n_2[\chi_{[0,\cut]}]_0$. To better understand the difference class $x$ we express it in the van Daele picture of $K_0$ by means of invertible self-adjoint elements, which we described in Section~\ref{sec-K-theory}. Then $x$ has the form $[[h-E_F],[h']]$ where $h'$ is a self-adjoint invertible element whose spectral projection onto negative energy states is homotopic to the direct sum of $n_2$ copies of $\chi_{[0,\cut]}$. Let us first consider the situation that $n_2=1$. Then we may take $h' = \lambda'(1-2\chi_{[0,\cut]})$ with $\lambda'>0$. The class $[[h-E_F],[h']]$ is the same as $[[h-E_F\oplus -h'],[1\oplus -1]]$ and, if $\lambda'>|E_F|$ then $h-E_F\oplus -h'$ is homotopic to  $h^s-E_F\in GL^{s.a.}_2(C(\Xi)\rtimes_\alpha\Z)$ where
$$h^s:=\begin{pmatrix}
	h & 0 \\ 0 & -h' 
\end{pmatrix}.$$
The Hamiltonian $h^s$ describes the stacking of the original system with an insulator which is an atomic limit, as its Hamiltonian has no kinetic term.
Going back to the description of $K_0$-classes by means of projections, we have
$$\che(\beta(x)) = \che\big(\beta([P_F(h^s)]_0- [P_F(1\oplus -1)]_0)\big)= \che\big(\beta([P_F(h^s)]_0\big)$$
as the projection $P_F(1\oplus -1)$ lies in the image of $\tilde\pi$. We thus can compute  
$\che(\beta(x))$ as for primary gaps but with $h$ replaced by $h^s$. 
For this computation needs an augmentation $\tilde h^s$ of $h^s$, that is, an element 
$$\tilde h^s:=\begin{pmatrix}
	\tilde h & \tilde c \\ \tilde c^* & -\tilde h' 
\end{pmatrix}$$
such that $q(\tilde h^s) = h^s$ and $\tilde h^s$ has a gap inside the gap $\Delta$ of $h^s$. For $h$ and $h'$ we follow the same augmentation procedure by interpolation of the potential as above. Furthermore, 
$\tilde c\in C(\tilde \Xi)\rtimes_\alpha\Z$ is an interaction term coupling the two layers. We cannot take it to be zero, as the secondary gaps of $h$ are filled in the spectrum of $\tilde h$. But since the potentials of $\tilde h$ and $-\tilde h'$ are multiples of each other with opposite sign, the spectrum of their direct sum is degenerated at one point in the gap $\Delta$ so that perturbation theory predicts the opening of a gap in $\tilde h^s$ provided the coupling $\tilde c$ is chosen appropriately. In that case we obtain the same formulas for 
$\che(\beta(x))$ as for primary gaps but with $\tilde h$ replaced by $\tilde h^s$. 
Different positive values for $n_2$ can be handled by going over to matrix valued operators, viewing $h$ and an element of $M_{n_2}(C(\Xi)\rtimes_\alpha\Z)$ as described in Section~\ref{sec-K-theory}. If $n_2$ is negative we replace $h'$ by $-h'$. 

\subsubsection{Interpretation of $\chab$}
Since $\gamma=0$ if $\cut$ is a multiple of $\theta$ we assume in this section that this is not the case. We express $\chab([P_F(h)]_0)$ in the representation $\rho_{\tilde \xi}$ of $C(\tilde\Xi)\rtimes_\alpha\Z$ on $\ell^2(\Z)$ restricted to its ideal $C_0(\mathcal S)\rtimes_\alpha\Z$. The relevant trace on the latter algebra is $\hat\tau_I$. Given $b = f u^k\in C_0(\mathcal S)\rtimes_\alpha\Z$ with $f$ $\tau_I$-trace-class, we have
$$\hat\tau_I(b) = \delta_{0,k} \int_{\mathcal S} f (\tilde \xi) \mathrm{d}\tilde\mu(\tilde\xi) = 
\delta_{0,k}\sum_{\phi\in \mathcal C} \int_0^1 f(\phi,t) \mathrm{d}t
$$
As $\rho_{\phi,t}(b)\psi(n) = f(\phi+n\theta,t)\psi(n-k)$ we find 
$$\mathrm{Tr}(\rho_{\phi,t}(b)) = \delta_{k,0} \sum_{n\in\Z} f(\phi+n\theta,t) .$$ 
Therefore, if $\theta$ is irrational we obtain
\begin{equation}\label{eq-trace-cutab}
\hat\tau_I(b) =  \sum_{\phi=0,\cut} \int_0^1 \mathrm{Tr}( \rho_{\phi,t}(b)) .\mathrm{d}t .
\end{equation}
Note the difference of this equation with the corresponding equation \eqref{eq-trace-cute} for the edge invariant: the sum is only over two cut points. For rational $\theta = \frac{\p}{\q}$ we can rewrite $\hat\tau_I(b)$ with the help of the trace per unit length $\mathrm{Tr}_{\mathrm{vol}}(\rho_{\phi,t}(b)) = \delta_{0,k} \frac1{\q} \sum_{n=0}^{\q-1} f(\phi+n\theta,t)$. This yields
\begin{equation}\label{eq-trace-cutab-rat}
\hat\tau_I(b) =  \sum_{\phi=0,\cut} \q \int_0^1 \mathrm{Tr}_{\mathrm{vol}}( \rho_{\phi,t}(b)) \mathrm{d}t = \sum_{\phi\in\mathcal C} \int_0^1 \mathrm{Tr}_{\mathrm{vol}}( \rho_{\phi,t}(b)) \mathrm{d}t .
\end{equation}
The derivation
$\delta$ is derivation along the arcs, but with negative orientation for the arcs in the orbit of  $0$, and extended trivially to the crossed product. 
More precisely, $\delta f(\phi+n\theta,t) = \sigma(\phi) \partial_t f(\phi+n\theta,t)$ where $\sigma(\phi)=+1$ if $\phi$ belongs to the orbit of  $\cut$ whereas $\sigma(\phi)=+1$ if $\phi$ belongs to the orbit of $0$. As above, we can argue that the lift $\tilde h$ of the bulk Hamiltonian $h$ into the augmented bulk algebra has the property 
that $\partial_t \tilde h$ is supported only on finitely many arcs and therefore $\delta \tilde h$ is $\tau$-trace class. Then, applying proposition \ref{useful_formula_for_trace}, we have
$$\chab(\hat\exp([P_F(h)]_0))=-\frac{\cut}2 \hat\tau_I(\varphi(\tilde h)\delta \tilde h).$$
For irrational $\theta$ we have
\begin{eqnarray} \label{physical_interpr_n2_a}
	\hat\tau_I(\varphi(\tilde h)\delta \tilde h) &=&  \int_0^1 
	\mathrm{Tr}
	\big( 
	\varphi(H_{\cut,t})\partial_tH_{\cut,t}
	-
	\varphi(H_{0,t})\partial_t H_{0,t}\big) \mathrm{d}t\\
 \label{physical_interpr_n2_b}	&=& \sum_i
\big(	\mathrm{SF}(E_i(\cut,t);t\in [0,1])-\mathrm{SF}(E_i(0,t);t\in [0,1])\big)
\end{eqnarray}
where $E_i(\phi,t)$ are the eigenvalues of $H_{\phi,t}$ in the gap $\Delta$. This is the difference of spectral flows when $t$ varies from $0$ to $1$ at the two cut points $\cut$ and $0$. 
For rational $\theta$ we have
\begin{eqnarray}\nonumber
	\hat\tau_I(\varphi(\tilde h)\delta \tilde h) & = &
	q \int_0^1 
	{\mathrm{Tr}_\mathrm{vol}}
	\big( 
	\varphi(H_{\cut,t})\partial_tH_{\cut,t}
	-
	\varphi(H_{0,t})\partial_t H_{0,t}\big) 
	\mathrm{d}t \\
 \label{physical_interpr_n2_c}	&=& \sum_i
\big(	\mathrm{SF}(E_i(\cut,t);t\in [0,1])-\mathrm{SF}(E_i(0,t);t\in [0,1])\big)
	\end{eqnarray}
Let us note that in the rational case we also can write 
$$\hat\tau_I(\varphi(\tilde h)\delta \tilde h) =	\sum_{\phi\in\mathcal C} \frac{\sigma(\phi)}{\q}  \sum_i
	\mathrm{SF}(E_i(\phi,t);t\in [0,1]).$$ 
	This is the signed spectral flow per unit length of the eigenvalues of $H_{\phi,t}$ in the gap as $\tilde\phi$ varies around $\tilde\Xi$, which is a circle.
We now show that, for our specific model \eqref{our_model}, the above quantity can take only the values $0$ or $\pm 2$, and the first is the case precisely if the augmentation does not fill the gap.
\begin{prop} \label{prop-n2}
Let $\Xi = \Xi_{\theta,\cut}$ be the hull as defined in \eqref{eq-hull}. Let $h(\lambda)\in C(\Xi)\rtimes_\alpha\Z$ be given by 
\begin{equation}
	\nonumber
	h(\lambda) := u+u^*+\lambda(1-2\chi_{[0, \cut]})
\end{equation}
where  
$\chi_{[0, \cut]}$ denotes the characteristic function on the closed interval
$[(0,+),(\cut,-)]$ and $\lambda\neq 0$. We assume that $\cut$ is not a multiple of $\theta$ so that $0$ and $\cut$ lie in different orbits under the rotation action. 
Let $\tilde h$ be the augmentation of $h$ given by linear interpolation which we discussed above. Either $\Delta$ is also a gap in the spectrum of $\tilde h$ and then $\hat\tau_I(\varphi(\tilde h)\delta \tilde h)=0$
or the gap is filled and 
$\hat\tau_I(\varphi(\tilde h)\delta \tilde h)=-2\mathrm{sgn}(\lambda)$.
\end{prop}
\begin{proof}
Recall that 
$\rho_{\phi,t}(\tilde h)= H_{\phi,t}$, 
$H_{\phi,t}\psi(n) = \psi(n+1)+\psi(n-1) + V_{\phi,t}(n)\Psi(n)$ with 
$$V_{\phi, t}(n) = \left\{\begin{array}{ll}
\lambda(1-2\chi_{[0, \cut]}(n\theta+\phi)) & \mbox{if } n\theta+\phi \neq 0,\cut \\
\lambda(1-2t) & \mbox{if } n\theta+\phi = 0 \\
\lambda(2t-1) & \mbox{if } n\theta+\phi =\cut
\end{array}
\right.$$
We consider the irrational case first. Then, if
$H_{\phi,t} - H_{\phi,0}\neq 0$, it is a rank one perturbation. It follows from perturbation theory that there is at most one eigenvalue of $H_{\phi,t}$ in the gap. Suppose there is an eigenvalue $E(\phi,t)$ in the gap for some $t$. Let $\psi(E,t)$ be its normalised eigenvector. Then 
$${\mathrm{Tr}}
	\varphi(H_{\phi,t})\partial_tH_{\phi,t} = \langle \psi(\phi,t)|\varphi(H_{\phi,t})\partial_t H_{\phi,t}|\psi(\phi,t) \rangle = \varphi(E(\phi,t))\partial_t E(\phi,t)$$
Moreover, $\partial_t E(\phi,t)$ must have the same sign as $\langle \psi(\phi,t)|\partial_t V_{\phi,t}|\psi(\phi,t) \rangle$ which is $\sigma(\phi)\mathrm{sgn}(\lambda)$. Therefore $E(\phi,t)$ must cross the gap exactly once when we vary $t$ from $0$ to $1$ so that the spectral flow is
$$\mathrm{SF}(E(\phi,t);t\in [0,1]) = \sigma(\phi)\mathrm{sgn}(\lambda).$$ 
As $V_{\cut,t}(n) = V_{0,1-t}(-n)$ we see that, for $0<t<1$, 
either both, $H_{\cut,t}$ and $H_{0,t}$, or none has an eigenvalue in the gap. Thus the r.h.s.\ of \eqref{physical_interpr_n2_b} is $0$ or $-2\mathrm{sgn}(\lambda)$.
 
The argument for the rational case, $\theta=\frac\p\q$, is essentially the same, except that,
 if
$H_{\phi,t} - H_{\phi,0}\neq 0$, it is the $\q$-periodic repetition of a rank one perturbation. The analogous arguments show that either all of 
$H_{\phi,t}$, $\phi\in \mathcal C$ or none have an eigenvalue in the gap, although this eigenvalue is then infinitely degenerated. Nevertheless, the spectral flow of this eigenvalue is then also $\mathrm{SF}(E(\phi,t);t\in [0,1]) = \sigma(\phi)\mathrm{sgn}(\lambda)$ and we conclude with \eqref{physical_interpr_n2_c}.
\end{proof}


\section{Numerical simulations of topological effects in the Kohmoto model}
\label{numerical_simulations}
The goal of this section is to visualize the spectral flow of the edge states and, if present, of the augmented bulk states of the Kohmoto model. We do this for both types of augmentations considered above. While we are primarily interested in aperiodic systems, simulation on the computer is only possible for rational values of the rotation angle $\theta$ 
and so the idea is to approximate an irrational angle by rational ones. However, it is not obvious what this means. Indeed, in our approach this means that we approximate the dynamical system given of an irrational rotation by one with a rational rotation. As we will see, for the 2-cut model, this is not the same as approximating an aperiodic operator by a periodic operator. Instead, it will be approximated by two periodic operators.
 
Let $\Xi = \Xi_{\theta,\cut}$ be the hull as defined in \eqref{eq-hull}. Our abstract Hamiltonian $h\in C(\Xi)\rtimes_\alpha\Z$ from \eqref{our_model} 
defines the covariant family of operators
$H_\Xi := \{H_\xi\}_{\xi\in\Xi}$ where 
$H_{\phi,\epsilon} =\rho_{\phi,\epsilon}(h)$ acts on $\ell^2(\Z)$ as 
$$H_{\phi,\epsilon} \Psi(n)=\Psi(n+1)+\Psi(n-1)+V_{\phi,\epsilon}(n) \Psi(n)$$
with 
\begin{equation}
	V_{\phi,\epsilon}(n) = \left\{\begin{array}{ll}
\lambda & \mbox{if } (\{n\theta+\phi\},\epsilon) \geq (\cut,+) \\
-\lambda & \mbox{if } (\{n\theta+\phi\},\epsilon) \leq (\cut,-)
\end{array}\right.
\end{equation}
We present below three different augmentations $\tilde h$ of $h$. They lead as well to covariant families 
$H_{\tilde \Xi}:= \{H_{\tilde\xi}\}_{\tilde\xi\in\tilde\Xi}$ acting on  $\ell^2(\Z)$ as 
$$H_{\tilde\xi} \Psi(n)=\Psi(n+1)+\Psi(n-1)+V_{\tilde \xi}(n) \Psi(n)$$
where the potential $V_{\tilde \xi}$ depends on the choice of augmentation.
We are interested in the spectral flow of the edge states of the family of operators $\hat H_{\tilde \Xi} := \{\hat H_{\tilde\xi}\}_{\tilde\xi\in\tilde\Xi}$ which we obtain if we restrict $\hat H_{\tilde\xi}$ to $\ell^2(\N)$. 

In our numerical simulations, we consider $\theta=\frac{F_{n-2}}{F_n}$, where $F_n$ is the $n^{th}$ Fibonacci number; hence $\theta$ is a rational approximation of the square of the golden ratio. The operators $H_{\xi}$ and $H_{\tilde \xi}$ are thus periodic with period $\q=F_n$.
Using Bloch theory each Hamiltonian could be represented by $k$-dependent matrices of size $\q$. However, it will be useful in cases to represent it by matrices of size $D$. If $H_\zeta$ is periodic then we take $D$ to be a multiple of $\q$. Its matrix representation has the form  
\[
H_\zeta(k) = \begin{pmatrix}
	V_\zeta(0) & 1 & 0 & \cdots & 0 & \mathrm{e}^{-2\pi k\mathrm{i}} \\
	1 & V_\zeta(1) & 1 & \cdots & 0 & 0 \\
	0 & 1 & V_\zeta(2) & \cdots & 0 & 0 \\
	\vdots & \vdots & \vdots & \ddots & \vdots & \vdots \\
	0 & 0 & 0 & \cdots & V_\zeta(D-2) & 1 \\
	\mathrm{e}^{2\pi k\mathrm{i}} & 0 & 0 & \cdots & 1 & V_\zeta(D-1)
\end{pmatrix}
\]
where $\zeta = \xi$ or $\zeta=\tilde\xi$. The spectrum of $H_\zeta$ is the union of the spectra of $H_\zeta(k)$ over all $k\in
[0, 2\pi)$.
The operator $\hat H_{\tilde\xi}$ is represented by the above matrix for $H_{\tilde\xi}$ in which we delete the entries at the upper right and lower left corner so that it becomes tri-diagonal. In this way we simulate Dirichlet boundary conditions. We must be aware that we are introducing left and right boundaries at the same time. To distinguish left and right edge states we test whether the corresponding wave function is sufficiently localised at one of the boundaries. To represent the half space operator $\hat H_{\tilde\xi}$ we do not need that $D$ is a multiple of $\q$.

\subsection{Augmentation of the 1-cut model with the mapping torus}
We consider here the two augmentations defined in \eqref{eq-augA} and \eqref{eq-augB} for the mapping torus construction, restricting our attention to the case that $\cut$ is a multiple of $\theta$ so that we have only one orbit of cut points.
For this model, the spectrum of $H_{\phi,\epsilon}$ does not depend on 
${\phi,\epsilon}$, and so the spectrum of $h$ coincides with that of $H_{0,-}$. 

For the simulations, we choose the parameters $\theta = \frac{55}{144}$, $\lambda=2$. In 
Figure~\ref{spectrum_Xitilde_MT} we see the spectrum of the augmented Hamiltonians $H^A_{\xi,t}$ for $t\in [0,1]$ (with $D=\q$).

\begin{figure}
	\centering
	\begin{tikzpicture}
		\node[inner sep=0pt] (img) at (0,0)
		{\includegraphics[
			height=7.5cm,
			trim=7.7cm 7.85cm 6.17cm 0cm,   
			clip
			]{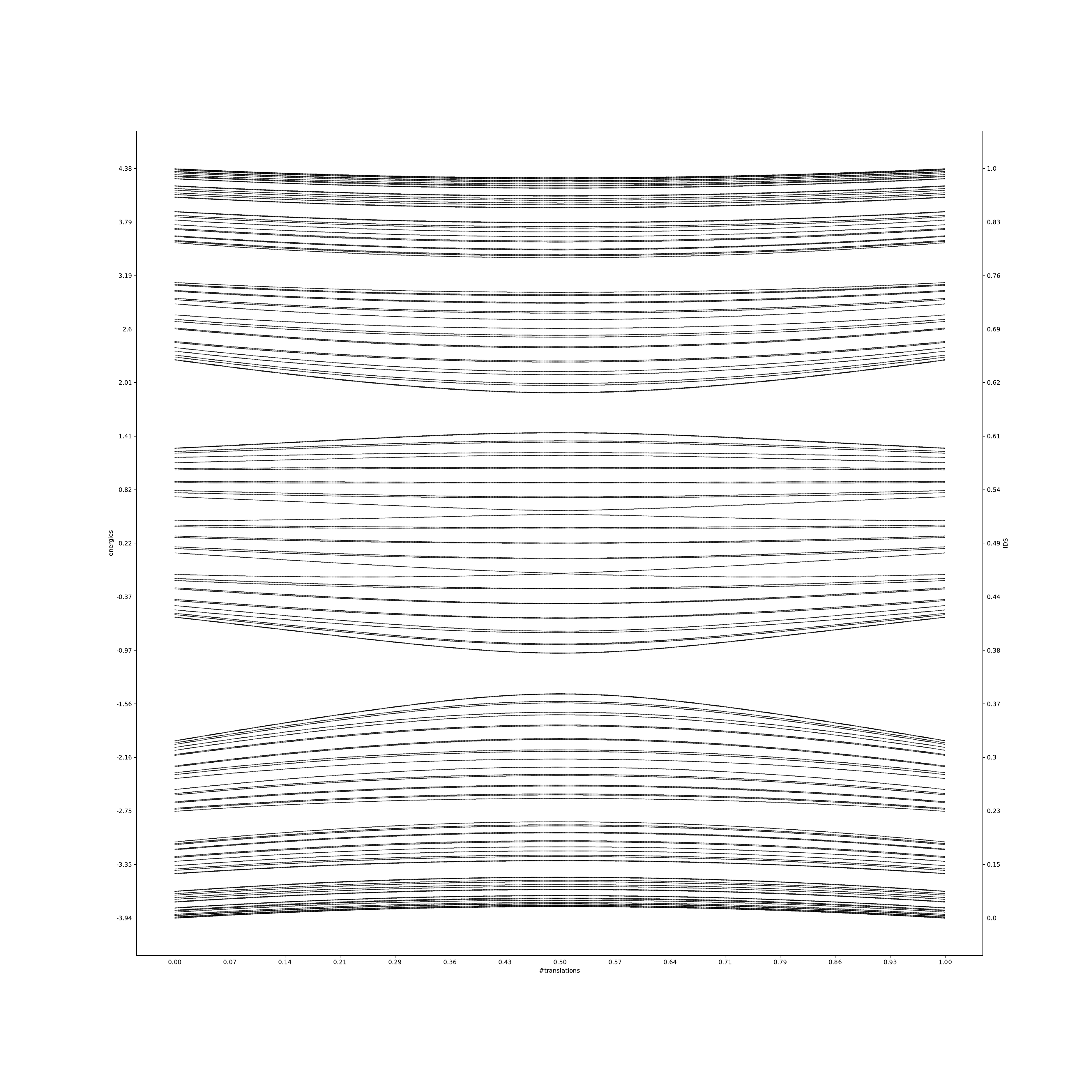}};
		
		\node at (0,-4) {\tiny $t$};
		
		\node[rotate=90] at (-3.8,-0.1) {\tiny energy};
		\node[rotate=90] at (3.8,0) {\tiny IDS};
		
		\node[scale=0.4] at (-3,-3.85) {$0$};
		\node[scale=0.4] at (-1.75,-3.85) {$0.21$};
		\node[scale=0.4] at (0,-3.85) {$0.5$};
		\node[scale=0.4] at (1.8,-3.85) {$0.79$};
		\node[scale=0.4] at (3,-3.85) {$1$};
		
		\node[scale=0.4] at (-3.5,-3.45) {$-3.9$};
		\node[scale=0.4] at (-3.5,-0.5) {$-0.2$};
		\node[scale=0.4] at (-3.5,2.4) {$4.4$};
		
		\node[scale=0.4] at (3.5,-3.45) {$0$};
		\node[scale=0.4] at (3.5,-0.5) {$0.5$};
		\node[scale=0.4] at (3.5,2.4) {$1$};
	\end{tikzpicture}
	
	\caption{Spectrum of $H^A_{\xi,t}$ for $t\in [\xi,1]$. 
		The $x$-axis corresponds to the mapping torus parameter $t\in [0,1]$. The vertical scale on the left $y$-axis corresponds to the energies and on the right to the IDS.}
	\label{spectrum_Xitilde_MT}
\end{figure}
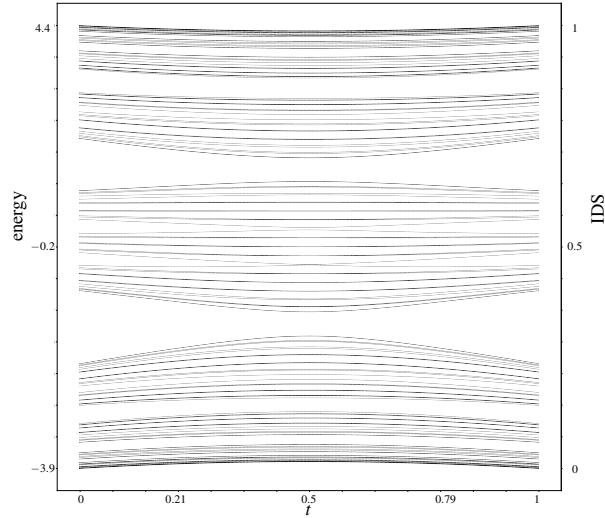

The spectrum does not depend on the choice of $\xi$ which is to be expected as the spectrum of $H_{\xi}$ does not depend on $\xi$. The spectrum of $H_\Xi$ can therefore by seen on the vertical line above $t=0$.
We observe that the gaps become smaller as $t$ goes to $\frac12$ but don't close (we make no attempt to give a mathematical explanation for this).

In order to simulate the spectrum of the half space Hamiltonians $\hat H^A_{\xi,t}$ we 
consider the parameter $t$ (whose variation was interpreted above as shifting the potential continuously relative to the boundary) in an interval $[0,T]$ where we recall that $(\phi,t+1) \sim (\phi+\theta,t)$. We take $T$ large but such that $T+D\le q $: in this way we do not create potential arrangements $V_0(T), V_0(T+1)\dots V_0(T+D-1)$ which are not already contained in $V_0(0), V_0(1)\dots V_0(q-1)$.  
Concretely, we take $\theta$ and $\lambda$ as above but $D=89$ and $T=q-D=55$. 
In Figure~\ref{spectrum_edge_MT} we see the spectrum of ${\hat H}^A_{\phi,t}$ for $t\in [0,T]$ which allows us to count the spectral flow of the eigenvalues associated to left edge states over $T=55$ translations.

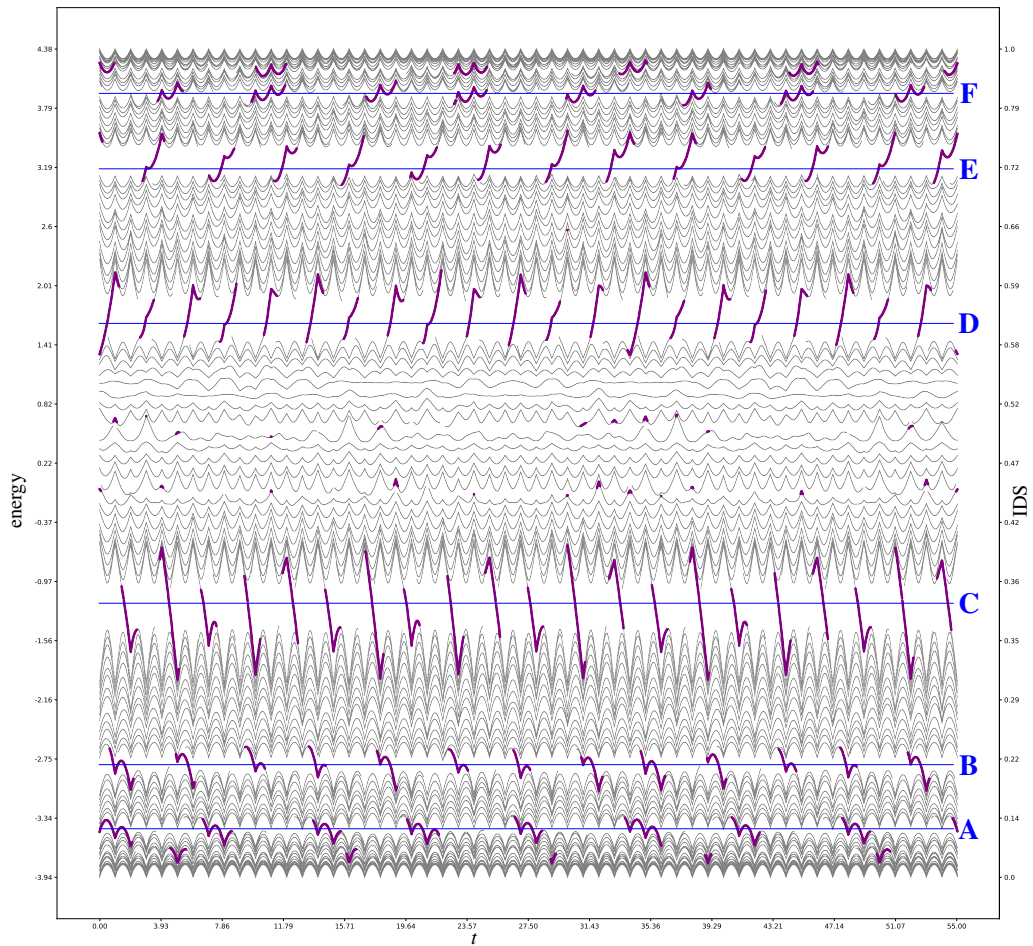
\begin{figure}
	\centering
	\includegraphics[width=1\textwidth]{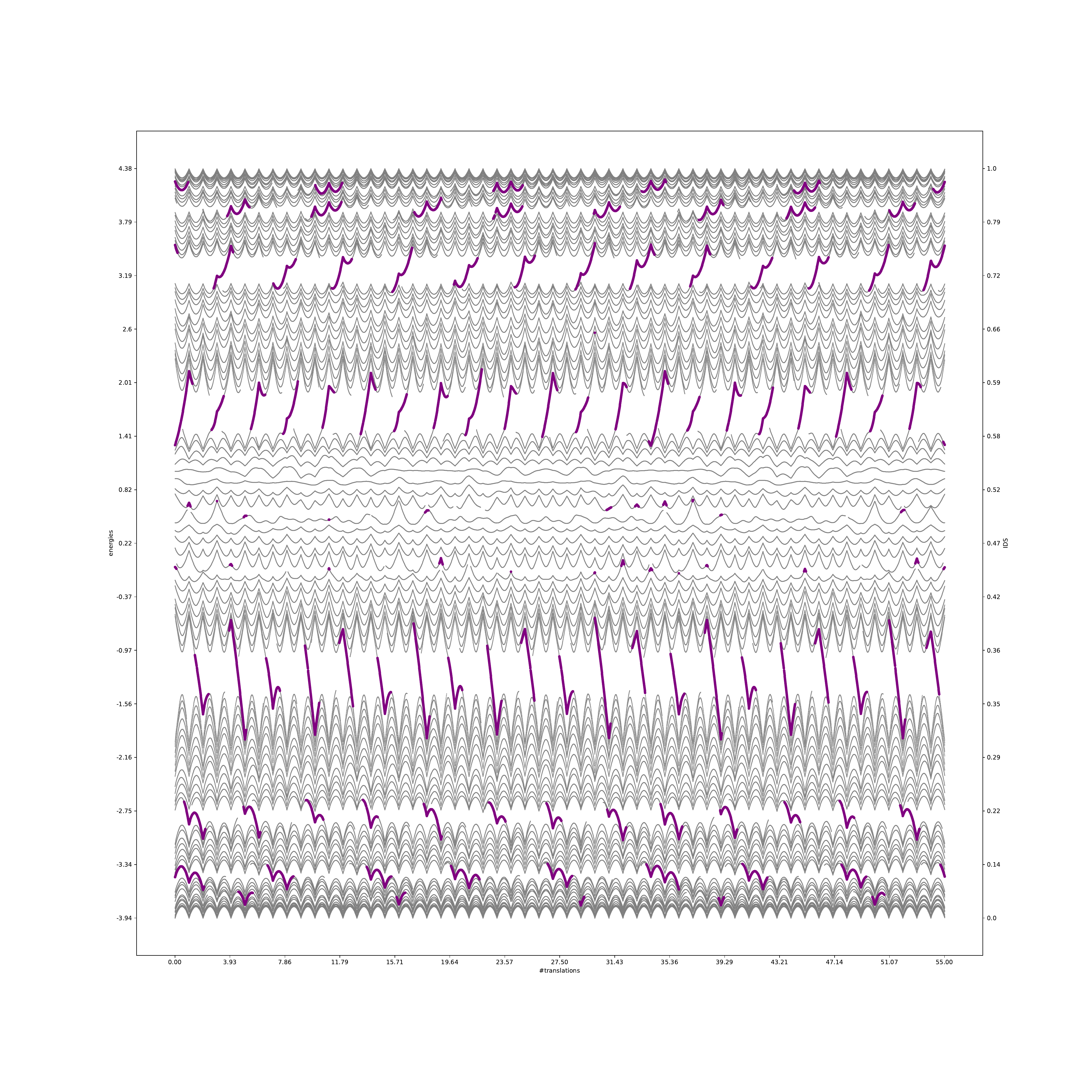}
	\begin{tikzpicture}[overlay]
		\node[anchor=center, color=blue, font=\bfseries] at (6,3.65) {A};    
		\draw[blue, line width=0.01pt] (-5.45,3.65) -- (5.8,3.65);
		\node[anchor=center, color=blue, font=\bfseries] at (6,4.5) {B};    
		\draw[blue, line width=0.01pt] (-5.45,4.5) -- (5.8,4.5);
		\node[anchor=center, color=blue, font=\bfseries] at (6,6.64) {C};    
		\draw[blue, line width=0.01pt] (-5.45,6.64) -- (5.8,6.64);
		\node[anchor=center, color=blue, font=\bfseries] at (6,10.35) {D};    
		\draw[blue, line width=0.01pt] (-5.45,10.35) -- (5.8,10.35);
		\node[anchor=center, color=blue, font=\bfseries] at (6,12.4) {E};    
		\draw[blue, line width=0.01pt] (-5.45,12.4) -- (5.8,12.4);
		\node[anchor=center, color=blue, font=\bfseries] at (6,13.4) {F};    
		\draw[blue, line width=0.01pt] (-5.45,13.4) -- (5.8,13.4);
		\fill[white] (-0.5,2.2) rectangle (0.5,2.3) ;
		\fill[white] (-6.5,8) rectangle (-6.3,9) ;
		\fill[white] (6.67,8) rectangle (6.9,9) ;
		
		\node[scale=0.7] at (-0.5,2.2) {$t$};
		\node[rotate=90,scale=0.7] at (-6.5,8) {energy};
		\node[rotate=90,scale=0.7] at (6.67,8) {IDS};
	\end{tikzpicture}
	\caption{Spectrum of $\hat{H}_{\tilde{\Xi}}$. The thicker lines (in purple in the online version) correspond to the part of the spectrum localised on the left boundary. The horizontal lines correspond to a choice of Fermi energy levels in the gaps. We observe the following spectral flow: \textcolor{blue}{A, $-8$}; \textcolor{blue}{B, $-13$}; \textcolor{blue}{C, $-21$}; \textcolor{blue}{D, $+21$}; \textcolor{blue}{E, $+13$}; \textcolor{blue}{F, $+8$}.}
	\label{spectrum_edge_MT}
\end{figure}

Comparing
for each gap, the IDS of the bulk Hamiltonian (a fraction with denominator equal to the matrix size $89$) with the observed spectral flow per $55$ translations we obtain
\begin{equation}
	\centering
	\renewcommand{\arraystretch}{1.5}
	\begin{tabular}{|c|c|c|c|c|c|c|}
		\hline
		Gap & A &  B &  C &D & E & F\\
		\hline
		IDS & $\frac{13}{89}$ &  $\frac{21}{89}$ &  $\frac{34}{89}$ &$\frac{55}{89}$ &  $\frac{68}{89}$ &  $\frac{76}{89}$ \\
		\hline
		SF per unit length & $-\frac{8}{55}$ &  $-\frac{13}{55}$ &  $-\frac{21}{55}$ &$\frac{21}{55}$ &  $\frac{13}{55}$ &  $\frac{8}{55}$\\
		\hline
	\end{tabular}
\end{equation}	
The IDS coincides with minus the spectral flow per unit length 	
up to the third digit after the comma modulo the ambiguity.
In particular, for the last three columns we need to substract 1 to have an equality. This confirms \eqref{eq-che-mapping-int}. Note that the sum of the spectral flows over all gaps is $0$.

\bigskip

In Figure~\ref{spectrum_edge_MT_proofslift} 
we see the spectrum of ${\hat H}^B_{\phi,t}$ for $t\in [0,55]$. This figure shows clearly that the gaps of the bulk model are not affected by the augmentation and also that the spectral flow per unit length of the eigenvalues of the left edge states coincides with (minus) the IDS of the bulk model without any ambiguity. Note that the values for the spectral flows over all gaps 
 do not add up to $0$.
\begin{figure}
	\centering
	\includegraphics[width=1\textwidth]{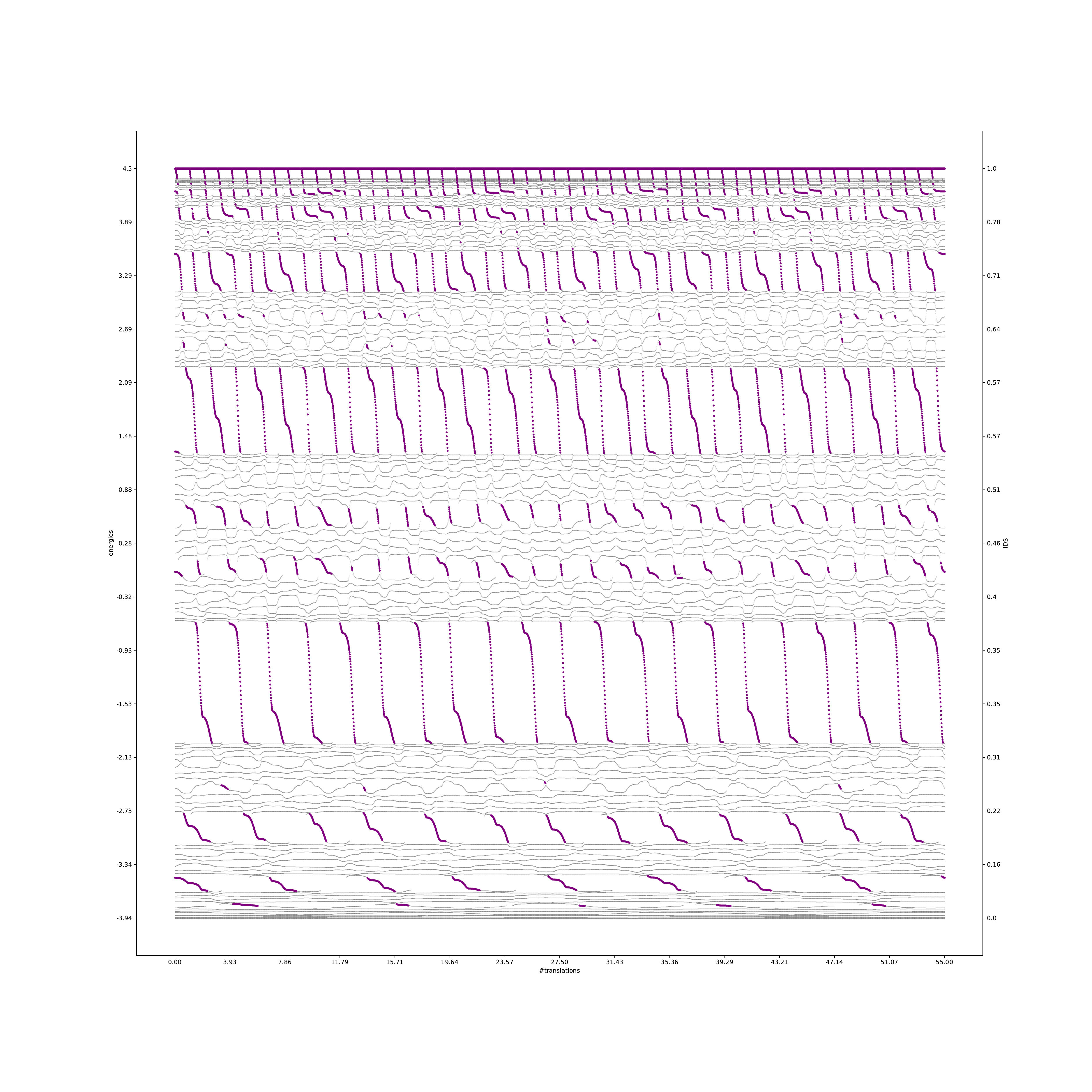}
	\begin{tikzpicture}[overlay]
		\node[anchor=center, color=blue, font=\bfseries] at (6,3.5) {A};    
		\draw[blue, line width=0.01pt] (-5.45,3.5) -- (5.8,3.5);
		\node[anchor=center, color=blue, font=\bfseries] at (6,4.35) {B};    
		\draw[blue, line width=0.01pt] (-5.45,4.35) -- (5.8,4.35);
		\node[anchor=center, color=blue, font=\bfseries] at (6,6.64) {C};    
		\draw[blue, line width=0.01pt] (-5.45,6.64) -- (5.8,6.64);
		\node[anchor=center, color=blue, font=\bfseries] at (6,10.35) {D};    
		\draw[blue, line width=0.01pt] (-5.45,10.35) -- (5.8,10.35);
		\node[anchor=center, color=blue, font=\bfseries] at (6,12.4) {E};    
		\draw[blue, line width=0.01pt] (-5.45,12.4) -- (5.8,12.4);
		\node[anchor=center, color=blue, font=\bfseries] at (6,13.35) {F};    
		\draw[blue, line width=0.01pt] (-5.45,13.35) -- (5.8,13.35);
		
		\fill[white] (-0.5,2.2) rectangle (0.5,2.3) ;
		\fill[white] (-6.5,8) rectangle (-6.3,9) ;
		\fill[white] (6.67,8) rectangle (6.9,9) ;
		
		\node[scale=0.7] at (-0.5,2.2) {$t$};
		\node[rotate=90,scale=0.7] at (-6.5,8) {energy};
		\node[rotate=90,scale=0.7] at (6.67,8) {IDS};
	\end{tikzpicture}
	\caption{Spectrum of $\hat{H}_{\tilde{\Xi}}$. The thicker lines (in purple in the online version) correspond to eigenvalues of states which are localized at the left boundary. The horizontal lines correspond to a choice of Fermi energy in the gaps. We observe the following spectral flow: \textcolor{blue}{A, $-8$}; \textcolor{blue}{B, $-13$}; \textcolor{blue}{C, $-21$}; \textcolor{blue}{D, $-34$}; \textcolor{blue}{E, $-42$}; \textcolor{blue}{F, $-47$}.}
	\label{spectrum_edge_MT_proofslift}
\end{figure}
\color{black}

\subsection{Augmentation of the 2-cuts model with arcs} We consider here the augmentation defined by adding arcs to join the doubled points. We focus on the case in which $\cut$ is not a multiple of $\theta$, as these are the new results compared to \cite{kellendonk_prodan}. 
The augmentation $\tilde h$ of $h$ is defined by replacing $\chi_{[0, \cut]}$ with the function which takes value $1$ on the interval $[(0,+),(\cut,-)]$, value $0$ on the interval $[(\cut,+),(1,-)]$, and increases from $0$ to $1$ on the added arc $\{(0,t)| t\in [0,1]\}$ while it decreases from $1$ to $0$ on the added arc $\{(\cut,t)| t\in [0,1]\}$. $H_{\phi,t}=\rho_{\phi,t}(\tilde h)$ is thus given on $\ell^2(\Z)$ by 
$$H_{\phi,t} \Psi(n)=\Psi(n+1)+\Psi(n-1)+V_{\phi,t}(n) \Psi(n)$$
with 
\begin{equation}
	V_{\phi,t}(n) = \left\{\begin{array}{ll}
\lambda & \mbox{if } (\cut,1) < (\{n\theta+\phi\},t) < (1,0) \\
-\lambda & \mbox{if } (0,1) < (\{n\theta+\phi\},t) < (\cut,0)\\
\lambda (1-2t) & \mbox{if } \{n\theta+\phi\} = 0\\
\lambda (2t-1) & \mbox{if } \{n\theta+\phi\} = \cut
\end{array}\right.
\end{equation}
Finally, the lift $\hat h$ of $\tilde h$ in the Toeplitz extension, which is obtained by replacing in $\tilde h$ the element $u$ with $\hat u$ and $u^*$ with ${\hat u}^*$ is represented by the family of operators 
$\hat\rho_{\phi,t}(\hat h)=\hat{H}_{\phi,t}$ which are the restrictions of $H_{\phi,t}$ to $\ell^2(\N)$.

For the simulations, we consider a multiple $D$ of $\q$ which is large enough to distinguish the edge states when we enforce Dirichlet boundary conditions. On the other hand, we take a small value for $\q$ in order to see properly the effects we are going to describe. 
Concretely, we fix the parameters $\theta = \frac{2}{5}$, $\cut=\frac{1}{3}$, $\lambda=1$ and $D=35$. 

\subsubsection{The IDS of the bulk system}
In Figure~\ref{spectrum_bulk_Xi} we see the spectrum of the bulk operator $H_\xi$. The jumps in the spectrum are at the cut points. 
\begin{figure}
	\centering
	\includegraphics[width=0.8\textwidth]{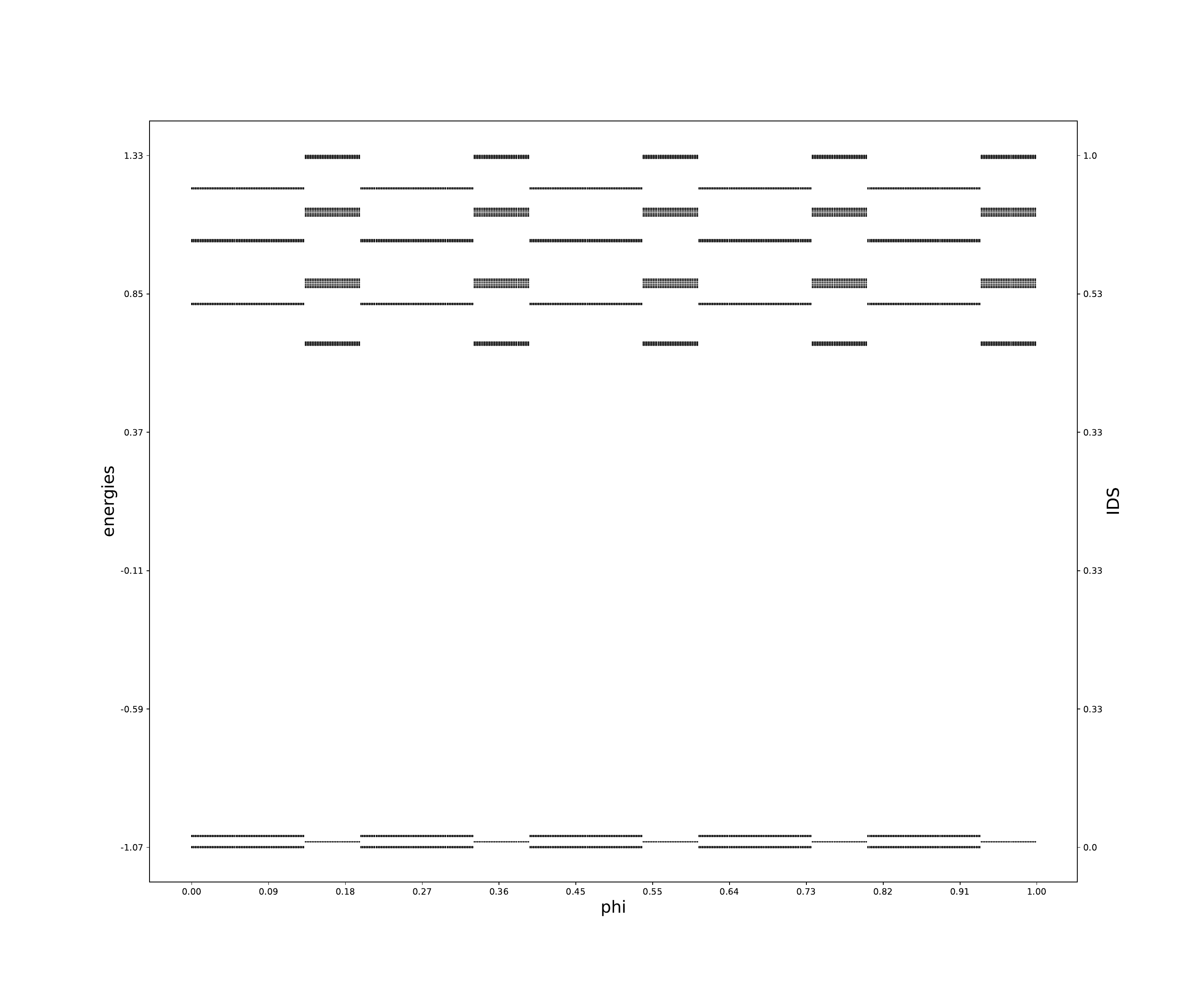}
	\caption{Spectrum of $H_\Xi$. The $x$-axis corresponds to the parameter $\xi$: for each value $\xi$ we plot the spectrum of $H_\xi$. The $y$-axis corresponds to energy. At the right we denote the values of the IDS at the gaps.}
	\label{spectrum_bulk_Xi}
	\begin{tikzpicture}[overlay]	
		\fill[white] (-0.5,3.1) rectangle (0.5,3.3) ;
		\fill[white] (-5.4,6) rectangle (-5.1,9) ;
		\fill[white] (5.4,6) rectangle (5.7,9) ;
		
		\node[scale=0.7] at (0,3.1) {$\xi$};
		\node[rotate=90,scale=0.7] at (-5.4,7.5) {energy};
		\node[rotate=90,scale=0.7] at (5.4,7.5) {IDS};
	\end{tikzpicture}
\end{figure}
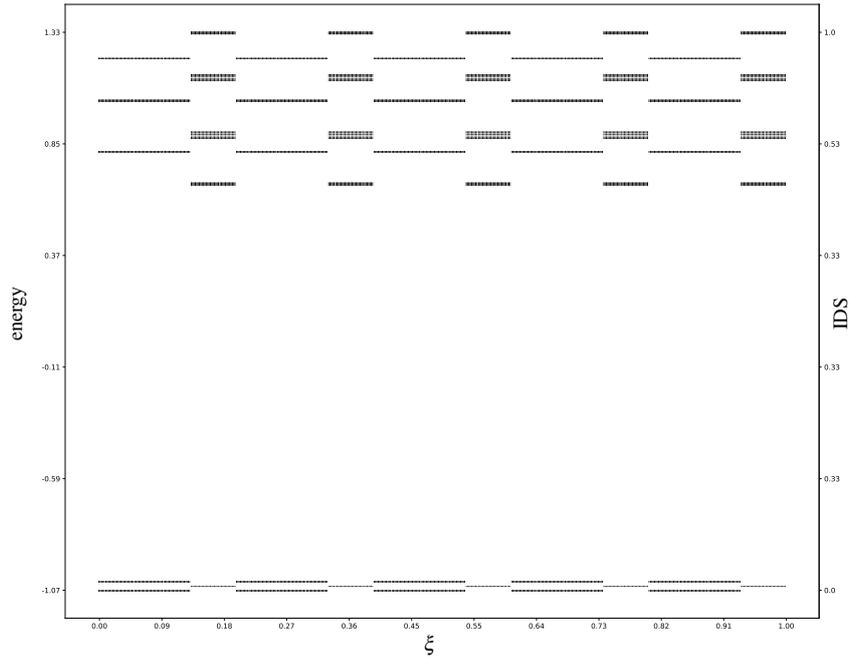

It is important to realize that the spectrum of $h$, which is the union of the spectra of the members $H_\xi$ of the family $H_\Xi$ (and which we refer to as the spectrum of $H_\Xi$), is no longer given by the spectrum of a single member, but it is the union of two of its members, for instance
$$\mathrm{spec } H_\Xi =  \mathrm{spec } H_{0,+} \cup \mathrm{spec } H_{1,-}.$$
It therefore has $2q-1$ interior gaps and not only $q-1$. Indeed, when approximating the topological results for irrational $\theta$ by means of a rational value, we need to approximate the hull, or, what amounts to the same, the covariant family defined by the hull, and not a single member of that family. An consequence of that is, that the IDS of the family $H_\Xi$ at some energy $E$ in some gap is, by definition, the $\mu$-average of the values of the integrated density of states at $E$ for the individual members of the family. In particular, the IDS of the family $H_\Xi$ is either of the form
\begin{equation}\label{eq-primary-gap}
 \mathrm{IDS}(E)= \frac{m}{\q}
 \end{equation}
 for some $m\in\N$, or of the form
 \begin{equation}\label{eq-secondary-gap}
	\mathrm{IDS}(E)= \frac{m+1}{q}\{q\cut\}+\frac{m}{q}(1-\{q\cut\})
 \end{equation}
 for some $m\in\N$. Gaps to which the first option applies will be called \emph{primary}, the others \emph{secondary}.
Indeed, we have $2q$ cut points, and when we move over these points there are $D/\q$ sites of the chain where the potential jumps from $-1$ to $1$ (or in the other direction). Therefore, on a fraction $\{q\cut\}$ of values for $\phi$ the energy of the bands of $H_{\phi,\epsilon}$ is lower and therefore, for some gaps, these operators have one more band below the value of the energy $E$. This can nicely be seen in 
Figure~\ref{spectrum_bulk_Xitilde} and rigorously be shown along the lines of the proof of Prop.~\ref{prop-n2}.
Averaging on $\phi$, we get the formula.

For our choice of parameters, we see that the IDS takes the values 
$$	\mathrm{IDS}\in \left\{ \left.\frac{m}{5}, \frac{2}{15}+\frac{m}{5} \right| m\in \{0,\cdots,4\} \right\}.$$
Solving the Diophantine equation \eqref{eq-dioph} we get 
the following solutions for the gaps. 
\begin{equation}
	\renewcommand{\arraystretch}{1.5}
	\begin{array}{c | c c c c c c c c c c} 
	\mbox{Gap} & \textcolor{orange}{\textbf{A}} & \textcolor{magenta}{\textbf{a}} & \textcolor{orange}{\textbf{B}} & \textcolor{magenta}{\textbf{b}}& \textcolor{orange}{\textbf{C}} & \textcolor{magenta}{\textbf{c}} & \textcolor{orange}{\textbf{D}} & \textcolor{magenta}{\textbf{d}}& \textcolor{orange}{\textbf{E}} & \textcolor{magenta}{\textbf{e}} \\
		\hline
		\mathrm{IDS} & \frac{2}{15} & \frac{1}{5} & \frac{5}{15} & \frac{2}{5} & \frac{8}{15} & \frac{3}{5} & \frac{11}{15} & \frac{4}{5} & \frac{14}{15} & 1 \\
		\hline
		n_1\!\!\!\!\!\mod5 & 2 & 3 & 0 & 1 & 3 & 4 & 1 & 2 & 4 & 0 \\
		\hline
		n_2 & 1 & 0 & 1 & 0 & 1 & 0 & 1 & 0 & 1 & 0 \\
		\hline
	\end{array}
	\label{table_solution}
\end{equation}
We label the primary gaps with small and the secondary with capital letters for better comparison with the numerical simulations. We see that IDS does not take all the values $\frac{1}{15}\Z\cap [0,1]$. This is a consequence of Prop.~\ref{prop-n2} which depends on the particularity of our model.

\subsubsection{The spectrum of the augmented family and $\chab$}
In Figure~\ref{spectrum_bulk_Xitilde} we see the simulation of the spectrum of the augmented family $H_{\tilde\Xi}$. We observe that the primary gaps are still open, while the secondary gaps are closed. To better understand the different nature of these two types of gaps, we recall the interpretation \eqref{physical_interpr_n2_c} of $n_2$ as the spectral flow of the eigenvalues of the augmented bulk system through the gaps per unit length. In this way, we can characterize primary gaps as those for which $n_2=0$ (no spectral flow of the augmented bulk system), while for the secondary gaps we observe that $n_2 = 1$. Indeed, to read off $n_2$, we take the difference of the spectral flow over the arcs corresponding to the orbit of $\cut$ (the red intervals on the $x$-axis) with the spectral flow over the arcs corresponding to the orbit of $0$ and then divide by $2q$.
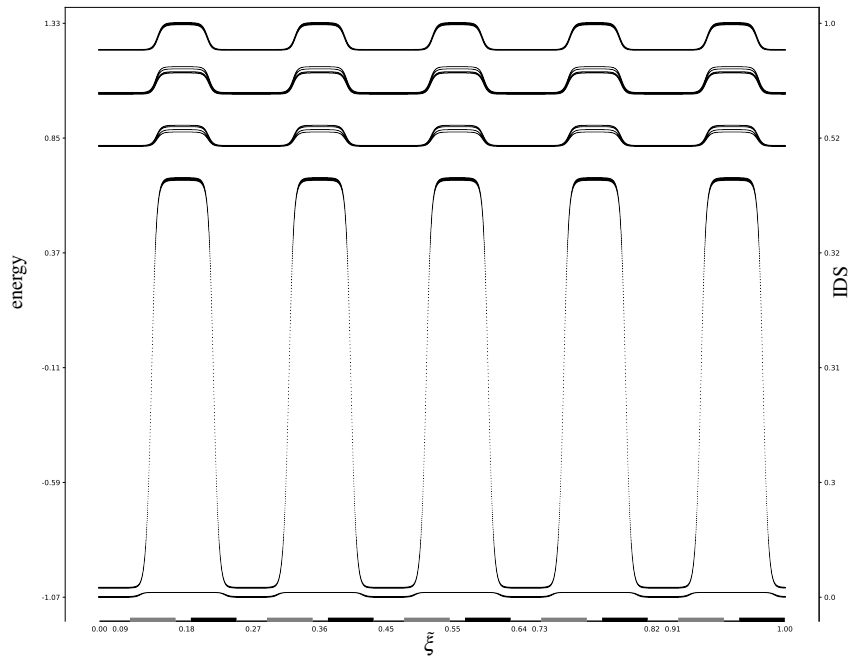
\begin{figure}
	\centering
	\includegraphics[width=0.8\textwidth]{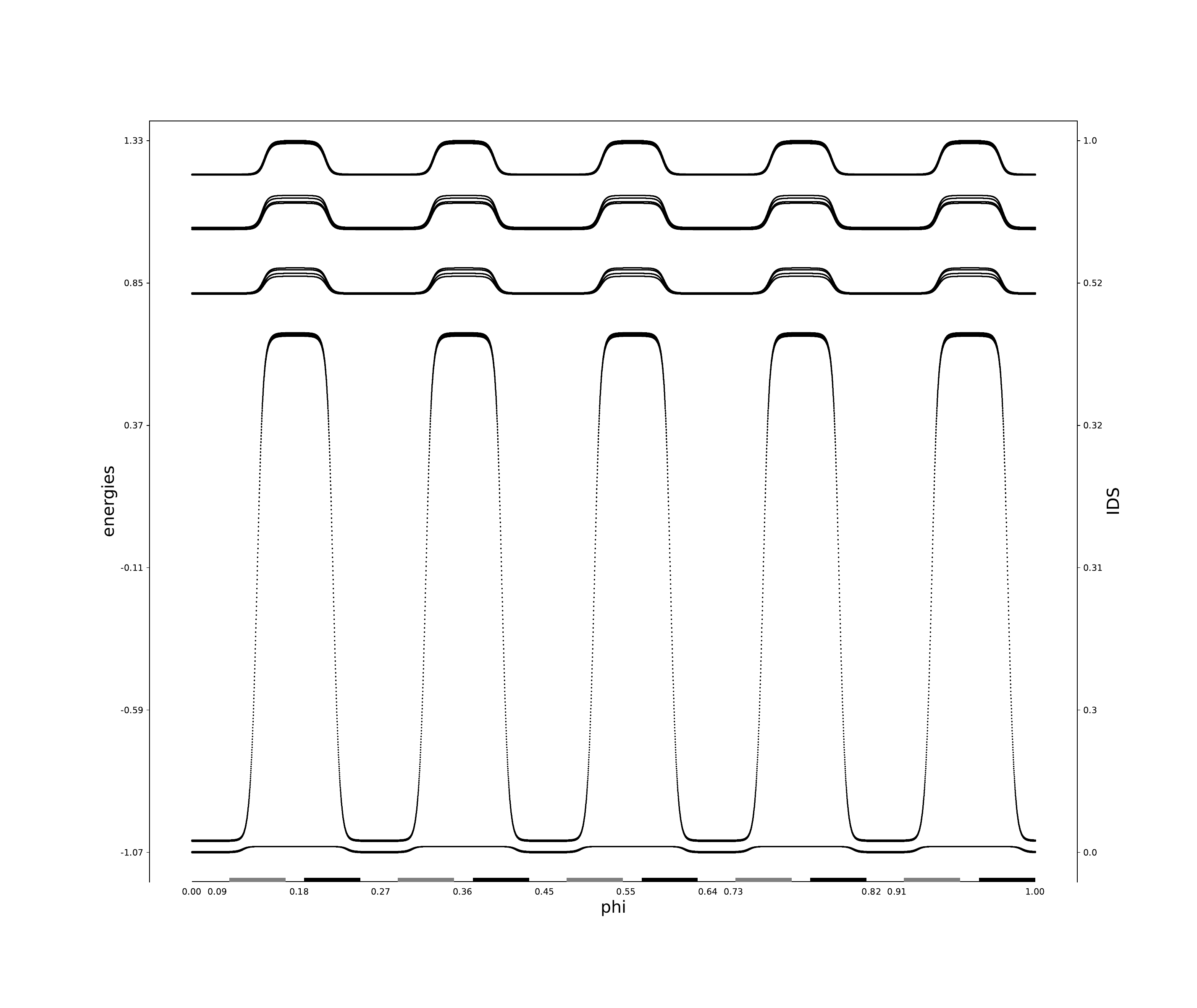}
	\caption{Spectrum of $H_{\tilde{\Xi}}$. The black (resp. gray) thick intervals along the $x$-axis correspond to the arcs on the singular points of the orbit of $0$ (resp. $ \cut$).}
	\label{spectrum_bulk_Xitilde}
	\begin{tikzpicture}[overlay]	
		\fill[white] (-0.5,2.6) rectangle (0.5,2.85) ;
		\fill[white] (-5.4,6) rectangle (-5.1,9) ;
		\fill[white] (5.4,6) rectangle (5.7,9) ;
		
		\node[scale=0.7] at (0,2.7) {$\tilde\xi$};
		\node[rotate=90,scale=0.7] at (-5.4,7.5) {energy};
		\node[rotate=90,scale=0.7] at (5.4,7.5) {IDS};
	\end{tikzpicture}
\end{figure}

\subsubsection{\texorpdfstring{Edge spectral flow in the primary gaps}{}}
In Figure~\ref{primary_gaps_WN} we display the spectrum of  $\hat H_{\tilde\Xi}$. The thicker parts of the $x$-axis correspond to the added arcs. There are $2\q = 10$ of them, $5$ of which are in light grey and correspond to arcs along the orbit of $\cut=\frac13$ while the other five are in dark grey and correspond to arcs along the orbit of $0$. 
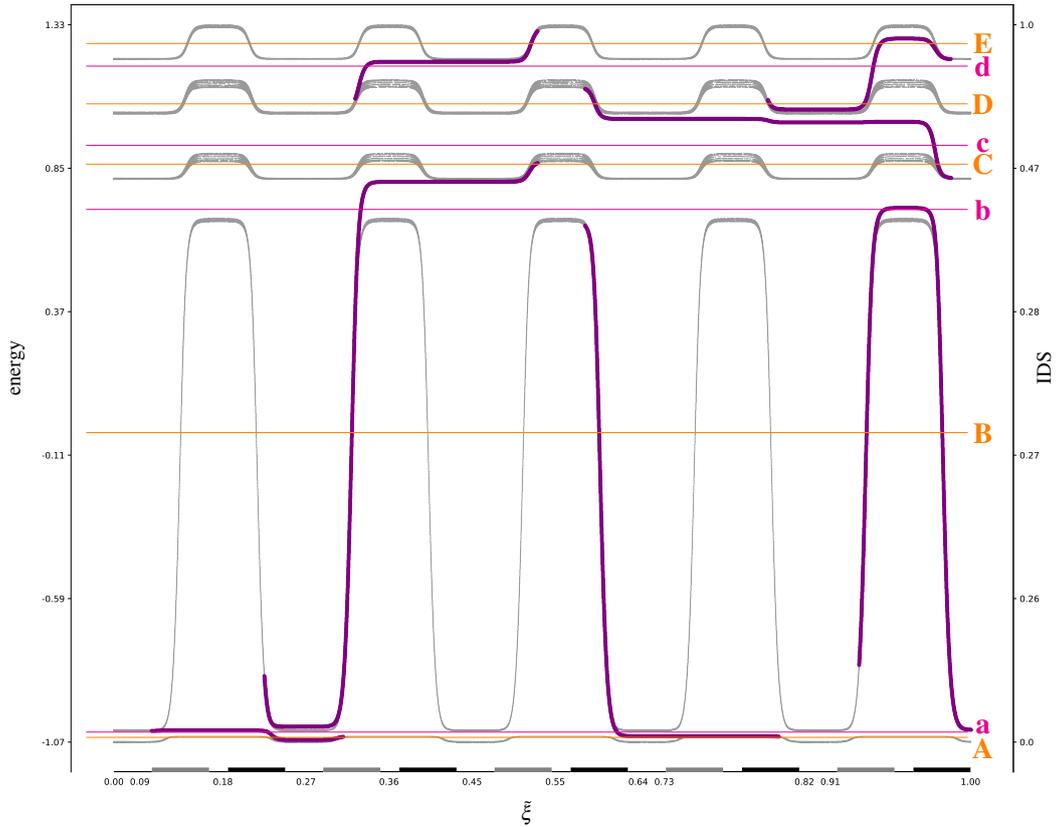
\begin{figure}
	\centering
	\includegraphics[width=1.\textwidth]{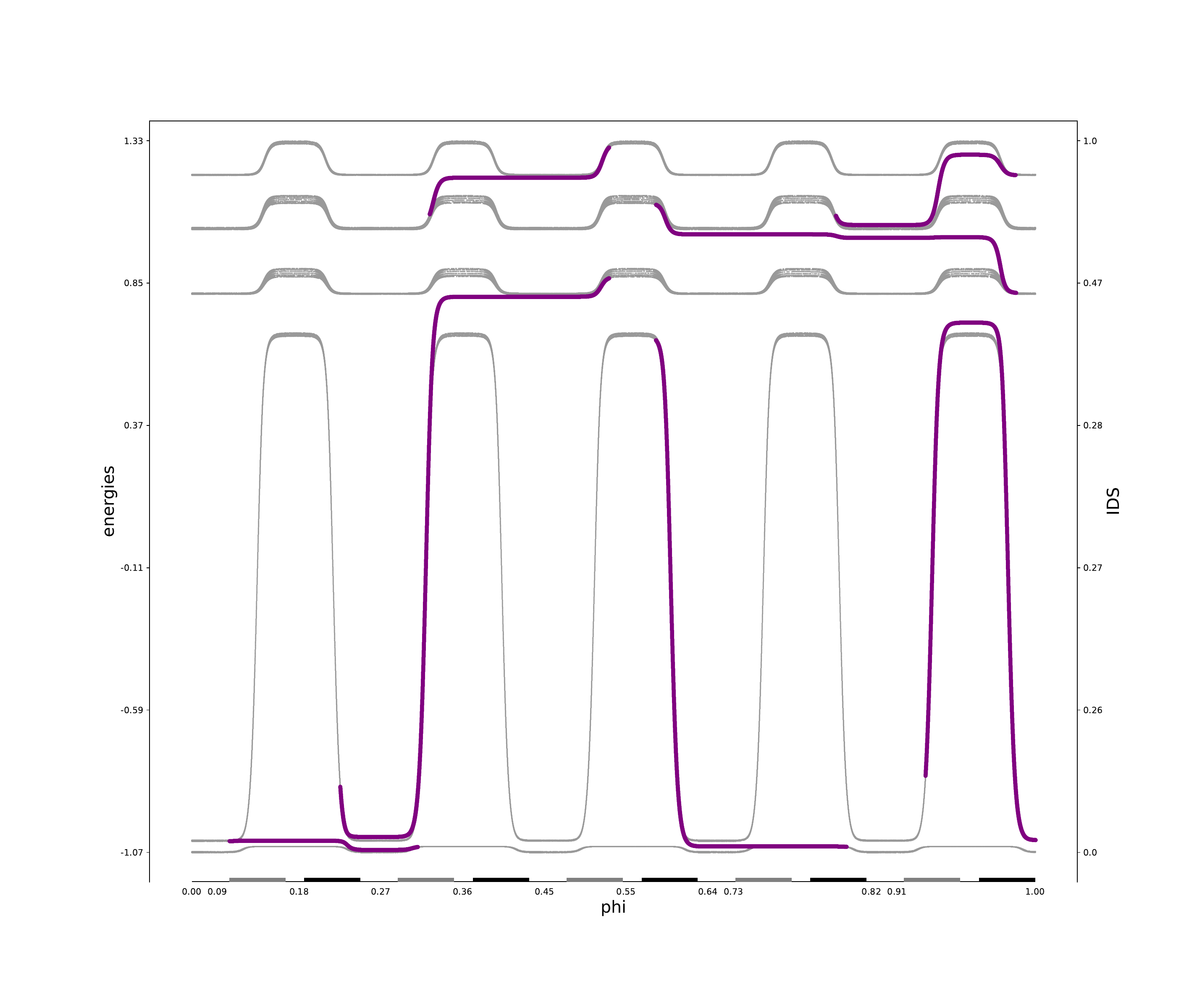} 
	\begin{tikzpicture}[overlay]
		\node[anchor=center, color=orange, font=\bfseries] at (6,2.45) {A};    
		\draw[orange, line width=0.2pt] (-5.8,2.6) -- (5.8,2.6);
		\node[anchor=center, color=orange, font=\bfseries] at (6,6.64) {B};    
		\draw[orange, line width=0.2pt] (-5.8,6.64) -- (5.8,6.64);
		\node[anchor=center, color=orange, font=\bfseries] at (6,10.2) {C};    
		\draw[orange, line width=0.2pt] (-5.8,10.2) -- (5.8,10.2);
		\node[anchor=center, color=orange, font=\bfseries] at (6,11) {D};    
		\draw[orange, line width=0.2pt] (-5.8,11) -- (5.8,11);
		\node[anchor=center, color=orange, font=\bfseries] at (6,11.8) {E};    
		\draw[orange, line width=0.2pt] (-5.8,11.8) -- (5.8,11.8);
		
		\node[anchor=center, color=magenta, font=\bfseries] at (6,2.74) {a};    
		\draw[magenta, line width=0.05pt] (-5.8,2.67) -- (5.8,2.67);
		\node[anchor=center, color=magenta, font=\bfseries] at (6,9.6) {b};    
		\draw[magenta, line width=0.05pt] (-5.8,9.6) -- (5.8,9.6);
		\node[anchor=center, color=magenta, font=\bfseries] at (6,10.45) {c};    
		\draw[magenta, line width=0.05pt] (-5.8,10.45) -- (5.8,10.45);
		\node[anchor=center, color=magenta, font=\bfseries] at (6,11.5) {d};    
		\draw[magenta, line width=0.05pt] (-5.8,11.5) -- (5.8,11.5);
		
		\fill[white] (-0.5,1.6) rectangle (0.5,1.9) ;
		\fill[white] (-6.7,6) rectangle (-6.4,9) ;
		\fill[white] (6.75,6) rectangle (7,9) ;
		
		\node[scale=0.7] at (0,1.6) {$\tilde\xi$};
		\node[rotate=90,scale=0.7] at (-6.7,7.5) {energy};
		\node[rotate=90,scale=0.7] at (6.8,7.5) {IDS};
	\end{tikzpicture}
	\caption{Spectrum of $\hat{H}_{\tilde{\Xi}}$. The thicker lines (in purple in the online version) correspond to the eigenvalues of $\hat{H}_{\tilde{\Xi}}$ whose eigenfunctions are localised on the left boundary. The eigenvalues with eigenfunctions localised on the right boundary are supressed. 
	The horizontal lines correspond to a choice of Fermi energy levels in the gaps: lower case letters identify primary gaps and uppercase secondary gaps. We observe the following spectral flow in the primary gaps: \textcolor{magenta}{a, $-2$}; \textcolor{magenta}{b, $+1$}; \textcolor{magenta}{c, $-1$}; \textcolor{magenta}{d, $+2$}. The values are in agreement with those of Table~\ref{table_solution}.}
	\label{primary_gaps_WN}  
\end{figure}
In primary gaps we can easily determine  graphically the spectral flow of the left edge states by 
counting the number (with sign) of spectral values of $\hat{H}_{\tilde{\Xi}}$ traversing the fiducial line in the gap of $H_{\tilde{\Xi}}$. This spectral flow corresponds to $n_1$, see \eqref{eq-SF-edge} and the observed results agree with those of (\ref{table_solution}).
While also in some secondary gaps there are eigenvalues of left edge states, these are not expected to be topologically stable, and indeed, such a count would seem to depend on the precise position of the fiducial line.

\subsubsection{Edge spectral flow for secondary gaps} We have seen in Section~\ref{sec-secondary-gaps} that the edge invariant for a secondary gap $\Delta$ can be obtained as for a primary gap upon replacing $h$ by a stacked operator $h^s$. Focussing on the model with $\lambda>0$, Prop.~\ref{prop-n2} shows that $n_2=1$ and so the augmentation of $h^s$ has the form 
$$\tilde h^s:=\begin{pmatrix}
	\tilde h & \tilde c \\ \tilde c^* & -\tilde h' 
\end{pmatrix}$$
with $h'=\lambda'(1-\chi_{[0,\cut]})$, $\lambda'>0$ and coupling $\tilde c$ chosen in such a way, that $\tilde h^s$ has a gap inside the gap $\Delta$. 
$\tilde h^s$ is represented by the covariant family $H^s_{\tilde\xi}=\rho^{(2)}_{\tilde\xi}(\tilde h^s)$, $\tilde \xi\in\tilde\Xi$ where  $\rho^{(2)}_{\tilde\xi} = \rho_{\tilde\xi}\otimes\mathrm{id}$ is the doubled representation on $\ell^2(\Z)\otimes\C^2$. The augmented periodic operator is represented by the matrix   
\[
H_{\phi,t}^{s}(k) = \begin{pmatrix}
	A_{\phi,t}(0) & J & 0 & \cdots & 0 & \mathrm{e}^{-2\pi k\mathrm{i}}J \\
	J & A_{\phi,t}(1) & J & \cdots & 0 & 0 \\
	0 & J & A_{\phi,t}(2) & \cdots & 0 & 0 \\
	\vdots & \vdots & \vdots & \ddots & \vdots & \vdots \\
	0 & 0 & 0 & \cdots & A_{\phi,t}(D-2) & J \\
	\mathrm{e}^{2\pi k\mathrm{i}}J  & 0 & 0 & \cdots & J & A_{\phi,t}(D-1)
\end{pmatrix}
\]
where 
$$A_{\phi,t}(n)=\begin{pmatrix}
	V_{\phi,t}(n) & \tilde c_{\phi,t}\\
	\tilde c_{\phi,t} & -\frac{\lambda'}{\lambda} V_{\phi,t}(n) 
\end{pmatrix}, \quad J =\begin{pmatrix}
	1 & 0\\
	0 & 0 
\end{pmatrix}.$$
For the coupling $\tilde c_{\phi,t}$ we take a bell-shaped function in the variable $t$ which vanishes at $t=0$ and $t=1$. Replacing the upper-right and lower-left blocks by $0$, we obtain the matrix representing $\hat{H}_{\phi,t}^{s}$ as the corresponding operator with Dirichlet conditions. We use $\lambda'=1.7$ and $\|c\| = 0.4$. In Figure~\ref{secondary_gapsWN} we  display the spectrum of $\hat{H}_{\tilde{\Xi}}^{s}$. 
The coupling opens the secondary gaps except for gap A. To better read off the spectral flow of the eigenvalues which have eigenfunctions which are localized at the left edge, we provide a zoom on the last three secondary gaps in which we encircle the crossing of the eigenvalues with a chosen energy line in the gap. Counting these crossings we obtain the spectral flow of these eigenvalues. The values agree with those of \eqref{table_solution}.
\begin{figure}
	\centering
	\includegraphics[width=1.\textwidth]{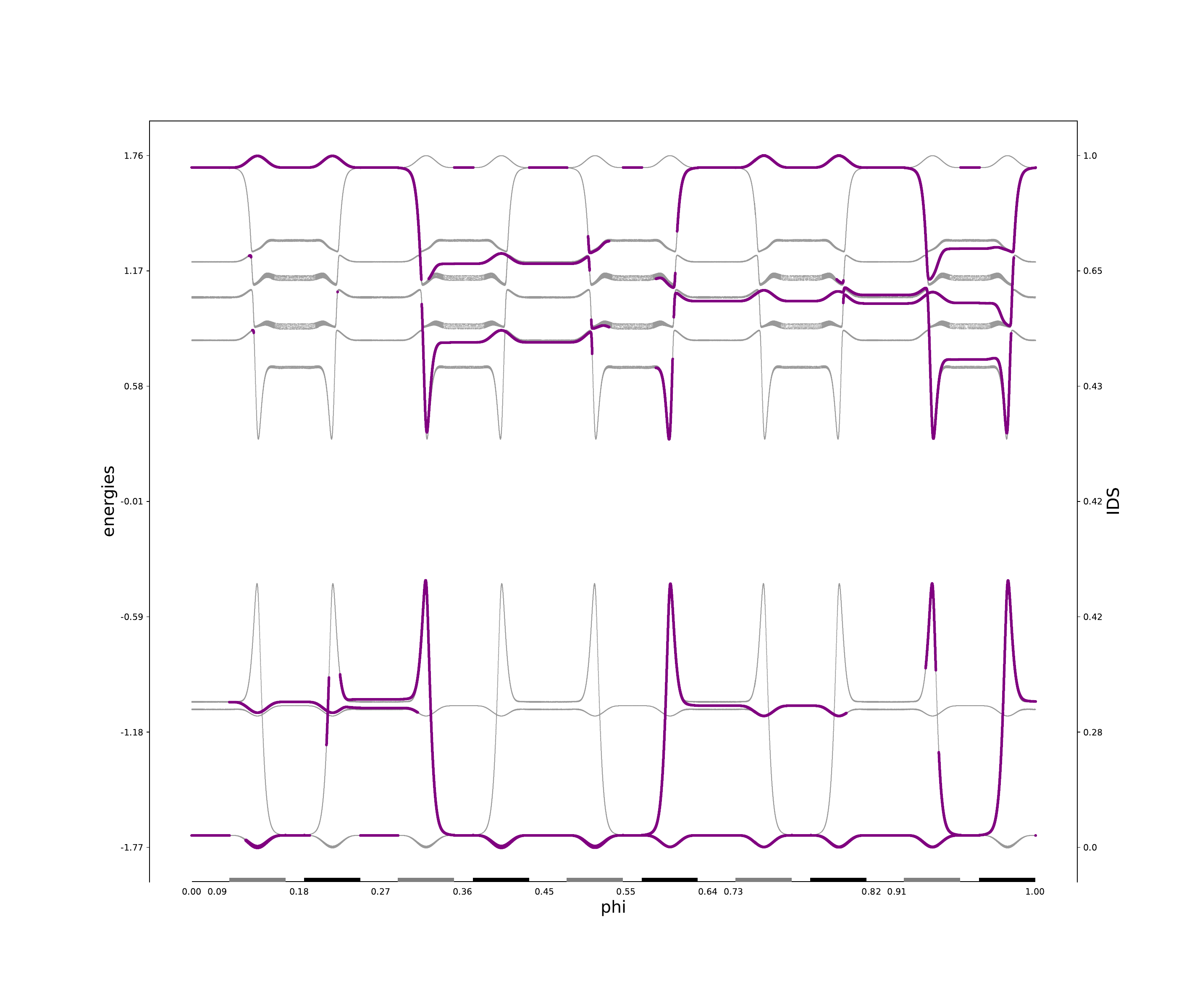}
	\begin{tikzpicture}[overlay]
		\node[anchor=center, color=orange, font=\bfseries] at (6,7.2) {B};
		\draw[orange, line width=0.05pt] (-5.8,7.2) -- (5.8,7.2);
		\node[anchor=center, color=orange, font=\bfseries] at (6,9.55) {C};
		\draw[orange, line width=0.05pt] (-5.8,9.55) -- (5.8,9.55);
		\node[anchor=center, color=orange, font=\bfseries] at (6,10.1) {D};
		\draw[orange, line width=0.05pt] (-5.8,10.1) -- (5.8,10.1);
		\node[anchor=center, color=orange, font=\bfseries] at (6,10.55) {E};
		\draw[orange, line width=0.05pt] (-5.8,10.55) -- (5.8,10.55);
		
		\fill[white] (-0.5,1.6) rectangle (0.5,1.9) ;
		\fill[white] (-6.7,6) rectangle (-6.4,9) ;
		\fill[white] (6.75,6) rectangle (7,9) ;
		
		\node[scale=0.7] at (0,1.6) {$\tilde\xi$};
		\node[rotate=90,scale=0.7] at (-6.7,7.5) {energy};
		\node[rotate=90,scale=0.7] at (6.85,7.5) {IDS};
	\end{tikzpicture}  
	\caption{Spectrum of $\hat{H}_{\tilde{\Xi}}^{s}$. The thicker lines (in purple in the online version) correspond to  the eigenvalues of $\hat{H}_{\tilde{\Xi}}$ whose eigenfunctions are localized on the left boundary. The horizontal lines correspond to a choice of Fermi energy levels in the secondary gaps. We observe the following spectral flow: \textcolor{orange}{B, $0$}; \textcolor{orange}{C, $-2$}; \textcolor{orange}{D, $+1$}; \textcolor{orange}{E, $-1$}.}
	\label{secondary_gapsWN}  
\end{figure}
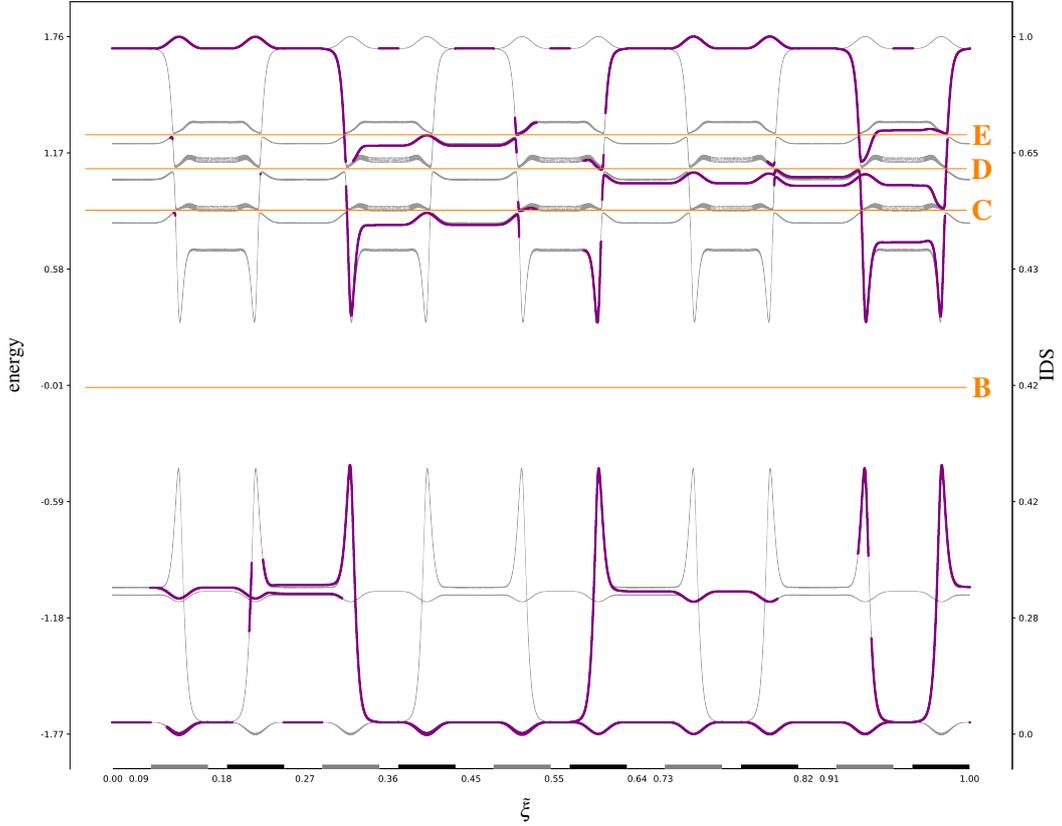
\begin{figure}
	\centering
\includegraphics[width=1.\textwidth]{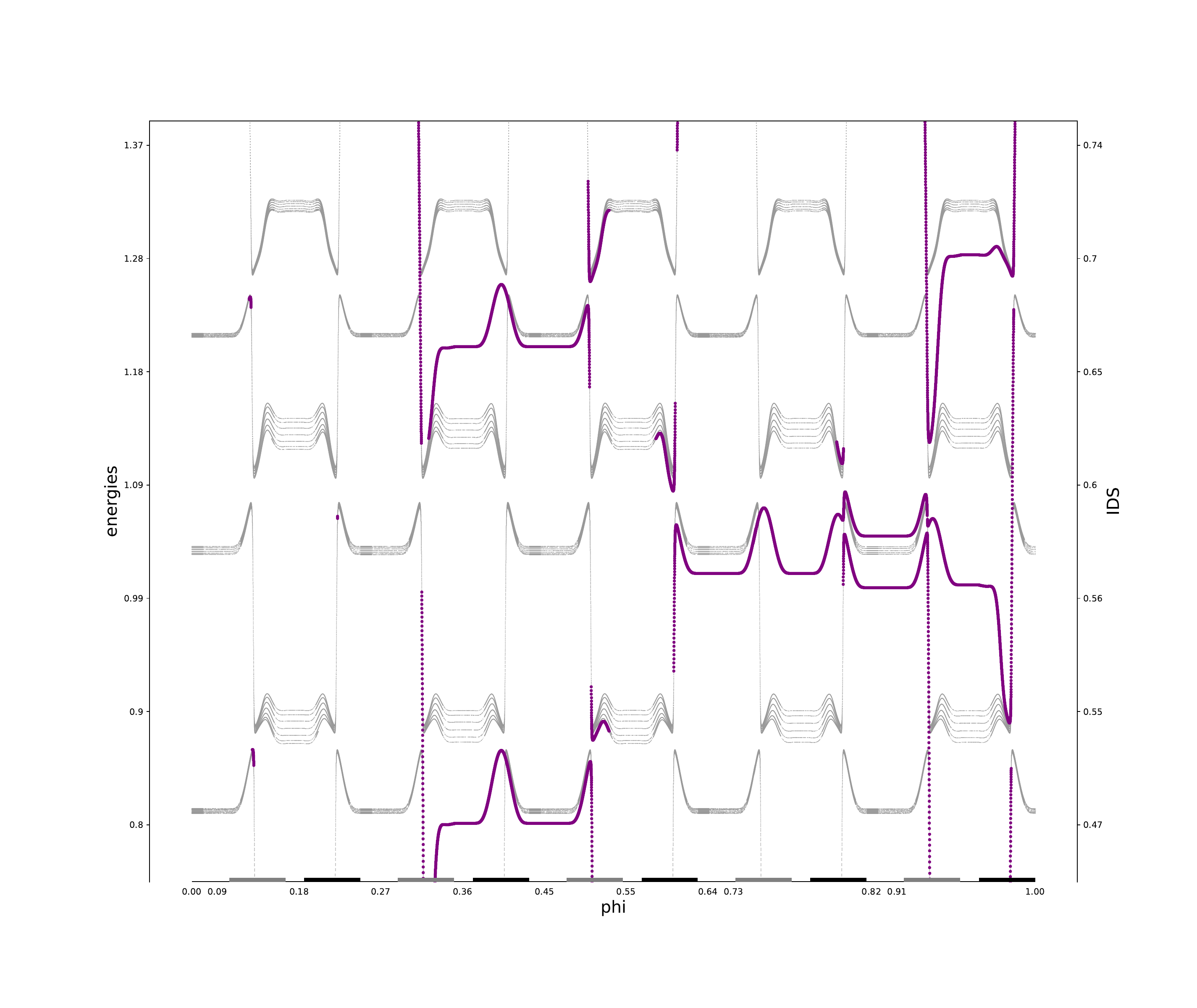}
	\begin{tikzpicture}[overlay]
		\node[anchor=center, color=orange, font=\bfseries] at (6,4) {C};
		\draw[orange, line width=0.01pt] (-5.8,4) -- (5.8,4);
		\node[anchor=center, color=orange, font=\bfseries] at (6,7.3) {D};
		\draw[orange, line width=0.01pt] (-5.8,7.33) -- (5.8,7.33);
		\node[anchor=center, color=orange, font=\bfseries] at (6,10.05) {E};
		\draw[orange, line width=0.01pt] (-5.8,10.05) -- (5.8,10.05);

		\draw[yellow, line width=1pt] (-2.4,4) circle (0.3cm);
		\node[anchor=center] at (-2.7, 4.6) {$-1$};
		\draw[yellow, line width=1pt] (4.45,4) circle (0.3cm);
		\node[anchor=center] at (4.15, 4.6) {$-1$}; 
		\draw[yellow, line width=1pt] (5.5,7.3) circle (0.3cm);
		\node[anchor=center] at (5.8, 7.9) {$+1$};
		\draw[yellow, line width=1pt] (-2.4,10.) circle (0.3cm);
		\node[anchor=center] at (-2.7, 10.6) {$-1$};
		\draw[yellow, line width=1pt] (4.45,10) circle (0.3cm);
		\node[anchor=center] at (4.4, 10.6) {$-1+1$}; 

		\fill[white] (-0.5,1.6) rectangle (0.5,1.9) ;
		\fill[white] (-6.7,6) rectangle (-6.4,9) ;
		\fill[white] (6.75,6) rectangle (7,9) ;
		
		\node[scale=0.7] at (0,1.6) {$\tilde\xi$};
		\node[rotate=90,scale=0.7] at (-6.7,7.5) {energy};
		\node[rotate=90,scale=0.7] at (6.85,7.5) {IDS};
	\end{tikzpicture} 
	\caption{Zoom on the gaps C, D and E of figure \ref{secondary_gapsWN}. The crossings of the eigenvalues with the fiducial energy lines are encircled. 
	}
	\label{secondary_gapsWN_zoom}  
\end{figure}
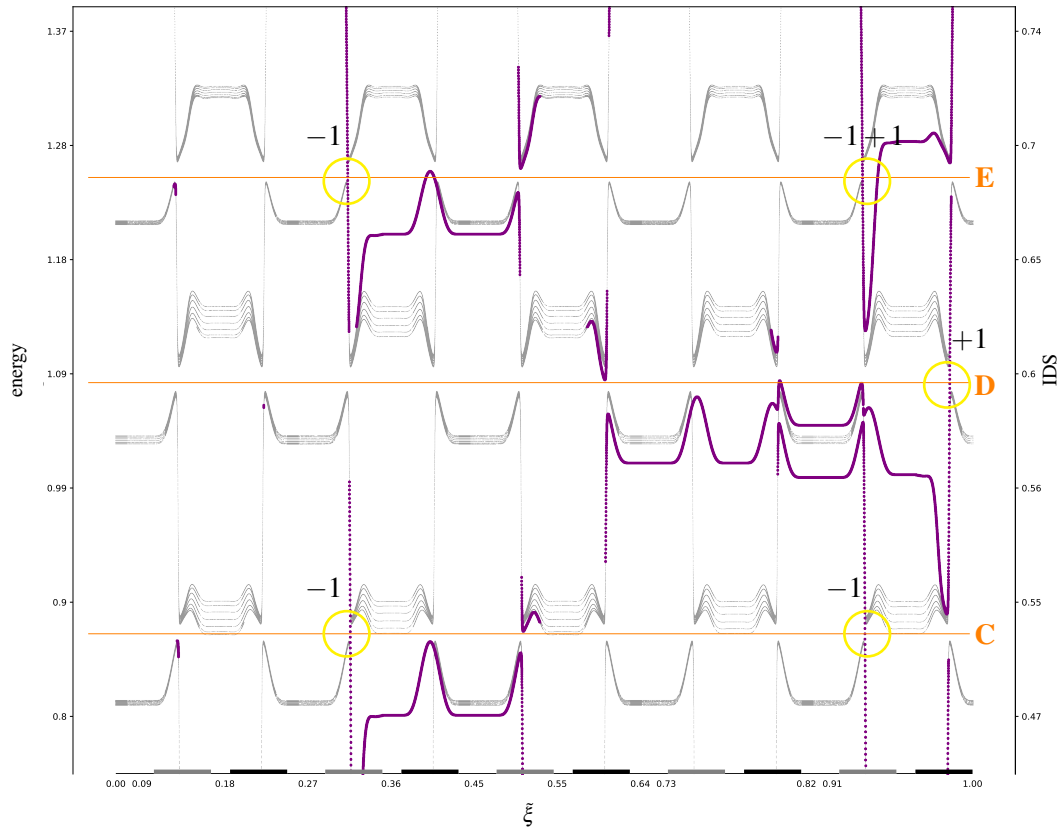

\section{conclusion}
A formalism named augmentation was introduced to overcome the difficulty that the boundary invariants in the standard bulk edge correspondence are trivial for one dimensional chains with totally disconnected hull. Through such an augmentation we obtained a spectral flow interpretation of the values of the Integrated density of states which also has a physical interpretation related to forces (and hence to pumping). Different augmentations lead to different spectral flows which are related. In this way we obtained a relation between the force induced by moving the boundary and that induced by phason flips. In the 2-cut Kohmoto model the spectral flow of the edge states becomes only visible if we consider a stacked two layer system which corresponds to a relative $K_0$-class, that is, an element that arrises through the Grothendieck construction. This illuminates that even the abstract Grothendieck construction manifests itself in a physical application to topological insulators.     
\bigskip

\noindent
{\bf Acknowledgement:} Most of the  material is part of the second author’s doctoral thesis \cite{ScaglioneThesis}.

\noindent
{\bf Conflict of interest:} All authors declare that they have no conflicts of interest.

\noindent
{\bf Data availability statement:} No new data were created or analysed during this study. Data sharing is not applicable to this article.

\nocite{*}
\bibliographystyle{amsplain}
\bibliography{bibliography}

\end{document}